\documentclass[12pt,a4paper]{JHEP3}
\usepackage{amscd,amsmath,amssymb,amsfonts,xspace,mathrsfs,amsthm}
\usepackage{color}
\let\normalcolor\relax
\usepackage{bbold}
\usepackage{latexsym}
\usepackage{graphicx}

\numberwithin{equation}{section}

  \newtheorem{proposition}{Proposition}
  \newtheorem{lemma}{Lemma}
  \newtheorem{conj}{Conjecture}

\def\lfig#1#2#3#4#5{
\begin{figure}[t]
 \centerline{\includegraphics[width=#3]{#2}}
 \vspace{#5}
  \caption{#1 \label{#4}}
 \end{figure}
}

\def\({\left(}
\def\){\right)}
\def\[{\left[}
\def\]{\right]}
\def\<{\left\langle}
\def\>{\right\rangle}
\def\hf{{1\over 2}}

\newcommand{\unit}{{\mathbb{1}}}

\newcommand{\de}{{\rm d}}
\newcommand{\I}{{\rm i}}
\renewcommand{\Im}{{\rm Im}}
\renewcommand{\Re}{{\rm Re}}
\newcommand{\sgn}{\mbox{sgn}}

\def\sign{{\rm sgn}}

\newcommand{\cA}{\mathcal{A}}

\newcommand{\cD}{\mathcal{D}}

\newcommand{\cT}{\mathcal{T}}
\newcommand{\cW}{\mathcal{W}}
\newcommand{\cX}{\mathcal{X}}

\newcommand{\cN}{\mathcal{N}}

\newcommand{\cM}{\mathcal{M}}

\newcommand{\IZ}{\mathbb{Z}}
\newcommand{\IR}{\mathbb{R}}

\newcommand{\IP}{\mathbb{P}}
\newcommand{\nn}{\nonumber}

\newcommand{\eps}{\epsilon}

\newcommand{\ha}{{\wh{A}}}

\def\ba{\bar a}

\def\bZ{\bar Z}

\newcommand{\be}{\begin{equation}}
\newcommand{\ee}{\end{equation}}
\newcommand{\ben}{\begin{eqnarray}\displaystyle}
\newcommand{\een}{\end{eqnarray}}
\newcommand{\half}{{\tfrac{1}{2}}}

\newcommand\bOm{\overline{\Omega}}

\definecolor{varcolor}{rgb}{0.1,0.55,0.25}
\definecolor{functioncolor}{rgb}{0.1,0.35,0.75}
\definecolor{paper_blue}{rgb}{0.3,0.2,0.75}
\definecolor{paper_red}{rgb}{0.65,0.1,0.15}
\definecolor{paper_green}{rgb}{0.05,0.35,0.125}
\definecolor{paper_grey}{gray}{0.375}
\definecolor{perm}{rgb}{0.1,0.45,0.85}
\definecolor{deemph}{rgb}{0.7,0.7,0.7}

\newcommand{\ta}{\tilde\alpha}

\renewcommand{\ha}{\hat\alpha}

\newcommand{\gref}{g_{\rm C}}
\newcommand{\gtr}{g_{\rm tr}}
\newcommand{\Ftr}[1]{F_{{\rm tr},#1}}

\newcommand{\OmS}{\Omega_{\rm S}}

\def\am#1{a^{(#1)}}
\def\bm#1{b^{(#1)}}
\def\cm#1{c^{(#1)}}
\def\dm#1{d^{(#1)}}
\def\tdm#1{\tilde d^{(#1)}}
\def\tam#1{\tilde a^{(#1)}}
\def\tbm#1{\tilde b^{(#1)}}
\def\ta{\tilde a}

\def\Cv{\mathscr{C}}
\def\Av{\mathscr{A}}

\def\Zv{\mathscr{Z}}

\def\Vv{\mathscr{V}}
\def\Uv{\mathscr{U}}

\def\tCv{\tilde\Cv}
\def\Avl{\Av^{(\ell)}}
\def\Cvl{\Cv^{(\ell)}}
\def\Zvl{\Zv^{(\ell)}}

\def\Cvi#1{\Cv^{(#1)}}
\def\Zvi#1{\Zv^{(#1)}}
\def\Vvi#1{\Vv^{(#1)}}
\def\Uvi#1{\Uv^{(#1)}}
\def\gl{g_{(\ell)}}
\def\gi#1{g_{(#1)}}

\def\Cvi#1{\Cv^{(#1)}}

\def\bea{\begin{eqnarray}}
\def\eea{\end{eqnarray}}
\def\be{\begin{equation}}
\def\ee{\end{equation}}
\def\ba{\begin{align}}
\def\ea{\end{align}}
\def\bse{\begin{subequations}}
\def\ese{\end{subequations}}

\newcommand{\disc}{{\rm disc}}

\def\Lv#1{L(#1)}
\def\Rv#1{R(#1)}
\def\FC{F_{C,n}}
\def\FCn#1{F_{C,#1}}
\def\Fs{F^{(\star)}}
\def\Fwl{F^{(0)}}
\def\Fsp{F^{\rm (+)}}
\def\tFwl{\tilde F^{(0)}}

\def\tFs{\tilde F^{(\star)}}

\def\Sym{\,{\rm Sym}\, }

\def\mm{m}
\def\cs{S}

\def\cMad{\cM^{\rm ad}}

\title{Attractor flow trees, BPS indices  and quivers}

\preprint{L2C:18-057\\arXiv:1804.06928v2}

\author{Sergei Alexandrov$^{1,2}$ and Boris Pioline$^{3}$
\\
$^1$ {\it
Laboratoire Charles Coulomb (L2C), Universit\'e de Montpellier,
CNRS, F-34095, Montpellier, France}\\
$^2$ {\it
Department of High Energy and Elementary Particle Physics,
Saint Petersburg State University,
7/9 Universitetskaya nab., St. Petersburg 199034, Russia}

$^3$ {\it Laboratoire de Physique Th\'eorique et Hautes
Energies (LPTHE), UMR 7589 CNRS-Sorbonne Universit\'e,
Campus Pierre et Marie Curie,
4 place Jussieu, F-75005 Paris, France} \\

\vspace*{2mm} {\tt e-mail:
\email{sergey.alexandrov@umontpellier.fr},
\email{pioline@lpthe.jussieu.fr}
}

\vspace*{-3mm}

}

\abstract{Inspired by the split attractor flow conjecture for multi-centered black hole solutions
in $\cN=2$ supergravity, we propose a formula expressing the
BPS index $\Omega(\gamma,z)$ in terms of `attractor indices'
$\Omega_*(\gamma_i)$. The latter count BPS states in their respective attractor chamber.
This formula expresses the index as a sum over stable flow trees weighted by products of attractor indices.
We show how to compute the contribution of each tree directly in terms of asymptotic data,
without having to integrate the attractor flow explicitly.
Furthermore, we derive new representations for the index which make it manifest
that discontinuities associated to distinct trees cancel in the sum, leaving only the discontinuities
consistent with wall-crossing.
We apply these results in the context of quiver quantum mechanics, providing
a new way of computing the Betti numbers of quiver moduli spaces,
and compare them  with the Coulomb branch formula,
clarifying the relation between attractor and single-centered indices.
}

\begin{document}

\maketitle

\section{Introduction and summary}

In four-dimensional supersymmetric field theories and string vacua with $\cN=2$ supersymmetry,
the spectrum of BPS states depends sensitively on the moduli and marginal deformations.
As the value $z$ of the moduli fields at spatial infinity are varied, some bound states may form or decay,
leading to a jump in the BPS index (or helicity supertrace) $\Omega(\gamma,z)$ counting BPS states
of electromagnetic charge $\gamma$ with signs. Such decays are only possible across walls of marginal stability
--- codimension one hypersurfaces in the moduli space associated  to pairs of
charges with non-vanishing Dirac-Schwinger-Zwanziger (DSZ) pairing
$\langle\gamma_L,\gamma_R\rangle\ne 0$. The corresponding wall is the locus where
the phases of $Z_{\gamma_L}(z)$ and $Z_{\gamma_R}(z)$ align. Here
$Z_\gamma(z)$ is the central charge, a complex-valued linear function of the charges whose phase determines
the supersymmetry preserved by a BPS state of charge $\gamma$. On such a wall
the mass $|Z_\gamma(z)|$ of a BPS state with charge $\gamma$ in the positive cone
spanned by $\gamma_L,\gamma_R$ coincides with the total mass $\sum_i |Z_{\gamma_i}(z)|$
of any set
of BPS states with charges $\gamma_i$ in the same cone
with $\sum_i \gamma_i=\gamma$, allowing the formation of threshold bound
states (see e.g. \cite{Pioline:2011gf} and references therein).

The jump of the BPS index $\Omega(\gamma,z)$ across the wall of marginal stability is governed
by a universal wall-crossing formula, first formulated in the mathematics literature by
Kontsevich--Soibelman \cite{ks} and Joyce--Song \cite{Joyce:2008pc,Joyce:2009xv},
and then established by physical reasoning in a series of papers \cite{Denef:2007vg,Gaiotto:2008cd,Andriyash:2010qv,Manschot:2010qz}.
There are also refined versions of
the index $\Omega(\gamma,z)$ and wall-crossing formula, which keep track of the spin
$\vec J$ and R-charge $\vec I$ of BPS states via fugacity parameters $y,t$
conjugate to the projections $J_3$ and $I_3$, respectively \cite{Diaconescu:2007bf,Dimofte:2009bv,Gaiotto:2010be}.
An important question for various applications, including the study of duality constraints on BPS indices,
is to express the moduli-dependent BPS index
$\Omega(\gamma,z)$ in terms of some physically motivated indices which depend only on the charges,
and possibly on the chemical potentials $y,t$, but are independent of the moduli $z$.
In this work, we investigate two different ways of answering to this question, which are
both motivated by the physics of BPS black holes in $\cN=2$ supergravity.

As shown in \cite{Denef:2000nb,Bates:2003vx}, $\cN=2$ supergravity admits a class of
stationary supersymmetric solutions obtained by superimposing $n$ BPS black holes with charges $\gamma_1,\dots, \gamma_n$,
subject to moduli and charge-dependent conditions on the
distances between the centers. In the vicinity of each center, the solution reduces to the usual
spherically symmetric BPS black hole with charge $\gamma_i$, in particular the moduli are
attracted to a fixed value $z_{\gamma_i}$ independently of their value at
spatial infinity \cite{Ferrara:1995ih}.
Such solutions typically exist only
within a certain chamber of moduli space, whose boundary precisely
consists of walls of marginal stability. Near one of the walls,
the distance between some of subsets of the constituents grows and becomes infinite exactly on the wall,
providing a clear physical picture of the wall-crossing phenomenon \cite{Denef:2007vg,Manschot:2010qz}
as well as the substructure of BPS bound states.

In general however, some connected components of the space of multi-centered solutions
are ruled out by the constraint that the metric should admit no closed time-like curves. While
this property is cumbersome to check explicitly, a simple criterium
has been put forward, called the split attractor flow conjecture \cite{Denef:2001xn,Denef:2007vg},
which in principle determines the allowed multi-centered solutions for a given value of the moduli at spatial infinity.
The idea is to model each solution by a nested sequence of two-centered bound states,
represented by a binary rooted tree. The vertices of the tree are decorated by the corresponding electromagnetic charge,
from the total charge $\gamma$ at the root to the constituent charges $\gamma_i$ at the leaves of the tree (see Fig. \ref{fig-AFtree}).
Along each edge of the graph, associated to  the vertex $\gamma\to \gamma_L+\gamma_R$, the moduli $z$
vary according to the usual attractor flow for a spherically symmetric black hole,
until they cross
the locus where  $\Im[Z_{\gamma_L}\bar Z_{\gamma_R}]$ vanishes (see Fig. \ref{fig-flow}).
If the central charges $Z_{\gamma_L}, Z_{\gamma_R}$ at this point are aligned (as opposed
to being anti-aligned), and if the
stability condition
$\langle \gamma_L,\gamma_R\rangle\, \Im[Z_{\gamma_L}\bar Z_{\gamma_R}]>0$
is obeyed before reaching the locus where $\Im[Z_{\gamma_L}\bar Z_{\gamma_R}]=0$,
the flow is repeated recursively for each of the two constituents, otherwise,
the tree is discarded. When no further splittings are allowed, the flow on each branch
terminates at the attractor point $z_{\gamma_i}$ for the corresponding constituent.
The split attractor flow conjecture stipulates that the space of
admissible multi-centered solutions is partitioned by  stable attractor flow trees
\cite{Denef:2001xn,Denef:2007vg,Andriyash:2010yf}.

Based on this picture, and building on earlier proposals in the literature
\cite{Denef:2001xn,Denef:2007vg,Manschot:2010xp,Manschot:2010qz,Manschot:2011xc},
we propose that the total BPS index $\Omega(\gamma,z)$ -- or rather
its variant $\bOm(\gamma,z)$ defined in \eqref{defbOm} so as to properly take
into account Bose-Fermi statistics \cite{Manschot:2010qz} -- can be expressed
as a sum over all stable attractor flow trees rooted at $(\gamma,z)$. The contribution of
each tree is proportional to the product of attractor indices
$\bOm_*(\gamma_i)\equiv \Omega(\gamma_i,z_{\gamma_i})$ associated to each center $\gamma_i$,
and to the DSZ pairings $\pm \langle \gamma_L,\gamma_R\rangle$
for all the vertices, corresponding to the BPS indices of the nested two-particle bound states.
The resulting `flow tree formula' in principle allows to reconstruct the BPS index
$\Omega(\gamma,z)$ in any chamber of moduli space, in terms of the
moduli-independent attractor indices $\Omega_*(\gamma_i)$.
However, this procedure raises various problems, both at the practical and conceptual levels,
which we outline below and address in this paper.

The first, practical problem is that this procedure seems to require integrating
the full attractor flow equations along the edges of the tree to check the stability of
each two-centered bound state in the hierarchy. Due to this difficulty,
explicit study of split attractor flows has been mostly confined to one-modulus models in the
literature \cite{Denef:2001xn,Collinucci:2008ht}.  However, it was shown in \cite{Manschot:2010xp}
that for three centers, the stability conditions could in fact be expressed in terms of
the moduli at infinity without explicitly solving the attractor flow along the edges, and it
was suggested that the same could be done for an arbitrary number of centers.
In this work, we  demonstrate that indeed,
for the purposes of checking stability, the continuous attractor flow along the edges can be reduced
to a `discrete attractor flow' from one vertex to its descendants. It is important to stress however
that this reduction assumes that the flow tree exists, in particular that the phases of the central
charges are aligned at each vertex and that the flow does not reach a singularity along the edges.
We expect that this assumption is obeyed in regions of moduli space where the
central charges of the constituents are nearly aligned, for example for D4-D2-D0 black holes
in type II strings on a Calabi-Yau threefold in the large volume limit.
Granting this assumption, the enumeration of  stable flow trees becomes considerably easier
and can be efficiently implemented
on a computer.\footnote{A new version of the mathematica package
{\tt CoulombHiggs.m} originally released along with \cite{Manschot:2013sya} and including
an implementation of the `flow tree formula' for quivers is available from the home page of
the second author.}

The second, more conceptual problem is that
the contribution of each flow tree has additional discontinuities
beyond those predicted by the wall-crossing formula, originating from violations of  the stability
conditions on the intermediate bound states.
For instance, for a three-centered solution the tree $((12)3)$ corresponding to the decay sequence
$\gamma\to (\gamma_1+\gamma_2)+\gamma_3 \to \gamma_1+\gamma_2+\gamma_3$
jumps across the codimension-one locus where the inner bound state $(\gamma_1+\gamma_2)$
becomes only
marginally stable. This however is {\it not} a wall of marginal stability
for the bound state with total charge $\gamma_1+\gamma_2+\gamma_3$
(unless $\gamma_3$ is collinear with $\gamma_1$ or $\gamma_2$). Fortunately,
the trees $((12)3)$, $((23)1)$ and $((31)2)$ all have discontinuities at the same locus
(see Fig. \ref{fig-threeplot} for two representative
3-center configurations), and these discontinuities cancel in
the sum \cite{Manschot:2010xp,Andriyash:2010yf,Chowdhury:2012jq}.
Thus, the total index $\Omega(\gamma,z)$ is constant
across the wall
although the internal structure of the bound state does change. We shall
refer to these loci as `fake walls', although the name `recombination walls' is perhaps more appropriate.
In this work, we prove that (under the same assumption as above) the discontinuities of individual flow trees across fake walls
always cancel for any number of centers, leaving only those required by wall-crossing.

In fact, we shall provide two alternative proofs of this result. The first, relegated to Appendix \ref{ap-vectors},
makes use of certain `flow vectors', which encode
stability conditions along flow trees. After establishing various symmetry properties for
these vectors and their mutual orthogonal projections, one can give an elementary proof
that discontinuities across fake walls indeed cancel. The second proof follows from a
new representation of the sum over all flow trees, which is manifestly
continuous across the fake walls. This new representation is obtained by decomposing
the refined tree index $\gtr(\{\gamma_i\},z,y)$ (corresponding to the sum over all flow trees
rooted at charge $\gamma$ and ending on the constituents of charges $\gamma_i$) in terms of
a sum \eqref{gtF} of `partial tree indices' $\Ftr{n}(\{\gamma_i\},z)$, which no longer depend on
the parameter $y$ conjugate to $J_3$.
The partial tree index is obtained by summing over planar trees only, and is shown to satisfy
the recursive equation \eqref{F-ansatz}. The latter makes it manifest that the only singularities
of the refined tree index $\gtr$ (and therefore also of its limit as $y\to 1$)
are those predicted by the wall-crossing formula.  In \eqref{gtreen} we further conjecture
a formula which makes it clear that $\gtr$ is a symmetric Laurent polynomial
in $y$, a fact which is obscured by the decomposition \eqref{gtF}, and renders
the specialization to $y=1$  trivial.
The equations \eqref{F-ansatz} and \eqref{gtreen} can be seen as the main technical
results of this paper.

Having cleared these issues, we then apply the `flow tree formula' in the context
of quiver quantum mechanics, which describes the interactions of mutually non-local
dyons in four-dimensional field theories with 8 supercharges \cite{Denef:2002ru},
at least in regions where the central charges of the constituents
are nearly aligned \cite{Denef:2007vg}.
This framework has many advantages compared to the general case.
Firstly, the charge vector $\gamma$
is a $K$-tuple of non-negative integers $(N_1,\dots, N_K)$, corresponding to the ranks in the
product gauge group, so that the enumeration of
all possible splittings $\gamma=\sum_i \gamma_i$ is straightforward. Secondly,
the stability conditions depend only on the supergravity moduli $z$ through the
Fayet-Iliopoulos parameters $\zeta=(\zeta_1,\dots, \zeta_K)$, on which the discrete attractor flow
naturally acts. Moreover, the central charges of the constituents can never anti-align (within the validity
of the quiver quantum mechanics) so the use of the discrete attractor flow
is  justified. Third, the BPS index is in principle computable by localization for any  charge vector
$\gamma$ and Fayet-Iliopoulos parameters $\zeta$ \cite{Hori:2014tda}
(see also \cite{Cordova:2014oxa,Hwang:2014uwa}).
Finally, the BPS index in this context has a mathematically rigorous definition as the
Poincar\'e polynomial of the moduli space of stable quiver
representations \cite{Denef:2002ru}, which is a central object  in algebraic geometry
and representation theory. After introducing a suitable notion of `attractor point' \eqref{attFI}
(already introduced in \cite{MPSunpublished}), we formulate
the flow tree formula in this context in a mathematically precise and self-contained way,
and outline a proof in the case of quivers without oriented loops.

Finally, we investigate the relation between the flow tree formula and the Coulomb branch
formula developed in \cite{Manschot:2011xc,Manschot:2012rx,Manschot:2013sya}
(see \cite{Manschot:2014fua} for a concise review). While both of them express
the BPS index $\Omega(\gamma,z)$ in terms of moduli-independent indices, the former
relies on attractor indices $\Omega_*(\gamma_i)$, whereas the latter relies
on the concept of `single-centered  indices' $\Omega_S(\gamma_i)$. The difference
between these two indices  is due to the  existence of so-called scaling solutions, i.e. multi-centered
solutions with $n\geq 3$ constituents which can become arbitrarily close to each other
and remain allowed in the attractor chamber \cite{Denef:2007vg,Bena:2012hf}.
Single-centered indices $\Omega_S(\gamma_i)$
are designed to isolate the contributions of of single-centered black hole micro-states,
for which the holographic correspondence is supposed to apply \cite{Sen:2008yk}.
Moreover, the fact that single-centered BPS black holes can only carry zero angular momentum 
\cite{Sen:2009vz,Dabholkar:2010rm}, strongly constrains their dependence
on the refinement parameter $y$ \cite{Lee:2012sc,Lee:2012naa,Manschot:2012rx}. 
The main drawback of single-centered indices however
is that they are defined recursively in a rather complicated way and their mathematical
significance is unclear (unlike attractor indices which are special instances of the usual BPS index 
in the attractor chamber).

In Section \ref{sec-gC}, we show that for charge configurations $\{\gamma_i\}$ described
by quivers without loops, the tree index $\gtr(\{\gamma_i\},z,y)$ entering the
flow tree formula and the Coulomb index $\gref(\{\gamma_i\},z,y)$ entering the Coulomb branch formula coincide, so that
the attractor and  single-centered indices coincide as well. In contrast,
for charge configurations  described by quivers with oriented loops,
the two in general differ due to the contributions of scaling configurations. We compute
the difference between $\gtr$ and $\gref$ for $n=3$
and $n=4$ centers, and deduce the relation between the attractor and single-centered indices
in the case where $\gamma$ decomposes into a sum of at most 4 distinct constituents.

The results obtained in this work will be useful in making further progress
in understanding modularity constraints on the counting of D4-D2-D0 brane bound states
in type II strings compactified on a Calabi-Yau threefold $\cX$. As first suggested  in
\cite{Maldacena:1997de} and confirmed in a series of more recent works (see e.g.
\cite{Alexandrov:2016tnf} and references therein), when the D4-brane wraps an irreducible
divisor $\cD$ inside $\cX$, the generating function of the BPS indices of D4-D2-D0 black
holes with fixed D4-brane charge is  a weakly
holomorphic Jacobi form, given by the elliptic genus of the superconformal
field theory describing an M5-brane wrapped on the same divisor. In the case where the divisor $\cD$ is reducible however , the BPS
indices are subject to wall-crossing and, since the duality group
acts both on the charge $\gamma$ and the moduli, the constraints from modularity
take a more subtle form. In order to uncover these
constraints,  it is useful to express
the total index $\Omega(\gamma,z)$ in terms of suitably chosen moduli-independent indices.
In the present context, it was argued in  \cite{Manschot:2009ia,Alexandrov:2012au} 
that a natural choice is the BPS index in the `large volume 
attractor chamber', or  `MSW invariant' in the terminology of \cite{Alexandrov:2012au}, 
which is expected to count BPS states in the MSW superconformal field theory \cite{Andriyash:2008it,deBoer:2008fk}.
Using the formalism developed in the present work, the BPS index for any modulus $z$
(still near the large volume limit) can be expressed  as a sum over flow trees  weighted by
MSW invariants. Since the latter are invariant under spectral flow,
the resulting partition function can still be formally decomposed as a sum of
indefinite theta series, with a  kernel given by a sum over flow trees \cite{Manschot:2009ia,Manschot:2010xp}.
While the convergence and modular properties of  this theta series are by now well understood
when the reducible divisor $\cD$ is the sum of two components  \cite{Manschot:2009ia,Alexandrov:2016tnf}, the representations of the tree index found in this paper
will be instrumental in extending these results to a general reducible divisor  \cite{ap-to-appear}.

The organization  of this paper is as follows.
In section \ref{sec-atrflow} we review the split attractor flow conjecture for multi-centered BPS
solutions in $\cN=2$ supergravity, and motivate the `flow tree formula', which expresses
the index $\Omega(\gamma,z)$ in terms of the attractor indices $\Omega_*(\gamma_i)$.
We then express the sum over stable flow trees $\gtr(\{\gamma_i\},z,y)$ purely in terms of asymptotic data,
and find new representations of the `tree index' which make it manifest that it is smooth
away from the walls of marginal stability,
despite the fact that individual flow trees also have discontinuities on the `fake walls'.
In section \ref{sec-quiver} we state a mathematically precise form of
the flow tree formula in the context of quiver quantum mechanics, where
it computes the Poincar\'e polynomial of quiver moduli spaces.
In section \ref{sec-Coulomb} we  compare the flow tree formula to the Coulomb branch formula,
and clarify the relation between attractor indices $\Omega_*(\gamma_i)$ and single-centered indices
$\Omega_S(\gamma_i)$. Some technical details are relegated to appendices, including
several identities between sign functions widely used in the main text (appendix \ref{ap-signs}),
an alternative proof of the cancellation of fake discontinuities using flow vectors
and their symmetry properties (appendix \ref{ap-vectors}), and explicit expressions
for the partial tree and Coulomb indices for $n\leq 4$ centers (appendix \ref{ap-indFF}).

\section{Attractor flow trees and BPS index}
\label{sec-atrflow}

In this section we recall basic facts about multi-centered BPS solutions in $\cN=2$ supergravity
and classification of admissible solutions via the  split attractor flow conjecture. After making
general comments about their quantization, we propose the `flow tree formula'
expressing the BPS index $\Omega(\gamma,z)$ as a sum of monomials in the attractor indices
$\Omega_*(\gamma_i)$  associated to the constituents in all possible decompositions $\gamma\to\sum \gamma_i$.
The coefficient of each monomial is a sum over stable flow trees
which we call the `tree index'. After illustrating this formula in the case of three centers, we
explain how to express the contribution of each tree in terms of the moduli $z$ at infinity,
without having to solve the attractor flow explicitly. We then establish a new representation
for the tree index, which makes it manifest that it is continuous away from walls of marginal stability,
despite the fact that contributions from individual trees typically jump across other loci.

\subsection{Multi-centered black holes in $\cN=2$ supergravity}

We start by recalling well-known facts about multi-centered  BPS solutions in ungauged $\cN=2$ supergravity
in four dimensions. A  general class of stationary BPS solutions
with total electromagnetic charge $\gamma=(p^\Lambda,q_\Lambda)$ ($\Lambda=1,\dots,r$)
can be written
as \cite{Denef:2000nb,Denef:2001xn,Bates:2003vx} (in units where the Newton constant
$G_4=1$)
\bea
\de s^2 &=& - e^{2U}(\de t+\omega)^2 +e^{-2U}\, \de\vec x^2 ,
\\
A &=& 2\, e^U\, \Re(e^{-\I\alpha}Z)\,  (\de t+\omega) + A_d,
\eea
where $A=(A^\Lambda, A_\Lambda)$ denotes the symplectic vector of the electric and magnetic gauge fields.
The scale factor $U$ and  one-forms $\omega$ and $A_d$ are determined by the equations
\be
2\, e^{-U} \Im(e^{-\I\alpha}Z )= - H \, ,
\qquad
\star\,\de\omega=\langle \de H,H\rangle\, ,
\qquad
\de A_d = \star\, \de H,
\label{addeq-BPSsol}
\ee
where $H=(H^\Lambda,H_\Lambda)$
is vector-valued harmonic function on $\IR^3$. Here $Z(z)=(X^\Lambda(z),F_\Lambda(z))$
is the usual holomorphic symplectic vector determined by the holomorphic prepotential $F(X)$,
such that
\be
Z_\gamma=\langle \gamma,Z\rangle=q_\Lambda X^\Lambda-p^\Lambda F_\Lambda
\ee
is the central charge of the unbroken supersymmetry algebra. Moreover
$\alpha$ is the (position-dependent)
phase of $Z_\gamma$.  The first equation in \eqref{addeq-BPSsol}  also determines
the spatial profile of the scalar fields $z$, valued in the vector multiplet moduli
space $\cM_{\rm SK}$ (a special K\"ahler manifold determined by the prepotential $F(X)$).

A multi-centered BPS solution is obtained by choosing a specific harmonic vector on $\IR^3$, namely
\be
\label{Hmc}
H = \sum_{i=1}^{n} \frac{\gamma_i}{|\vec x-\vec x_i|} -2 \, \Im(e^{-\I\alpha} Z)_{\infty}\, ,
\ee
where the second term is a constant  symplectic vector determined by the total charge $\gamma$ and
the values of the moduli at spatial infinity.
This choice ensures that the solution is asymptotically flat, i.e. $U\to 0$ and $\omega\to 0$.
The integrability condition $\de(\de\omega)=0$ constrains the positions of the centers to satisfy, for all $i$,
\be
\label{DenefEq}
\sum_{j\neq i} \frac{\gamma_{ij}}{|\vec x_i-\vec x_j|}
= 2\, \Im(e^{-\I\alpha} Z_{\gamma_i})_{\infty}\, ,
\ee
where  $\gamma_{ij}$ is a shorthand notation for $\langle \gamma_i,\gamma_j\rangle$.
While the charge vectors $\gamma_i$ can in principle
be chosen at will in the charge lattice $\Lambda$,  we shall restrict attention to the case where
they all lie in the same positive cone $\Lambda_+$, defined as the set of all vectors $\gamma$ whose
central charge $Z_\gamma(z_\infty)$ lies in a fixed half-space defining the splitting between
BPS particles ($\gamma\in\Lambda_+$) and anti-BPS particles ($\gamma\in - \Lambda_+$).

For a given choice of charges $\gamma_1,\dots, \gamma_n$ and moduli $z_\infty$, the equations
\eqref{DenefEq} impose $n-1$ independent constraints
(indeed, the sum of these $n$ equations trivially vanishes)
on  $3n$ variables $\{\vec x_i\}_{i=1\dots n}$.
Ignoring the three translational degrees of freedom, this in general leaves a $2n-2$-dimensional
space of solutions $\cM_n(\{\gamma_1,\dots,\gamma_n\},z_\infty)$. The latter carries a canonical
symplectic form $\varpi$ and an Hamiltonian action of the rotation group $SO(3)$
generated by the total angular momentum \cite{deBoer:2008zn}
\be
\label{defJ}
\vec J = \frac12 \sum_{i<j} \gamma_{ij}\,  \frac{\vec x_i-\vec x_j}{|\vec x_i-\vec x_j|}\, .
\ee
For these solutions to be physical however, one must check in addition that they are smooth everywhere
(except possibly near the location of the centers), and have no closed timelike curves.
This typically rules out some  of the connected components in the phase space
$\cM_n(\{\gamma_i\},z_\infty)$, leaving only a (possibly empty) subset that we shall
denote by $\cMad_n(\{\gamma_i\},z_\infty)$.

For a single center, the solution is static, spherically symmetric and manifestly free of closed
timelike curves. The radial profile of the scalar fields $z(r)$ follows from the attractor flow equation
\be
\label{attflow}
e^{-U}\,  \frac{\Im[Z_{\gamma'} \bZ_\gamma(z(r))]}{|Z_\gamma(z(r))|} =
\frac{\langle\gamma,\gamma'\rangle}{2r} +
\left(  \frac{\Im(Z_{\gamma'} \bZ_\gamma)}{|Z_\gamma|} \right) _{r=\infty},
\ee
obtained by pairing the first equation in \eqref{addeq-BPSsol} with an arbitrary vector $\gamma'$.
We denote by $\cA(\gamma_,z_{\infty})$ the solution of
\eqref{attflow} with $z(r)=z_{\infty}$ at $r=\infty$,
and by $z_\gamma$ the value of the moduli at the attractor point $r\to 0$.
If $z_\gamma$ lies in the interior of $\cM_{\rm SK}$,
the solution interpolates between $\IR^{3,1}$ at spatial infinity and $AdS_2\times S_2$
near the horizon \cite{Ferrara:1995ih}. If instead
$z_\gamma$ lies on the boundary of $\cM_{\rm SK}$, it may still
be trusted outside some shell of radius $r$, inside which massless states can condense.
For future reference, we note that at the attractor point $z_\gamma$
it follows from \eqref{attflow} that
\be
\sign\,\Im[ Z_{\gamma'}\bZ_\gamma(z_\gamma)]=
\sign \langle\gamma,\gamma'\rangle
\label{attrZZ}
\ee
for all $\gamma'$.
From now on, we omit the dependence of $Z_\gamma(z)$ on $z$ whenever it is
evaluated at spatial infinity, i.e. $Z_\gamma\equiv Z_\gamma(z_{\infty})$.

\lfig{The profile of scalar fields around a two-center bound state $\gamma=\gamma_1+\gamma_2$
can be represented by a trajectory in $\cM_{\rm SK}$, which starts from $z_{\infty}$ and follows
the usual attractor flow for a single black hole of charge $\gamma$ until it reaches
the locus where the phases of $Z_{\gamma_1}$ and $\bZ_{\gamma_2}$ become aligned.
At this point $z_1$, the trajectory forks into two parts which follow the usual attractor flow
for charge $\gamma_1$ and $\gamma_2$, respectively, and converge to  the respective
attractor points $z_{\gamma_1}$ and $z_{\gamma_2}$. In reality,
the scalar fields $z(x)$ map $\IR^3$ to a 3-dimensional, amoeba-like domain in $\cM_{\rm SK}$ which
concentrates around this forked trajectory in the regions near $z_\infty$ and $z_{\gamma_i}$, while
the value $z_1$ is reached in the crossover region where $|\vec x-\vec x_i|$ is of the
order of the distance between the two centers.}{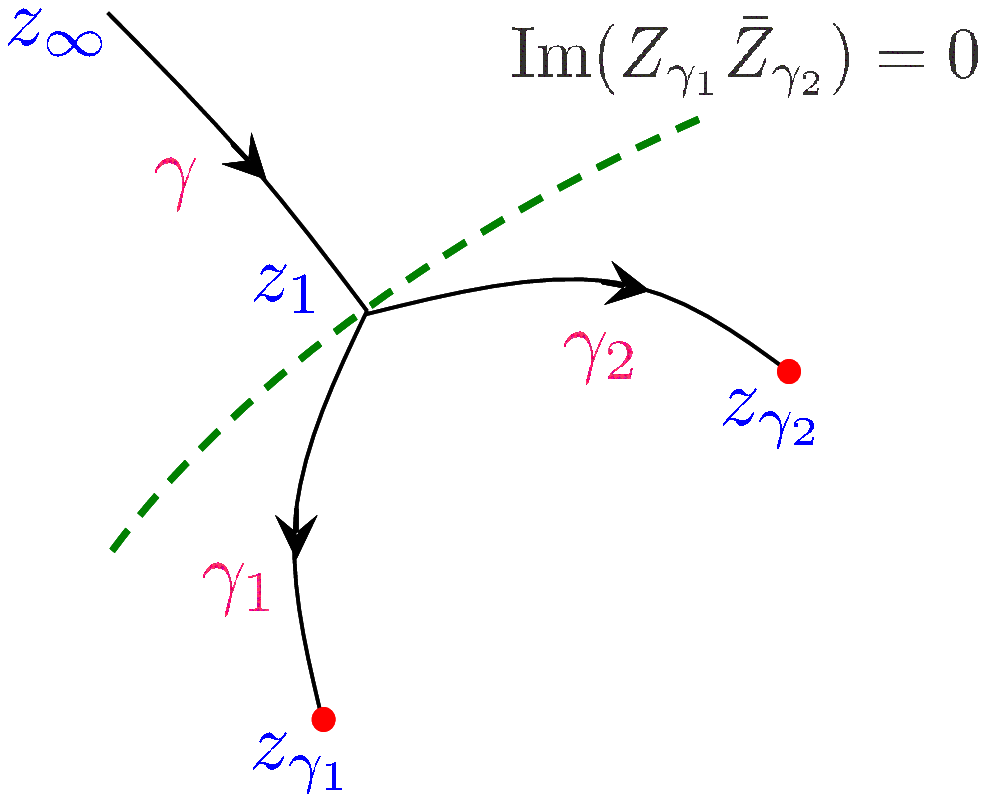}{7.1cm}{fig-flow}{-1.2cm}

For two centers, it is easily seen that the solution to \eqref{DenefEq} exists only if \cite{Denef:2000nb}
\be
\gamma_{12}\, \Im (Z_{\gamma_1} \bZ_{\gamma_2}) > 0 ,
\label{cond2center}
\ee
in which case the (inverse) distance between the two centers is  given by
\be
\frac{1}{2|\vec x_1-\vec x_2|} =
\frac{ \Im( Z_{\gamma_1} \bZ_{\gamma_2})}{\gamma_{12} \, |Z_{\gamma_1}+Z_{\gamma_2}|}\, .
\label{splitpoint}
\ee
If the condition \eqref{cond2center} is satisfied, the phase space $\cM_2$
is the sphere $S^2$ parametrizing the orientation of the dipole, otherwise it is empty.
In the former case, $\cM_2$ carries a  symplectic
form $\varpi=\frac12\gamma_{12}\varpi_{S^2}$ where $\varpi_{S^2}$ is the volume form.

While the condition \eqref{cond2center} is necessary, it is not sufficient.
Indeed, note that the distance $|\vec x_1-\vec x_2|$ in \eqref{splitpoint}
coincides with the radius $r=r_1$
where the attractor flow \eqref{attflow}
crosses the wall of marginal stability $\Im[ Z_{\gamma_1} \bar Z_{\gamma_2}(z_1)]=0$.
In order for the two-centered solution to be admissible,  the central
charges $Z_{\gamma_1}$ and $Z_{\gamma_2}$ should also satisfy
\be
\label{cond2centeralign}
\Re[Z_{\gamma_1} \bZ_{\gamma_2}(z_1)]>0 ,
\ee
in other words they should
have the same phase at this point (modulo $2\pi$),
as opposed to having opposite phases \cite{Denef:2000nb,Andriyash:2010yf}. If
the condition \eqref{cond2centeralign} is obeyed, then the admissible
phase space $\cMad_2$ coincides with $\cM_2$,
otherwise it is empty.

Outside the radius $r_1$,
the solution is well approximated by the spherically symmetric solution
with charge $\gamma=\gamma_1+\gamma_2$ and moduli $z_\infty$ at spatial infinity, while in the vicinity of each center,
it is  approximately given by a spherically symmetric solution with charge $\gamma_1$ or
$\gamma_2$, and moduli $z_1=z(r_1)$ far away (but not infinitely far) from the center.
As explained in Fig. \ref{fig-flow}, the behavior of the scalar fields around the two-centered
solutions can be approximated by a `split attractor flow' in the moduli space $\cM_{\rm SK}$, which
forks at the point $z_1$ where the phases of  $Z_{\gamma_1}$ and $\bZ_{\gamma_2}$
become aligned (in particular, $\Im[Z_{\gamma_1}\bar Z_{\gamma_2}(z_1)]$ vanishes).
Although we cannot compute $z_1$ explicitly, we can constrain the
central charges evaluated at this point.
Indeed, substituting \eqref{splitpoint} into \eqref{attflow}, one finds
that for $\gamma=\gamma_1+\gamma_2$ and any vector $\gamma'$,
\be
\label{attflowst}
e^{-U(r_1)}\,  \frac{\Im[Z_{\gamma'} \bZ_\gamma(z_1)]}{|Z_\gamma(z_1)|} =
\frac{\langle\gamma,\gamma'\rangle}{\langle \gamma_1,\gamma_2\rangle}\,
\frac{ \Im( Z_{\gamma_1} \bZ_{\gamma})}{|Z_{\gamma}|}
+ \frac{\Im(Z_{\gamma'} \bZ_{\gamma})}{|Z_{\gamma}|} \, .
\ee
This property will be key for expressing the stability of more general multi-centered solutions
in terms of asymptotic data,
as explained in the next subsections (specifically in \S\ref{subsec-flow}).

\subsection{Split attractor flow conjecture}
\label{subsec-afconj}

For more than two centers, determining the subset $\cMad_n\subset \cM_n$
of admissible solutions (i.e. corresponding to metrics without
closed time-like curves) is in general a difficult problem. In \cite{Denef:2001xn} (see also
\cite{Denef:2007vg,Andriyash:2010yf} for subsequent developments),
it was conjectured that $\cMad_n$  is partitioned into components\footnote{As
pointed out in \cite{Andriyash:2010yf}, the components are not necessarily disconnected,
but the main point is that the complement of $\cMad_n$ in $\cM_n$  is
not covered by any flow tree.}  labeled
by `split attractor flows', also known as `attractor flow trees' or simply `flow trees'.
The latter are obtained by iterating the basic splitting depicted in Fig. \ref{fig-flow},
and correspond to nested sequences
of two-centered bound states.

More precisely, these trees are (unordered, full)
rooted binary trees $T$ with endpoints (or leaves) labelled by $\gamma_1,\dots, \gamma_n$
and satisfying stability conditions at each vertex.\footnote{Ignoring stability conditions,
the number of unordered rooted binary trees with $n$ leaves is 
$b_n=(2n-3)!!=(2n-3)!/[2^{n-2}(n-2)!]=\{1,3,15,105,945,...\}$.}
To spell out these conditions, we first introduce some useful notations.
Let $V_T$ denotes the set of vertices of $T$ excluding the leaves.
Each vertex $v\in V_T$ has two descendants $\Lv{v}$, $\Rv{v}$ and parent $p(v)$ (see  Fig. \ref{fig-AFtree}).
Furthermore, to each vertex $v\in V_T$ we assign a charge $\gamma_v$ and a point $z_v\in
\cM_{\rm SK}$, both of them defined recursively. For the former, we start from
the leaves with charges $\gamma_i$ and assign charge $\gamma_v= \gamma_{\Lv{v}}+\gamma_{\Rv{v}}$
to the parent of two vertices $\Lv{v}$ and $\Rv{v}$. The root of the tree then
carries  charge $\gamma_{v_0}= \gamma\equiv\sum_{i=1}^n \gamma_i$.
For the latter, we instead start from the root of the tree, and assign $z_{p(v_0)}=z_\infty$,
the value of the moduli at spatial infinity, to its (fictitious) parent. We then follow each edge of
the graph downward from the root and, to a vertex $v$ with parent $p(v)$,
assign the value $z_v$ of the moduli where the attractor flow $\cA(\gamma_v,z_{p(v)})$ crosses
the wall of marginal stability for the bound state $\gamma_v \to \gamma_{\Lv{v}}+\gamma_{\Rv{v}}$
(i.e. $Z_{\gamma_{\Lv{v}}} \bar Z_{\gamma_{\Rv{v}}}(z_{v}) \in \IR^+$).
With these definitions, the admissible attractor flow trees are those which satisfy, for all $v\in V_T$,\footnote{The original
formulation of split attractor flows \cite{Denef:2007vg} also required that the attractor points $z_{\gamma_i}$
are regular in $\cM_{\rm SK}$. We shall
implement this condition by setting $\Omega_*(\gamma_i)=0$ if this is not the case.}
\be
\label{condtree}
\gamma_{\Lv{v}\Rv{v}}
\, \Im\bigl[ Z_{\gamma_{\Lv{v}}}\bZ_{\gamma_{\Rv{v}}}(z_{p(v)})\bigr]>0
\qquad
\mbox{and}\qquad
\Re\bigl[ Z_{\gamma_{\Lv{v}}}\bZ_{\gamma_{\Rv{v}}}(z_{v})\bigr]>0,
\ee
where $\gamma_{LR}=\langle \gamma_L,\gamma_R \rangle$.
We denote the set of all flow trees with $n$ leaves by $\cT_n(\{\gamma_i\},z)$, and the set of admissible trees
by $\cT^{\rm ad}_n(\{\gamma_i\},z)$. When all the charges $\gamma_i$ are distinct,
it is convenient to label trees by bracketings of the unordered set $\{1,\dots, n\}$,
e.g. the tree displayed in Fig. \ref{fig-AFtree} corresponds to $((13)(2(45)))$.

\lfig{An example of attractor flow tree corresponding to the bracketing $((13)(2(45)))$.}{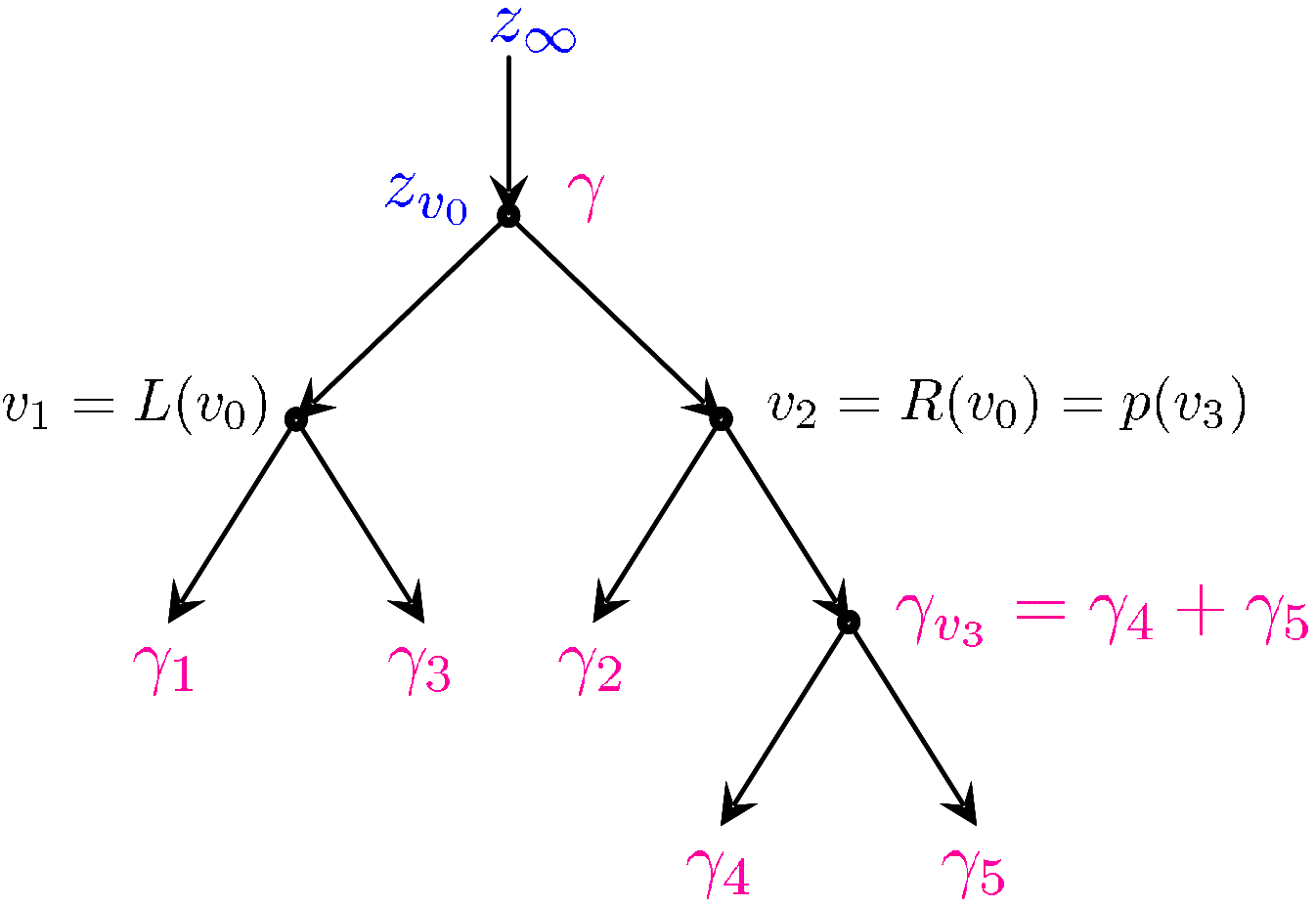}{9.5cm}{fig-AFtree}{-1.5cm}

As we shall see shortly, the split attractor flow conjecture is  not only useful  for classifying
admissible classical solutions, but also naturally suggests a formula for computing the
BPS index $\Omega(\gamma,z)$ in terms of the attractor indices $\Omega_*(\gamma_i)$
associated to the leaves of the tree.

\subsection{Quantizing multi-centered solutions}
\label{subsec-BPS}

The BPS index $\Omega(\gamma,z)$ counts all states with total charge $\gamma$
which exist in the discrete spectrum for a given value $z$ of the moduli at spatial infinity,
weighted by a sign $(-1)^{2J_3}$ corresponding to the parity of the total angular momentum in $\IR^3$
(after factoring out the center of motion degrees of freedom). In the presence of a R-symmetry,
it is also natural to consider the refined index $\Omega(\gamma,z,y,t)$ including fugacities
$y^{2J_3}$ and $t^{2I_3}$ for the angular momentum and the R-charge.\footnote{In string vacua,
the refined index is typically not protected away from $y=1$, but in $\cN=2$
field theories the value at $y=t$ is protected, and known as the protected spin
character \cite{Gaiotto:2010be}.} In the following we shall retain only the dependence on $y$,
but the parameter $t$ can be easily restored at any point, keeping in mind that the interactions
between the centers are insensitive to its value \cite{Manschot:2012rx}.

In the regime where the supergravity description is supposed to be valid\footnote{Namely,
when the  coupling governing the genus expansion in type II strings becomes strong. In this regime,
the horizon area of BPS black holes becomes much larger than the Planck length and
supergravity is in fact weakly coupled.},
the index $\Omega(\gamma,z)$ is expected to receive contributions from all consistent multi-centered solutions
with an arbitrary number of constituents $n$ of charges $\gamma_i\in\Lambda_+$ such that
$\gamma=\sum_{i=1}^n \gamma_i$, for a fixed value of the moduli at spatial infinity.
Assuming that the internal dynamics of the black holes decouples from their relative
motion,  it is natural to expect that for a given splitting, the contribution will be proportional
to the product of the BPS indices $\Omega_*(\gamma_i)=\Omega(\gamma_i,z_{\gamma_i})$
counting BPS states associated to each center, evaluated at the respective attractor
point since the moduli $z$ are attracted to their attractor value $z_{\gamma_i}$ in the
vicinity of each center. Moreover, its contribution
 should also include the BPS index $g(\{\gamma_i\},z,y)$
of the supersymmetric quantum mechanics describing the relative motion.
Since the BPS sector of this quantum mechanics is described
classically by the phase space $\(\cMad_n,\varpi\)$, it is also reasonable
to identify $g$ with the equivariant index\footnote{In the limit where the symplectic
form $\varpi$ is scaled to infinity, the latter reduces to the equivariant symplectic volume
$\int_{\cMad_n} e^{\varpi \log y}$, or to the ordinary symplectic volume when $y=1$.}
of the Dirac operator on $\(\cMad_n,\varpi\)$ \cite{Manschot:2011xc,Kim:2011sc}.

This simple picture however assumes that all of the charges $\gamma_i$ are distinct, so that
the centers are distinguishable. When some of the $\gamma_i$'s are equal, it is necessary to
take into account Bose-Fermi statistics, which requires to project on symmetric or antisymmetric
wave functions depending on the sign of $\Omega_*(\gamma_i)$.
In \cite{Manschot:2010qz}, it was shown that the simpler rules of Boltzmann statistics can be applied
provided one considers the rational invariant
\be
\label{defbOm}
\bOm(\gamma,z,y) = \sum_{m|\gamma} \frac{y-1/y}{m(y^m-1/y^m)} \Omega(\gamma/m, z,y^m) ,
\ee
where $m$ runs over all positive integers such that $\gamma/m$ is in the charge lattice.
If $\gamma$ is a primitive vector, $\bOm(\gamma,z)$ of course coincides with $\Omega(\gamma,z)$.
Thus, based on this physical reasoning we expect that the total index $\Omega(\gamma,z)$
can be written as
\be
\label{Omsumgen}
\bOm(\gamma,z,y) =
\sum_{\gamma=\sum_{i=1}^n \gamma_i}
\frac{g(\{\gamma_i\},z,y)}{|{\rm Aut}\{\gamma_i\}|}\,
\prod_{i=1}^n \bOm_*(\gamma_i,y),
\ee
where the sum over $\{\gamma_i\}$ runs over  unordered decompositions of $\gamma$
into sums of vectors $\gamma_i\in\Lambda_+$ (i.e. two decompositions differing only by the order of the $\gamma_i$'s
are considered identical).
The symmetry factor $|{\rm Aut}\{\gamma_i\}|$ is the cardinality of the stabilizer of the ordered $n$-tuple
$(\gamma_i)$ inside the permutation group $S_n$ for any fixed choice of ordering of the $\gamma_i$'s.
In order to  make this formula useful, it remains to find a practical way of computing
the BPS index $g(\{\gamma_i\},z,y)$ for the quantum mechanics of the configurational degrees of freedom.

In previous work \cite{Manschot:2010qz,Manschot:2011xc,Manschot:2013sya}, Manschot, Sen and the second author
applied localization techniques to evaluate the equivariant index of the Dirac operator on
$(\cM_n,\varpi)$. Namely, they used the fact that this phase space
admits an Hamiltonian action of the rotation group $SO(3)$ generated by the angular momentum
$\vec J$ in \eqref{defJ}, to reduce the problem to the enumeration of collinear black
hole configurations, i.e. one-dimensional solutions to \eqref{DenefEq}. The main difficulty with
this approach however is that for $n\geq 3$, the space $\cM_n$ is generically
non-compact (despite having finite volume),  due to the existence of scaling solutions,
where some subset of the centers can become arbitrarily close
\cite{Denef:2007vg,Bena:2012hf}.\footnote{For three centers, this can happen whenever the DSZ products
$\gamma_{12}$, $\gamma_{23}$, $\gamma_{31}$
are all of the same sign and satisfy the triangular inequalities:
$\gamma_{12}<\gamma_{23}+\gamma_{31}$ and its cyclic permutations \cite{Denef:2007vg,Bena:2012hf}.
More generally, scaling solutions can occur when the total angular
momentum \eqref{defJ} of a subset of charges vanishes \cite{Manschot:2011xc}.} As a consequence,
the index of the Dirac operator is not well-defined and
the naive result from localization is not a symmetric Laurent polynomial in $y$.
In \cite{Manschot:2011xc, Manschot:2012rx,Manschot:2013sya}, a prescription was proposed to repair this problem,
leading to a more intricate version of the \eqref{Omsumgen} known as the `Coulomb branch formula'
which we shall discuss in section \ref{sec-Coulomb}. Note however that this prescription does not
take into account the condition that collinear solutions should have no closed time-like curves.

In the next subsection, we shall propose a different way of computing the BPS index
$g(\{\gamma_i\},z,y)$ which instead relies on the split attractor flow conjecture, and which
is in principle free of these issues.

\subsection{The flow tree formula}
\label{subsec-afc}

Since, according to the split attractor conjecture, $\cMad_n$ is partitioned into
components labelled by stable attractor flow trees, it is natural to
propose\footnote{Various
precursors of the formula that we are about to state have appeared in the literature, including
\cite{Denef:2001xn,Denef:2007vg,Manschot:2010xp}. Our proposal is novel inasmuch
as quantum statistics is properly taken into account and the sum over attractor flow trees
has been reduced to a combinatorial problem.}
 that the index
$g(\{\gamma_i\},z,y)$ can be obtained as the sum of the indices
of the corresponding nested sequences of two-centered bound states.

In the simplest case of two centers, as noted earlier, the
phase space $\cMad_2(\gamma_1,\gamma_2;z)$
is either empty when the conditions \eqref{cond2center} or  \eqref{cond2centeralign}
are violated, or a two-sphere equipped
with the symplectic form $\varpi=\frac12\gamma_{12} \varpi_{S^2}$.
The index of the Dirac operator coupled to $\varpi$ is
well-known to be equal (up to sign) to $|\langle \gamma_1,\gamma_2\rangle|$, corresponding
to the number of states in an angular momentum multiplet of spin $J=\frac12 |\langle \gamma_1,\gamma_2\rangle|-1$.
Accordingly, the equivariant Dirac index is equal to the character of this representation,
\be
\label{gt2}
g(\gamma_1,\gamma_2;z,y) = -\sign(\gamma_{12})\, \kappa(\gamma_{12}),
\ee
where we denoted\footnote{For brevity we shall always omit the dependence of $\kappa$ on $y$.}
\be
\kappa(x)=(-1)^x \, \frac{ y^x-y^{-x}}{y-y^{-1}}\, ,
\label{kappadef}
\ee
as in \cite{Manschot:2010qz}. This answer can also be obtained by directly solving
the quantum mechanics describing two  mutually non-local
dyons  \cite{Denef:2002ru}, or by localization with respect to rotations
along a fixed axis. In that case the two opposite powers of $y$ in the numerator
arise from the north and south pole on the sphere.
The two-centered configuration corresponds
to the single flow tree $\gamma\to\gamma_1+\gamma_2$, corresponding
to the bracketing $(12)$.  Assuming that the condition \eqref{cond2centeralign} is
automatically satisfied for the range of moduli $z$ of interest, the contribution of this tree  to the
rational index $\bOm(\gamma,z)$ is then
\be
\label{flow2}
\bOm_{(12)}(z)=
-\frac12\Bigl[\sign\, \Im (Z_{\gamma_1} \bar Z_{\gamma_2})+ \sign(\gamma_{12})
\Bigr] \kappa(\gamma_{12})\, \bOm_*(\gamma_1)\, \bOm_*(\gamma_2) .
\ee
The `sign factor' in square brackets ensures that
this contribution  is absent unless  the stability condition \eqref{cond2center} is
obeyed \cite{Manschot:2009ia}.
In particular, due to \eqref{attrZZ}, it is always absent in the vicinity of the attractor
point $z_{\gamma}$.

For more than two centers,  the contribution of each tree should then be
given by the product of indices of  two-centered bound
states \eqref{gt2} appearing at each level of the tree, namely
(up to a sign $\epsilon_T \equiv \prod_{v\in V_T}\sign(\gamma_{\Lv{v}\Rv{v}})$ which we will treat separately)
\be
\label{kappaT0}
\kappa(T) \equiv (-1)^{n-1} \prod_{v\in V_T} \kappa( \gamma_{\Lv{v}\Rv{v}}).
\ee
Assuming that identical constituents can be treated as distinguishable particles at the expense
of replacing the index $\Omega(\gamma)$ by its rational counterpart \eqref{defbOm}, as discussed
in the previous subsection,  we are
therefore lead to conjecture that the total index is given by
\be
\label{Omsumtree}
\bOm(\gamma,z,y) =
\sum_{\gamma=\sum_{i=1}^n \gamma_i}
\frac{\gtr(\{\gamma_i\}, z,y)}{|{\rm Aut}\{\gamma_i\}|}\,
\prod_{i=1}^n \bOm_*(\gamma_i,y),
\ee
where the sum over $\{\gamma_i\}$ runs over  unordered decompositions of $\gamma$
into sums of positive vectors $\gamma_i\in\Lambda_+$,
and the `tree index' $\gtr$ is a sum over all stable
flow trees,
\be
\label{defgtree}
\gtr(\{\gamma_i\}, z,y)
=  \sum_{T\in \cT^{\rm ad}_n(\{\gamma_i\},z)}\epsilon_T\, \kappa(T)
=\sum_{T\in \cT_n(\{\gamma_i\},z)}\Delta(T)\, \kappa(T) .
\ee
In the second equality, following \cite{Manschot:2010xp} we extended the sum to all trees $T\in\cT_n$, 
at the cost of inserting a factor vanishing unless the stability condition \eqref{condtree} is obeyed
at each vertex, in which case it is equal\footnote{The equivalence between the stability condition  \eqref{condtree} and the condition
$\Delta(T)\ne 0$ only holds provided none of the arguments of the sign functions in \eqref{kappaT}
vanish. If one of the DSZ pairings $\gamma_{\Lv{v}\Rv{v}}$ vanishes,
the prefactor $\kappa(T)$ vanishes as well, so the second equality in \eqref{defgtree}
is still valid. More problematic
is the case where $\Im\bigl[ Z_{\gamma_{\Lv{v}}}\bZ_{\gamma_{\Rv{v}}}(z_{p(v)})\bigr]$
vanishes for some vertex $v$. If we assume that $z$ does not sit on any wall of marginal stability,
then this cannot happen for the root vertex $v_0$, but it may still happen for one of its descendants.
A simple example is that of a 3-node tree $(1(23))$ with $\gamma_3=\gamma_1$: on the
locus where $\Im[Z_{\gamma_1} \bar Z_{\gamma_1+\gamma_2}(z_1)]=0$,
which defines the moduli $z_1$  at the first splitting,
the quantity $\Im[Z_{\gamma_2} \bar Z_{\gamma_3}(z_1)]$ relevant for the bound state
$\gamma_2+\gamma_3$ also vanishes,
so that $\Delta(T)$ becomes ill-defined. (We thank J. Manschot for pointing out this issue.)
This issue can be traced to  the matrix $\gamma_{ij}=\langle \gamma_i, \gamma_j\rangle$ 
being non-generic, in the sense that some linear
combinations  $\sum_{i<j} p_{ij} \gamma_{ij}$ with integer coefficients $p_{ij}$ vanish.
To avoid this problem, we follow the prescription of \cite{Manschot:2013sya} and
define $\gtr$ by perturbing $\gamma_{ij}$ infinitesimally so that it becomes generic.
Since $\gtr$ is manifestly a continuous function of the $\gamma_{ij}$'s
for generic values of $z$, due to the factor $\kappa(T)$ multiplying $\Delta(T)$,
the result does not depend on the choice of the perturbation.\label{foogen}}
 to $\epsilon_T$,
\be
\label{kappaT}
\Delta(T)=\frac{1}{2^{n-1}} \prod_{v\in V_T}
\Bigl[\sign \,\Im\bigl[ Z_{\gamma_{\Lv{v}}}\bZ_{\gamma_{\Rv{v}}}(z_{p(v)})\bigr]+ \sign (\gamma_{\Lv{v}\Rv{v}}) \Bigr] .
\ee
Again, in writing \eqref{kappaT} we assumed that the second condition
in \eqref{condtree} is automatically satisfied for the range of moduli $z$ of interest.

\medskip

Several comments about the proposal \eqref{Omsumtree}, \eqref{defgtree} are in order:

\begin{itemize}
\item[i)]
By construction the tree index \eqref{defgtree} is a Laurent polynomial in $y$ with integer coefficients, 
symmetric under $y\to 1/y$.

\item[ii)] Due to the observation  in \eqref{attrZZ}, the factor in $\Delta(T)$ associated
to the root vertex automatically  vanishes at the attractor point $z=z_\gamma$. Therefore
the tree index $\gtr$ vanishes at this point, except in the case  $n=1$.
Thus, \eqref{Omsumtree} automatically holds at the attractor point $z=z_\gamma$.
In order to prove that it is true for any $z$, it  suffices to prove that it is consistent with
the wall-crossing formula.

\item[iii)]
To evaluate the sign factor \eqref{kappaT}, it appears that one needs to
compute the attractor flow along each edge and find the value $z_v$ at which it
crosses the wall  of marginal stability for the bound state $\gamma_v \to \gamma_{\Lv{v}}+\gamma_{\Rv{v}}$.
While this is a non-trivial problem in general, the
precise value of $z_v$  is however irrelevant, since we only need to evaluate the sign of
$\Im\bigl[ Z_{\gamma_{\Lv{v}}}\bZ_{\gamma_{\Rv{v}}}(z_{p(v)})\bigr]$ for each vertex,
 in terms of moduli at infinity $z_\infty$.
In section \ref{subsec-flow} below, we shall show that  these
signs can  actually be determined in terms of $z_\infty$ for an arbitrary flow tree,
without solving the attractor flow along the edges explicitly (but assuming that
such a flow does exist).

\item[iv)]
Across the wall of marginal stability
$\Im[Z_{\gamma_L} \bZ_{\gamma_R} (z)]=0$
defined by a pair of primitive\footnote{Here, by primitive we mean that
all charges with non-zero index in the two-dimensional lattice
spanned by $\gamma_L$ and $\gamma_R$ are linear combinations
$N_L \gamma_L + N_R\gamma_R$ with coefficients $N_L, N_R$
of the same sign.} vectors $(\gamma_L,\gamma_R)$,
the discontinuity of Eq. \eqref{Omsumtree} with  $\gamma=\gamma_L+\gamma_R$
arises from the contribution of all flow trees which
start with the same splitting $\gamma\to \gamma_L+\gamma_R$ at the root of the tree.
The discontinuity is then
\be
\label{primwc}
\Delta\bOm(\gamma_L+\gamma_R) =  -\sign(\gamma_{LR})\, \kappa(\gamma_{LR})\, \bOm(\gamma_L,z)\, \bOm(\gamma_R,z) ,
\ee
where $\Delta\bOm$ is defined as the difference between the value
of $\bOm(\gamma_L+\gamma_R,z)$ in the region where $\gamma_{LR}\, \Im(Z_{\gamma_L} \bar Z_{\gamma_R})>0$
(the bound state exists), minus the value in the region where
$\gamma_{LR}\, \Im(Z_{\gamma_L} \bar Z_{\gamma_R})<0$ (the bound state does not exist).
The jump \eqref{primwc} is indeed consistent with (refined) primitive wall-crossing formula
\cite{Denef:2007vg,Diaconescu:2007bf}. When $\gamma=N_L \gamma_L + N_R\gamma_R$
with $N_L$ and/or $N_R$ bigger than one, all trees whose first splitting
$(\gamma_{L(v_0)},\gamma_{R(v_0)})$ lies in the two-dimensional lattice spanned
by $(\gamma_L,\gamma_R)$ can contribute. We have not proved that the resulting
discontinuity is consistent with the general wall-crossing  formula of Kontsevitch and
Soibelman \cite{ks}, but the similarity of \eqref{Omsumtree} with the
Coulomb branch formula discussed in Section \ref{sec-gC} strongly suggests that
this is the case  (see Section \ref{sec-quiver} for a more detailed discussion).

\item[v)]
In addition, the contribution of each tree  is also discontinuous whenever
$z_{p(v)}$ crosses the wall of marginal stability associated to $(\gamma_{\Lv{v}},\gamma_{\Rv{v}})$
for any vertex $v$ in the tree. Unless $v$ is the root vertex,
this locus does not coincide with any wall of marginal stability.
We shall demonstrate that these apparent discontinuities in fact cancel after summing over all trees.
Thus, the sum over trees in \eqref{defgtree} is discontinuous only
on the walls of marginal stability for the bound states of charge $\gamma$,
as predicted by the primitive wall-crossing formula \eqref{primwc}. The appearance
or disappearance of individual trees across `fake walls' reflects a change in the
internal structure of the bound state inside a given chamber, as noted previously in
\cite{Manschot:2010xp,Andriyash:2010yf,Chowdhury:2012jq}.

\item[vi)]
As a side remark, if  none of the charges $\gamma_i$ coincide (and assuming that the matrix
$\gamma_{ij}$ is generic for all possible splittings, see footnote \ref{foogen}),
the symmetry factor in \eqref{Omsumtree}
is equal to one and the sum over splittings and flow trees
can be generated by iterating the quadratic equation\footnote{As a baby version of this phenomenon,
note that the generating function $B(z)$ of the numbers $b_n$
of unordered full binary trees with $n$ leaves satisfies
$$
z= B(z) - \frac12\, B(z)^2,
\qquad
B(z)=z+\sum_{n\geq 2} b_n\, \frac{z^n}{n!} = 1-\sqrt{1-2z} .
$$
\label{foot-Ntree}
}
\bea
\label{itereq}
\bOm(\gamma,z) &= &\bOm_*(\gamma)
 \\
&- & \sum_{\substack{\gamma=\gamma_L+\gamma_R\\
\langle \gamma_L,\gamma_R \rangle \neq 0}}
\frac12\,\Bigl[ \sign\,\Im\bigl[Z_{\gamma_L} \bar Z_{\gamma_R}(z)\bigr] +
\sign (\gamma_{LR}) \Bigr] \kappa(\gamma_{LR})\,
\bOm(\gamma_L,z_{LR})\, \bOm(\gamma_R,z_{LR}),
\nn
\eea
where $z_{LR}$ is the point where the attractor flow $\cA(\gamma_L+\gamma_R,z)$
crosses the wall of marginal stability $\Im(Z_{\gamma_L}\bar Z_{\gamma_R}(z_{LR}))=0$,
and we omit the dependence on $y$ on both sides.

\end{itemize}

In the remainder of this section, we shall explain how to evaluate the tree index $\gtr$
purely in terms of asymptotic data, introducing the notion of `discrete attractor flow'
along the tree, and provide alternative representations for $\gtr$ which make
the cancellation of discontinuites across `fake walls' manifest.

\subsection{Example: three centers}
\label{subsec-ex3}

As a warm-up, we first consider the case $n=3$,
which has been analyzed in detail in \cite{Manschot:2010xp} and nicely
illustrates the mechanisms at play.
In the following it will be useful to define\footnote{Since the stability
condition is unaffected by an overall rescaling, the parameters
$c_i(z)$ are really valued in the real projective space $\IR \IP^n$. In Section \ref{sec-quiver}
we shall see that they coincide with the Fayet-Iliopoulos parameters
in quiver quantum mechanics, again up to an overall scale.}
\be
c_i(z) =  \Im\bigl[ Z_{\gamma_i}\bZ_\gamma(z)\bigr],
\qquad
c_{ij}(z) =  \Im\bigl[ Z_{\gamma_i}\bZ_{\gamma_j}(z)\bigr].
\label{notaion-c}
\ee
When the argument is omitted, these parameters will be implicitly evaluated at infinity, i.e. $c_i\equiv c_i(z_\infty)$.
Note that due to $\gamma=\sum_{i=1}^n \gamma_i$, they satisfy
\be
\sum_{j=1}^n c_{ij}(z)=c_i(z),
\qquad
\sum_{i=1}^n c_i(z)=0.
\ee
We shall abuse notation and write $\gamma_{i+j}=\gamma_i+\gamma_j$,
$\gamma_{i+j,k}=\gamma_{ik}+\gamma_{jk}$, $c_{i+j}=c_i+c_j$, etc.

\lfig{Attractor flow tree with three centers.}{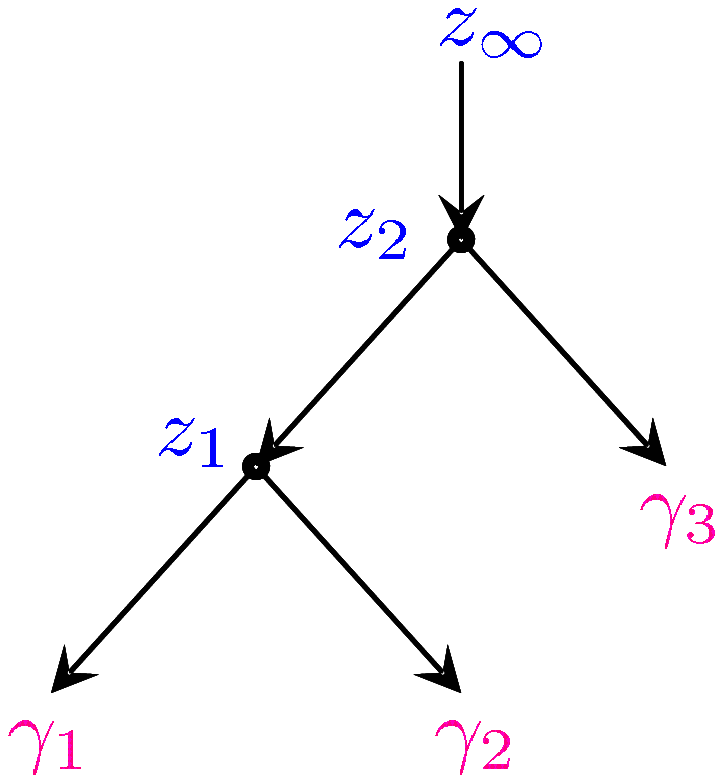}{5.5cm}{fig-AFtree3}{-1.2cm}

For three centers, the three possible flow trees are related  to the tree $((12)3)$ depicted
in Fig. \ref{fig-AFtree3} by cyclic permutations of the charges.
The moduli $z_1$ and $z_2$  at the two vertices  are defined by the conditions
$\Im[Z_{\gamma_1}\bar Z_{\gamma_2}(z_1)] = \Im[Z_{\gamma_1+\gamma_2}\bar Z_{\gamma_3}(z_2)]=0$.
Using the notations \eqref{notaion-c}, the formula \eqref{kappaT} implies that the contribution
of this tree to the BPS index $\bOm(\gamma,z)$ is given by
\be
\label{Om-af3}
\begin{split}
\bOm_{((12)3)}(z) =  &\,
\frac14\, \Bigl[\sign (\gamma_{1+2,3}) +\sign(c_{1+2}) \Bigr]
\, \Bigl[  \sign (\gamma_{12}) +\sign (c_{12}(z_2))\Bigr]
\\
&\,  \times
\kappa(\gamma_{12})\,\kappa(\gamma_{1+2,3} )\,
\bOm_*(\gamma_1)\,  \bOm_*(\gamma_2)\,  \bOm_*(\gamma_3) .
\end{split}
\ee

Our first goal is to express $\sign (c_{12}(z_2))$ in terms of $z_\infty$.
To this end, let us substitute $(\gamma_1,\gamma_2,\gamma') \to (\gamma_{1+2},\gamma_3,\gamma_1)$  in \eqref{attflowst},
obtaining
\be
e^{-U_2}\,\frac{c_1(z_2)}{|Z_\gamma(z_2)|} =
\frac{\langle\gamma,\gamma_1\rangle}{\langle \gamma,\gamma_3\rangle}\,
\frac{c_{1+2}}{|Z_{\gamma}|} +\frac{c_1}{|Z_\gamma|} \, .
\ee
Therefore, one finds
\be
\begin{split}
\sign( c_1(z_2))=&\,
\sign \bigl[\gamma_{1+2,3}
( \langle\gamma,\gamma_1\rangle\,c_{1+2}+ \langle \gamma,\gamma_3\rangle\,c_1)
\bigr]
\\
= &\, \sign(\gamma_{1+2,3})\,
\sign(\gamma_{23}c_1+\gamma_{31}c_2+\gamma_{12}c_3).
\end{split}
\label{t2t3}
\ee
Since  by assumption $Z_{\gamma_{1+2}}(z_2)$, $Z_{\gamma_3}(z_2)$ and their sum $Z_{\gamma}(z_2)$
all have the same phase,
it follows that
\be
\label{samephase}
\sign(c_{12}(z_2)) =
\sign\, \Im( Z_{\gamma_1} \bar Z_{\gamma_{1+2}}(z_2)) =
\sign(c_1(z_2)),
\ee
which was computed in \eqref{t2t3}.  Thus, the contribution of the tree $((12)3)$ to the BPS index
becomes
\be
\label{flow3}
\begin{split}
\bOm_{((12)3)}(z) =  &\,
 \frac14\, \Bigl[\sign (\gamma_{1+2,3}) -\sign(c_{3}) \Bigr]
\, \Bigl[  \sign (\gamma_{12}) + \sign (\gamma_{1+2,3})\,\sign(A_{123})\Bigr]
\\
&\,  \times  \kappa(\gamma_{12})\, \kappa(\gamma_{1+2,3} )\,
\bOm_*(\gamma_1)\,  \bOm_*(\gamma_2)\,  \bOm_*(\gamma_3),
\end{split}
\ee
where we defined
\be
\label{defA123}
A_{123} =\gamma_{23}c_1+\gamma_{31} c_2+\gamma_{12} c_3.
\ee
Assuming that the only possible splittings of $\gamma$ involve the charge vectors $\gamma_1, \gamma_2, \gamma_3$,
the total index is then given by
\be
\label{flow3tot}
\begin{split}
\bOm(\gamma,z)  = &\,
\bOm_*(\gamma) +  \bOm_{(1,2+3)}(z)
+  \bOm_{(2,3+1)}(z) + \bOm_{(3,1+2)}(z)
\\
&\, +
\bOm_{((12)3)}(z) + \bOm_{((23)1)}(z) + \bOm_{((31)2)}(z),
\end{split}
\ee
where $\bOm_{(1,2+3)}$ denotes \eqref{flow2} with $\gamma_2$ replaced by $\gamma_{2+3}$, and $\bOm_{((23)1)}$, $\bOm_{((31)2)}$
are obtained by cyclic permutations of $(\gamma_1,\gamma_2,\gamma_3)$
in \eqref{flow3}. In particular, the tree index for three centers is given by
\be
\label{gtr3}
\gtr =  \frac{3}{4}\,\Sym\left\{\kappa(\gamma_{12} )\,\kappa(\gamma_{1+2,3} )
\Bigl(\sign(\gamma_{1+2,3})-\sign(c_3)\Bigr)
\Bigl(\sign(\gamma_{12})+\sign(\gamma_{1+2,3}A_{123})\Bigr)\right\},
\ee
where $\Sym$ denotes the complete symmetrization over all charges (with weight $1/n!$).

\begin{figure}[t]
\centerline{\includegraphics[width=7.7cm]{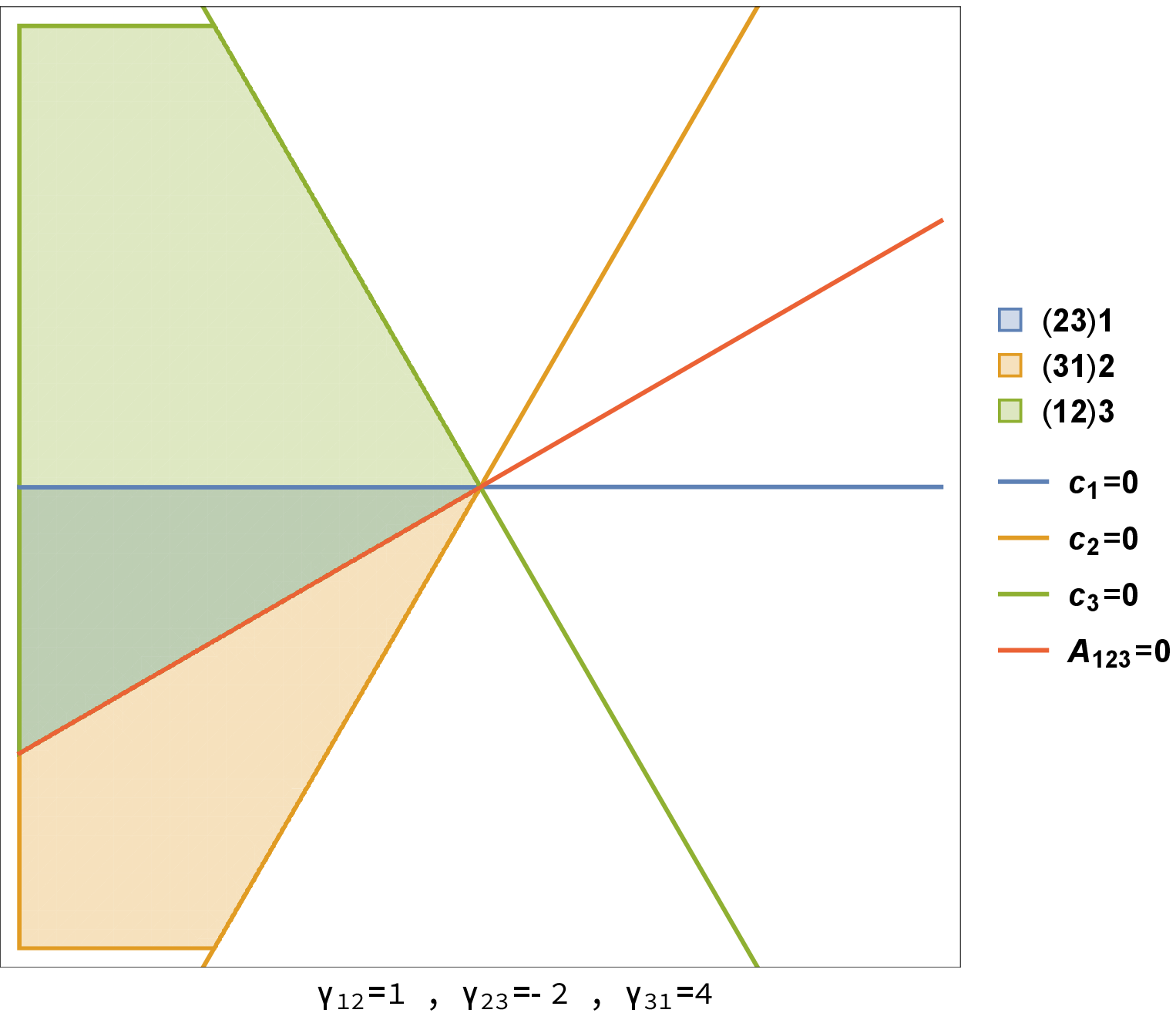}\hfill
\includegraphics[width=7.7cm]{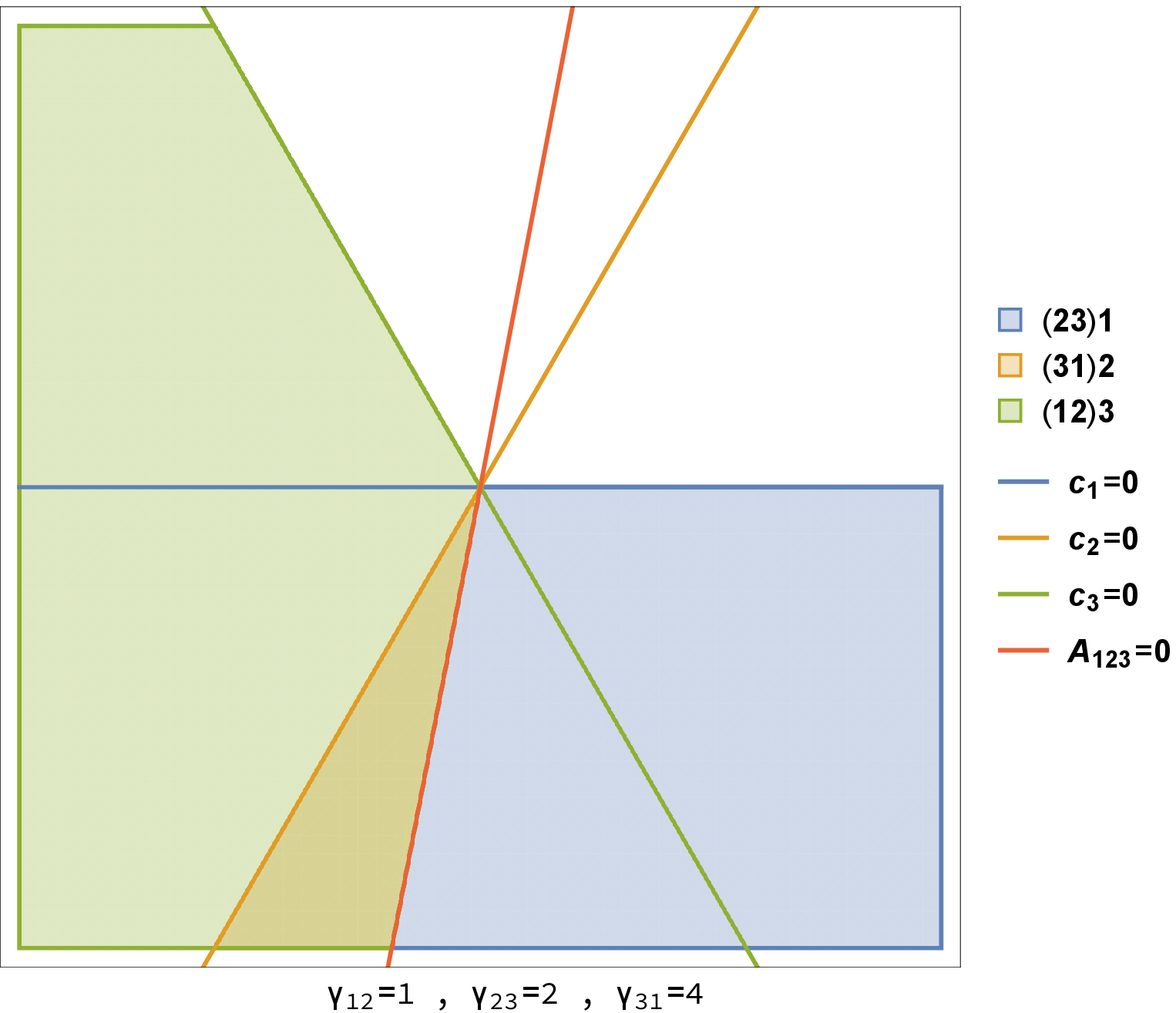}}
\caption{In these two figures we depict the regions of stability of the flow trees for $n=3$ centers
as a function of the $c_i$'s (subject to the constraint $c_1+c_2+c_3=0$)
for two representative choices of the products $\gamma_{12}, \gamma_{23}, \gamma_{31}$.
The horizontal and vertical axis are
$x=\frac{\sqrt{3}}{2}(c_2-c_3)$ and $y=-\frac32(c_2+c_3)$, so that a cyclic permutation of $1,2,3$
amounts to a rotation by $2\pi/3$ around the origin. The lines $c_i=0$ are walls of marginal
stability for the decay $\gamma\to \gamma_i + \gamma_{j+k}$, while the red line $A_{123}$ corresponds to a
``fake wall'', or recombination wall. On the left, the trees $((23)1)$ and $((12)3)$ co-exist in the region
between $c_1=0$ and $A_{123}=0$, and BPS bound states only exist in the colored region between $c_2=0$ and $c_3=0$.
On the right, the trees $((31)2)$ and $((12)3)$ coexist
in the region between $c_2=0$ and $A_{123}=0$, and BPS bound states
exist in the colored region between $c_1=0$ and $c_3=0$, but there are scaling solutions for
any value of the $c_i'$'s. In both cases, the index is constant across the fake wall $A_{123}=0$,
even though  the allowed trees differ on both sides.\label{fig-threeplot}}
\end{figure}

It is  straightforward to check the consistency of the representation \eqref{flow3tot} with the
primitive wall-crossing formula \eqref{primwc}.
For example, on the wall of marginal stability where
$c_3=\Im (Z_{\gamma_{3}} \bZ_{\gamma_{1+2}})$ vanishes we have $z_\infty=z_2$
and, as noted above \eqref{samephase}, $Z_{\gamma_{1+2}}$, $Z_{\gamma_3}$
and $Z_{\gamma}$ all have the same phase, so \eqref{defA123} reduces to
\be
A_{123}|_{c_3=0}= \gamma_{23} c_1+\gamma_{31} c_2
= \gamma_{1+2,3} c_{12}.
\ee
Using this in \eqref{flow3}, the discontinuity of  $\bOm(\gamma,z)$ is then proportional
to
\be
\begin{split}
&\, \kappa(\gamma_{1+2,3})
\( \bOm_*(\gamma_{1+2})
-\frac12 \Bigl[  \sign(\gamma_{12}) +\sign(c_{12})
\Bigr]
\kappa(\gamma_{12}) \,\bOm_*(\gamma_1)\,\bOm_*(\gamma_2) \)
\bOm_*(\gamma_3)
\\
=&\, \kappa(\gamma_{1+2,3})\,\bOm(\gamma_{1+2},z)\,\bOm(\gamma_3,z),
\end{split}
\ee
consistently with the primitive wall-crossing formula \eqref{primwc}.

On the other hand, it appears that \eqref{flow3tot} is discontinuous on the codimension one
locus where the sign of $A_{123}$ changes, which does not coincide with any wall of marginal stability.
However, due to the cyclic symmetry of $A_{123}$, the discontinuity arises simultaneously
for the flow tree $((12)3)$ and its images under cyclic permutations (see Fig. \ref{fig-threeplot}
for a plot of the stability regions of the various trees for two representative examples).
As a result, while each of these contributions is discontinuous across $A_{123}=0$,
the sum over flow trees turns out to be smooth. Indeed, the coefficient of $\sign(A_{123})$
in \eqref{flow3tot} is the product of the $\bOm_*(\gamma_i)$'s times
\be
\label{coefA123}
\Bigl[ 1 - \sign(\gamma_{1+2,3}c_3) \Bigr]
\kappa(\gamma_{12} )\,\kappa(\gamma_{1+2,3})
+ {\rm cycl}\, .
\ee
The key property ensuring the vanishing of this expression is the identity (valid for any $y$)
\be
\label{kappa123}
\kappa(\gamma_{12} )\,\kappa(\gamma_{1+2,3} ) + {\rm cycl} = 0.
\ee
Using it, one finds that \eqref{coefA123} reduces to
\be
\Bigl[\sign(\gamma_{1+2,3}c_3)-\sign(\gamma_{2+3,1} c_1) \Bigr]
\kappa(\gamma_{23} )\,\kappa(\gamma_{2+3,1}) + (1\leftrightarrow 2).
\ee
Rewriting \eqref{defA123}  as $A_{123}=\gamma_{1+2,3} c_1-\gamma_{2+3,1} c_3$, it
is now clear that this expression vanishes on the locus $A_{123}=0$.
Thus, after summing all trees, the total index is smooth across this locus and
does not have any discontinuities beyond those predicted by wall crossing.

Given the fact that the discontinuity at $A_{123}=0$ cancels, one may wonder whether
it is possible to rewrite the index in a form which does not involve $\sign(A_{123})$ at all,
but only $\sign(c_i)$ corresponding to the decay $\gamma\to \gamma_{i}+\gamma_{j+k}$.
To this end, let us rewrite the tree index \eqref{gtr3} as
\be
\gtr=
\frac{3}{4}\,\Sym\left\{\kappa(\gamma_{12} )\,\kappa(\gamma_{1+2,3} )\Bigl[-\sign(\gamma_{1+2,3}^{-1}c_3)\,\sign (A_{123})
+\sign(\gamma_{12})\Bigl(\sign(\gamma_{1+2,3})-\sign(c_3)\Bigr)\Bigr]\right\},
\ee
where we used the identity \eqref{kappa123} to drop the term proportional
to $[\sign(\gamma_{1+2,3})]^2$. Next, we use the same identity to replace
\be
\kappa(\gamma_{12} )\,\kappa(\gamma_{1+2,3}) = \frac23\, \kappa(\gamma_{12} )
\,\kappa(\gamma_{1+2,3} )-\frac13\, \Bigl( \kappa(\gamma_{23} )\,\kappa(\gamma_{2+3,1} )
+ \kappa(\gamma_{31} )\,\kappa(\gamma_{3+1,2}) \Bigr)
\ee
and relabel the charges after this replacement, arriving at
\be
\begin{split}
\gtr=&\, \frac12\,\Sym\left\{\kappa(\gamma_{12} )\,\kappa(\gamma_{1+2,3} )\Bigl[
\Bigl(\sign(\gamma_{2+3,1}^{-1}c_1)-\sign(\gamma_{1+2,3}^{-1}c_3)\Bigr)\sign (A_{123})
\right.
\\
&\,\left.
+\sign(\gamma_{12})\Bigl(\sign(\gamma_{1+2,3})-\sign(c_3)\Bigr)-\sign(\gamma_{23})\Bigl(\sign(\gamma_{2+3,1})-\sign(c_1)\Bigr)
\Bigr]\right\}.
\end{split}
\label{Z3charge0}
\ee
Finally, taking into account that
$A_{123}= \gamma_{1+2,3}c_1-\gamma_{2+3,1}c_3$
and using the sign identity \eqref{signprop-ap},
one obtains that the tree index takes the form
\be
\begin{split}
\gtr=&\, \frac12\,\Sym\left\{\kappa(\gamma_{12} )\,\kappa(\gamma_{1+2,3} )\Bigl[
\sign(\gamma_{2+3,1})\sign(\gamma_{1+2,3})-\sign(c_1)\sign(c_3)
\right.
\\
&\, \left.
+\sign(\gamma_{12})\Bigl(\sign(\gamma_{1+2,3})-\sign(c_3)\Bigr)-\sign(\gamma_{23})\Bigl(\sign(\gamma_{2+3,1})-\sign(c_1)\Bigr)
\Bigr]\right\}.
\end{split}
\label{Z3charge}
\ee
This representation involves sign functions whose arguments are all expressed through the moduli
at infinity and vanish only on the walls of marginal stability corresponding to the bound states of total charge $\gamma$.

\subsection{Discrete attractor flow}
\label{subsec-flow}

As mentioned in \S\ref{subsec-afc}, it appears that in order to compute
the weight factor \eqref{kappaT}
ensuring the stability condition, one needs to
solve the attractor flow along each edge in order to determine the explicit values of the
moduli $z_v$ attached to each vertex. However, the weight factor \eqref{kappaT} only
depends on the moduli through
the sign of $\Im\bigl[ Z_{\gamma_{\Lv{v}}}\bZ_{\gamma_{\Rv{v}}}\bigr]$
evaluated at the parent vertex $z_{p(v)}$. We now explain how these signs can be determined
in terms of asymptotic data
without knowing $z_{p(v)}$ itself, generalizing the procedure used above for $n=3$.

The key property is again Eq. \eqref{attflowst}. Recalling the notation \eqref{notaion-c}
and specializing to $(\gamma_1,\gamma_2,\gamma')=(\gamma_L,\gamma_R,\gamma_i)$,
it may be rewritten as
\be
 \frac{e^{-U(r_1)}|Z_{\gamma}(z)|}{|Z_\gamma(z_1)|} \, c_i(z_1) =
\frac{\langle\gamma,\gamma_i\rangle}{\langle \gamma_L,\gamma_R\rangle}\,
\sum_{j=1}^n  \mm^j_L c_j(z)+ c_i(z) ,
\ee
where $\mm^i_v$ are the coefficients (equal to 0 or 1) of the charge $\gamma_v$
at the vertex $v\in V_T$ on the basis spanned by the vectors $\gamma_1,\dots,\gamma_n$ assigned to the leaves,
i.e. $\gamma_v=\sum_{i=1}^n \mm_v^i \gamma_i$. This equation relates the coefficients
$c_i(z_1)$, determining the stability of the BPS bound state $\gamma_L+\gamma_R$, in terms
of the coefficients $c_i(z)$, up to an irrelevant overall positive scale factor. While the relation \eqref{attflowst}
was derived
in the context of a two-centered black hole with $z$ labelling the moduli at infinity, it holds just
as well for nested sequences of two-centered bound states, where $z$ now labels the moduli
at the parent vertex $p(v)$ and the total charge should be replaced by $\gamma_v$.
Thus, it allows to determine the coefficients
$c_{v,i}=\Im\bigl[ Z_{\gamma_i}\bZ_{\gamma_v}(z_{p(v)})\bigr]$ at all vertices of the
tree, starting from the root where $c_{v_0,i}=c_i(z_\infty)$
and propagating them down the tree using the `discrete attractor flow'
\be
\label{cistar}
c_{v,i} =
 c_{p(v),i} - \frac{\langle \gamma_v,\gamma_i \rangle}{\langle \gamma_v,\gamma_{\Lv{v}} \rangle}
\,\sum_{j=1}^n  \mm^j_{\Lv{v}} c_{p(v),j}
\ee
at each vertex.
Note that this relation is invariant under exchanging $\gamma_L$ and $\gamma_R$,
and ensures that $\sum \mm^i_{\Lv{v}} c_{v,i}= \sum \mm^i_{\Rv{v}}c_{v,i}=0$. It is important
however to keep in mind that the relation \eqref{cistar} assumes that the flow starting
from $z_{p(v)}$ can be continued all the way until it crosses the locus where
$\Im\bigl[ Z_{\gamma_{\Lv{v}}}\bZ_{\gamma_{\Rv{v}}}(z_v)\bigr]=0$ (as opposed to terminating
on a point where $Z_{\gamma_v}$ vanishes), and moreover that the central charges
$Z_{\gamma_{\Lv{v}}}$ and $Z_{\gamma_{\Rv{v}}}$ are actually aligned at this point
(as opposed to being anti-aligned).

For some purposes, it can be useful to have a formula expressing $\Delta(T)$
directly in terms of the moduli at spatial infinity, or more precisely in
terms of the corresponding parameters $c_{i}=c_i(z_\infty)$.
To this aim, let us consider a branch inside a flow tree, starting from the root $v_0$
and consisting of $\ell-1$ edges (see Fig. \ref{fig-AFtreen}). We denote by $\alpha_1,\dots, \alpha_\ell$
the charges attached to the descendants of the vertices along the branch, such that it corresponds to
the nested sequence of bound states
\be
\begin{array}{cccccccccc}
\gamma &\to& \ha_{\ell-1} & \to & \cdots & \to & \ha_{2}  & \to & \alpha_1 &
\\
&& + &&&& + && +
\\
&& \alpha_{\ell} &&&& \alpha_{3} && \alpha_2
\end{array}
\label{deftree}
\ee
and $\ha_k=\alpha_1+\cdots+\alpha_k$.
The charges $\alpha_i$ ($i=1,\dots, \ell$) are in general linear combinations of the  charges $\gamma_i$,
($i=1,\dots, n$) of the constituents attached to the leaves of the full tree, from which the branch has been extracted.
We denote the moduli at the vertices along the tree by $z_i$ in the same order as charges, $z_\ell=z_\infty$
corresponding to the moduli at spatial infinity.

\lfig{A branch of a flow tree and its labeling relevant for the expression of the stability
conditions in terms of $z_\infty$.}
{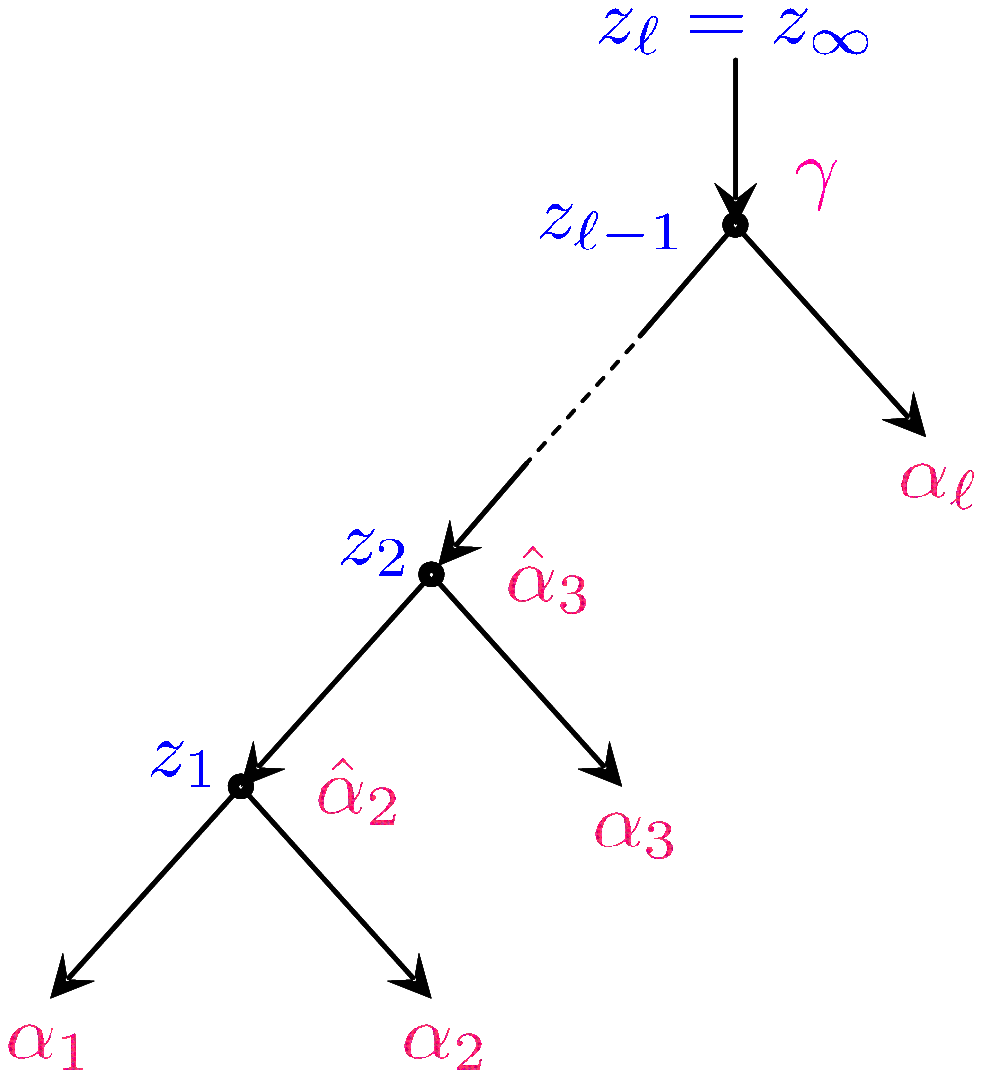}{6.5cm}{fig-AFtreen}{-1.cm}

Our goal is to express the sign factor $\sign\,\Im\bigl[ Z_{\alpha_1}\bZ_{\alpha_2}(z_2)\bigr]$ governing the stability of the
innermost bound state in terms of $z_\infty$.
For this purpose, we define the following family of sign functions
\be
S_{k}(\{a_i\})
= \sign\, \Im\[ \(\sum_{i=1}^{k} a_{i} Z_{\alpha_i} (z_{k})\) \bZ_{\ha_{k}}
(z_{k})\],
\label{defSk}
\ee
where $k=2, \dots, \ell$. The real parameters $a_{i}$ will be fixed
momentarily in such a way that $S_{k-1}(\{a_i\})$ is related to its counterpart $S_{k}(\{a_i\})$,
in which the central charges are evaluated one step up along the attractor flow.
To this end, we note that the discrete attractor flow equation \eqref{cistar} specialized for the vertex $v_{k-1}$
implies
\be
\Im[Z_{\alpha_i}\bZ_{\ha_k}(z_{k-1})] \propto
\Im[Z_{\alpha_i} \bZ_{\ha_k}(z_{k})] + \frac{\beta_{ki}}{\beta_{kk}}\,
\Im[ Z_{\ha_{k-1}} \bZ_{\ha_k}(z_{k})],
\label{dafeq-n}
\ee
where the proportionality coefficient is independent of $i$ and positive,
and we defined
\be
\beta_{ki}= \langle \alpha_1+\dots+\alpha_k,\alpha_i\rangle.
\ee
Since by definition $z_{k-1}$ is the point where the attractor flow $\cA(\ha_k,z_{k})$
crosses the wall of marginal stability $\Im[Z_{\ha_{k-1}}\bZ_{\alpha_{k}}(z_{k-1})]=0$,
the central charges $Z_{\ha_{k-1}}(z_{k-1})$ and $Z_{\ha_k}(z_{k-1})$
have the same phase.
Hence, one can replace the latter by the former in the left-hand side of \eqref{dafeq-n},
which then reproduces one term in the sum in \eqref{defSk}.
Thus, we find
\be
\label{twoSk}
S_{k-1}(\{a_i\})= \sign \,\Im\left[
\sum_{i=1}^{k-1}  a_{i}  \left( Z_{\alpha_i}
- \frac{\beta_{ki}}{\beta_{kk}}\, Z_{\alpha_{k}}\right)
\bZ_{\ha_{k}}(z_{k}) \right] = S_{k}(\{a_i\}) ,
\ee
where in the definition of $S_{k}(\{a_i\})$ on the right-hand side, we choose the
first $k-1$ coefficients $a_1,\dots ,a_{k-1}$ to be  identical to the original ones, while
the $k$-th coefficient is related to the preceding ones by
\be
a_{k} =
- \sum_{i=1}^{k-1} \frac{\beta_{ki}}{\beta_{kk}}\, a_i  .
\label{rec-a}
\ee

The recursive relation \eqref{rec-a} fixes all coefficients $a_k$ in terms of $a_1$ and $a_2$.
Moreover, due to $\sum_{i=1}^k \beta_{ki}=0$, a shift of the initial conditions
$a_{1},a_{2}$ by $\lambda$ results into an overall shift $a_k \to a_k+\lambda$ for all $k$,
which does not affect the $S_k$'s. Therefore, we can choose $a_1=0$.
Then \eqref{rec-a} leads to
\be
\begin{split}
a_{1}=0, &
\qquad a_{2}=-1,
\qquad
a_{3} = \frac{\beta_{32}}{\beta_{33}},
\qquad
a_{4}=\frac{\beta_{42}}{\beta_{44}}- \frac{\beta_{43} \beta_{32}}{\beta_{44}\beta_{33}},
\qquad \dots
\end{split}
\ee
More generally,  we find
\be
a_{i} = \sum_{r=2}^{i-1}\, \sum_{2=j_1<j_2<\dots<j_r=i} (-1)^r\,
\prod_{\ell=1}^{r-1} \frac{\beta_{j_{\ell+1}, j_{\ell}}}{\beta_{j_{\ell+1},j_{\ell+1}}}
\label{expr-a}
\ee
for all $i\geq 3$. Using these relations, we can finally evaluate the relevant sign,
which determines the stability of the innermost bound state in terms of $z_\ell=z$,\footnote{Note that,
upon multiplying $a_i$ by the largest denominator $\prod_{j=3}^m \beta_{jj}$, the argument
of the sign becomes a homogenous polynomial of degree $n-2$ in the $\beta_{ij}$'s.}
\be
\sign\,\Im\bigl[ Z_{\alpha_1}\bZ_{\alpha_2}(z_2)\bigr]=-S_{2}(a_1,a_2) = - S_{\ell}(\{a_i\})=
-\sign\( \sum_{i=1}^\ell a_{i} \, \Im\bigl[ Z_{\alpha_i}\bZ_\gamma(z)\bigr] \).
\label{sign-zinf}
\ee

Since the charges
$\alpha_1,\dots, \alpha_\ell$ are  linear combinations of the charges $\gamma_i$ attached
to the leaves of the full tree from which the branch is extracted, the argument of the sign
can be written as a linear combination
$\sum_{i=1}^n a_{vi}  c_i$, where $c_i$ are defined in \eqref{notaion-c}, and $a_{vi}$
form a $n$-dimensional real vector associated to the vertex $v$ and constructed out of the coefficients $a_i$ found above.
In terms of these vectors, the weight factor \eqref{kappaT} is rewritten as
\be
\label{kappaTa}
\Delta(T)=\frac{1}{2^{n-1}} \prod_{v\in V_T}
\[-\sign \( \sum_{i=1}^n a_{vi} c_i\)+ \sign (\gamma_{\Lv{v}\Rv{v}}) \] ,
\ee
which only depends on the asymptotic moduli $z$ through the variables $c_i=c_i(z)$
defined in \eqref{notaion-c}. This result gives a straightforward, algorithmic way of
evaluating the tree index \eqref{defgtree} in terms of asymptotic data without having
to integrate the flow along each edge, provided such a flow exists. Furthermore,
it shows that the tree index
only depends on the charges through the DSZ matrix $\gamma_{ij}$
and on the moduli $z$ through the vector $c_i$. The fact that the factor
$\Delta(T)$ comes multiplied by $\kappa(T)$ in the definition \eqref{defgtree}
also shows that the tree index  is a continuous function of $\gamma_{ij}$ for generic values of the $c_i$'s,
justifying the prescription given in footnote \ref{foogen} for dealing with the case
of non-generic DSZ matrix.

Using the result \eqref{kappaTa}, in appendix \ref{ap-vectors}
we show that the tree index is also continuous across
the `fake walls' associated to the stability of intermediate bound states, despite the fact
that contributions of  distinct flow trees may jump on these loci.
This is done with help of certain `flow vectors' constructed from the coefficients $a_{vi}$.
The recursive relation \eqref{rec-a} allows to prove various symmetry properties of these vectors,
which in turn reduce the proof of cancellation of the fake discontinuities to the cyclic property \eqref{kappa123}
of $\kappa$-factors. We defer this proof to the appendix because in the next subsection
we shall provide a new representation for the BPS index which is manifestly smooth across all fake walls.

\subsection{New formulae for the tree index}
\label{subsec-newrepr}

We have seen at the end of subsection \ref{subsec-ex3} that for $n=3$, it was possible
to rewrite the tree index in such way that it was manifestly constant away from the physical
walls of marginal stability corresponding
to splittings of the total charge $\gamma$ into $\gamma_L+\gamma_R$.
It is natural to expect that this should also be possible for any number of centers.
Unfortunately, it is not straightforward to extend the tricks used for $n=3$, due to the weights
$\kappa(T)$ appearing explicitly in \eqref{defgtree}. While this can be done with some effort
for $n=4$, this procedure becomes extremely cumbersome.

To overcome this problem, we shall first rewrite the tree index in a form which does not contain the
$\kappa$-factors anymore
and is expressed in terms of another, $y$-independent `partial index'. This form
was inspired by the representation \eqref{gCFy} of the Coulomb index,
which will be the subject of section \ref{sec-Coulomb}.

\begin{proposition}
The tree index defined in \eqref{defgtree} can be decomposed as
\be
\gtr(\{\gamma_i\}, z,y)
= \frac{(-1)^{n-1+\sum_{i<j} \gamma_{ij} }}{(y-y^{-1})^{n-1}} \,\sum_{\sigma\in S_n}\,
\Ftr{n}(\{\gamma_{\sigma(i)}\},z)\,
y^{\sum_{i<j} \gamma_{\sigma(i)\sigma(j)}},
\label{gtF}
\ee
where $\sigma$ runs over all permutations of $\{1,\dots, n\}$ and the `partial tree index' $\Ftr{n}(\{\gamma_i\},z)$
is defined by\footnote{The subscript $n$ on $\Ftr{n}$ is redundant
since it equals the cardinality of the set $\{\gamma_i\}$. Nevertheless, we find it useful to display it
and instead sometimes omit the arguments, the set of vectors $\gamma_i$ and moduli $z$.}
\be
\label{defFpl}
\Ftr{n} (\{\gamma_{i}\},z) \equiv
\sum_{T\in \cT_n^{\rm pl}(\{\gamma_i\},z)} \Delta(T),
\ee
where the sum runs over the set $\cT_n^{\rm pl}(\{\gamma_i\},z)$ of planar flow trees with $n$
leaves carrying ordered charges $\gamma_1,\dots, \gamma_n$.
\end{proposition}

Note that it is crucial
to consider the {\it refined} index at $y\neq 1$ for getting the representation \eqref{gtF}, since
each term in the sum over permutations is singular at $y=1$. Nevertheless, the  sum must be
smooth in this limit, since the tree
index $\gtr$ (unlike the Coulomb index $\gref$)
is a symmetric Laurent polynomial.

\begin{proof}
To see the origin of the representation \eqref{gtF}, let us expand
all factors $\kappa(\gamma_{ij})$ in \eqref{kappaT0} using the definition of $\kappa(x)$ in \eqref{kappadef}.
Then each of the $b_n$ flow trees produces, up to a factor of $\frac{(-1)^{\sum_{i<j} \gamma_{ij}}}{(y-1/y)^{n-1}}$,
a sum of $2^{n-1}$ monomials
of the form $y^{\sum_{i<j} \epsilon_{ij} \gamma_{ij}}$ where $\epsilon_{ij}=\pm 1$. For each
assignment of $\epsilon_{ij}$, there exists a unique permutation $\sigma\in S_n$ such that
$\sum_{i<j} \epsilon_{ij} \gamma_{ij}=\sum_{i<j} \gamma_{\sigma(i)\sigma(j)}$. For a given
permutation $\sigma$,  all trees $T$ contributing a term proportional
to $y^{\sum_{i<j} \gamma_{\sigma(i)\sigma(j)}}$ are {\it planar} flow trees 
ordered with respect to $\sigma$, i.e. trees ending on ordered points
$z_{\sigma(1)},\dots, z_{\sigma(n)}$  which can be drawn on the upper half plane without crossings.\footnote{
For any $\sigma$, the number of such planar trees is $\tilde b_n=C_{n-1}$,
where $C_n=\frac{(2n)!}{(n+1)(n!)^2}$ is the $n$-th Catalan number ($C_n=1,1,2,5,14,\dots$ for $n=0,1,2,3,4,\dots$).
The Catalan numbers satisfy the recursion relation $C_{n+1} = \sum_{i=0}^{n} C_i\, C_{n-i}$,
so that $\tilde b_{n}=\sum_{\ell=1}^{n-1} \tilde b_\ell\, \tilde b_{n-\ell}$
corresponding to the obvious ways of constructing a planar tree with $n$ leaves
by merging two planar trees with $\ell$ and $n-\ell$ leaves, respectively. \label{foot-planar}}
As for the usual flow trees, they are labelled by charges, with the charges $\gamma_{\sigma(i)}$ assigned to the end-points,
and contribute the weight $\Delta(T)$ given in \eqref{kappaT} or \eqref{kappaTa}.
\end{proof}

Since planar trees with $n$ end-points may be generated by merging two planar trees with
$\ell$ and $n-\ell$ leaves, with $\ell$ running from 1 to $n-1$  (see footnote \ref{foot-planar}),
it is easy to see that
the  partial index satisfies the following iterative equation\footnote{Throughout
this section we assume that the second condition
in \eqref{condtree} is automatically satisfied for the range of moduli $z$ of interest.}
\be
\Ftr{n}(\{\gamma_i\},z)=\hf\sum_{\ell=1}^{n-1} \bigl( \sign(\cs_\ell)-\sign (\Gamma_{n\ell})\bigr)\,
\Ftr{\ell}(\{\gamma_i\}_{i=1}^\ell,z_\ell)\,
\Ftr{n-\ell}(\{\gamma_i\}_{i=\ell+1}^n,z_\ell),
\label{inductFn}
\ee
where we defined
\be
\cs_k=\sum_{i=1}^k c_i(z),
\qquad
\Gamma_{k\ell}=\sum_{i=1}^k\sum_{j=1}^\ell \gamma_{ij},
\qquad
\beta_{k\ell}=\sum_{i=1}^k \gamma_{i\ell},
\ee
while $z_\ell$ is the value of the moduli where the attractor flow crosses
the wall for the decay $\gamma\to(\gamma_{1+\cdots +\ell},\gamma_{(\ell+1)+\cdots +n})$.
According to \eqref{cistar}, this value corresponds to the parameters
\be
c_i(z_\ell)= c_i(z) -\frac{\beta_{ni}}{\Gamma_{n\ell}}\, \cs_{\ell}.
\label{flowc}
\ee
While Eq. \eqref{inductFn} allows for a very efficient evaluation of $\Ftr{n}$,
it is also the starting point for obtaining new representations for the partial index
which satisfy the properties stated at the beginning of this subsection.

To formulate our results, let us define the following quantities
\be
\Fwl_n(\{\gamma_i,c_i\})=\frac{1}{2^{n-1}}\prod_{i=1}^{n-1}\sgn(\cs_i),
\quad
\tFwl_n(\{\gamma_i,c_i\})=\frac{1}{2^{n-1}}\prod_{i=1}^{n-1}\bigl(\sign(\cs_i) +\sign(\gamma_{i,i+1})\bigr),
\label{Fwlstar-a}
\ee
and their specialization at $c_i=\beta_{ni}$, which we denote by
\be
\label{Fwlstar}
\Fs_{n}(\{\gamma_i\})=\frac{1}{2^{n-1}}\prod_{i=1}^{n-1}\sgn(\Gamma_{ni}),
\quad
\tFs_{n}(\{\gamma_i\})=\frac{1}{2^{n-1}}
\prod_{i=1}^{n-1}\bigl(\sign(\Gamma_{ni}) +\sign(\gamma_{i,i+1})\bigr).
\ee
For $n=1$ all these objects are understood to be equal to one. For $n=2$ we have the vanishing property $\tFs_{2}=0$.
In terms of these notations one can give two iterative representations for the partial index:

\begin{proposition}\label{theoremF}
The partial tree index satisfies the following two recursion relations
\be
\begin{split}
\Ftr{n}(\{\gamma_i\},z)
=&\,\Fwl_n(\{\gamma_i,c_i\})- \sum_{n_1+\cdots +n_m= n\atop n_k\ge 1, \ m<n}
\Ftr{m}(\{\gamma'_k\},z)
\prod_{k=1}^m \Fs_{n_k}(\gamma_{j_{k-1}+1},\dots,\gamma_{j_{k}})
\\
=&\,\tFwl_n(\{\gamma_i,c_i\})- \sum_{n_1+\cdots +n_m= n\atop n_k\ge 1, \ m<n-1}
\Ftr{m}(\{\gamma'_k\},z)
\prod_{k=1}^m \tFs_{n_k}(\gamma_{j_{k-1}+1},\dots,\gamma_{j_{k}}),
\end{split}
\label{F-ansatz}
\ee
where the sum runs over ordered partitions of $n$, with largest size $n-1$ in the first relation,
or $n-2$ in the second.\footnote{The restriction $m<n-1$ in the sum in the second line can be relaxed to $m<n$.
The stronger inequality is then automatically fulfilled due to $\tFs_{2}=0$.}
For $k=1,\dots,m$, where $m$ is the number of parts, we defined
\be
j_0=0,
\qquad
j_k=n_1+\cdots + n_k,
\qquad
\gamma'_k=\gamma_{j_{k-1}+1}+\cdots +\gamma_{j_{k}}.
\label{groupindex}
\ee
\end{proposition}

Note that the two representations differ only by a redistribution of the function $\Fs_{2}$:
in the first one it contributes only to the sum over splittings of the set of charges into subsets,
where in particular it arises as a factor in the last product,
whereas in the second its effect is incorporated by the terms proportional
to $\sign(\gamma_{i,i+1})$ in \eqref{Fwlstar}.
The two representations may be useful for different purposes.
For instance, the second representation is suitable
for the proof of convergence of the BPS partition function, whereas the first representation
is more convenient for analyzing its modular properties \cite{ap-to-appear},
this is why we give here both of them. Since their proofs are essentially identical,
we only  present the proof for the first representation.

\begin{proof}
Our proof will be inductive. It starts from $n=2$ in which case the formula
\eqref{F-ansatz} coincides with the definition of $F_2$ (see \eqref{resF2}).
Let us now assume that it holds up to $n-1$.
Our aim is to show that the iterative equation \eqref{inductFn} reduces to \eqref{F-ansatz}.
Defining $x_k=\cs_k/\Gamma_{nk}$, we find that
\be
\sum_{i=1}^k c_i(z_\ell)=\Gamma_{nk}(x_k-x_\ell),
\qquad
\sum_{i=\ell+1}^k c_i(z_\ell)=\Gamma_{nk}(x_k-x_\ell).
\ee
Substituting \eqref{F-ansatz} and \eqref{Fwlstar-a} into the r.h.s. of \eqref{inductFn}
and denoting the second term in \eqref{F-ansatz}
by $\Fsp_n(\{\gamma_i\},z)$, one therefore gets
\be
\begin{split}
\Ftr{n}=&\,
\frac{1}{2^{n-1}}
\(\prod_{i=1}^n \sgn(\Gamma_{ni})\) \sum_{\ell=1}^{n-1} \bigl( \sgn(x_\ell)-1\bigr) \prod_{k=1\atop k\ne \ell}^{n-1}
\sgn(x_k-x_\ell)
\\
&\,
-\hf\sum_{\ell=1}^{n-1} \bigl( \sgn(\cs_\ell)-\sgn (\Gamma_{n\ell})\bigr)\[
\Fsp_\ell\Fwl_{n-\ell}+\Fwl_\ell\Fsp_{n-\ell}-\Fsp_\ell\Fsp_{n-\ell}
\]_{z\to z_\ell}
\\
=&\,
\frac{1}{2^{n-1}}\(\prod_{i=1}^n \sgn(\cs_i)-\prod_{i=1}^n \sgn(\Gamma_{ni})\),
\\
&\,
-\hf\sum_{\ell=1}^{n-1} \bigl( \sgn(\cs_\ell)-\sgn (\Gamma_{n\ell})\bigr)\[
\Fsp_\ell \Ftr{n-\ell}+\Ftr{\ell} \Fsp_{n-\ell}+\Fsp_\ell\Fsp_{n-\ell}
\]_{z\to z_\ell},
\end{split}
\label{deriv-iterFn}
\ee
where we have used the sign identity \eqref{signident3}.
In the first contribution we immediately recognize the difference $\Fwl_n-\Fs_n$,
whereas in the second contribution all terms can be combined into one sum over splittings
by adding the condition $\ell\in \{j_k\}$.
Denote the index $k$ for which this happens by $k_0$.
Then this contribution reads
\be
\begin{split}
&\,
-\hf\sum_{\ell=1}^{n-1} \bigl( \sgn(\cs_\ell)-\sgn (\Gamma_{n\ell})\bigr)
\sum_{n_1+\cdots +n_m= n\atop n_k\ge 1, \ m<n, \ \ell\in \{j_k\}} \Ftr{k_0}(z_\ell)
\Ftr{m-k_0}(z_\ell)\prod_{k=1}^m \Fs_{n_k}
\\
=&\,
-\sum_{n_1+\cdots +n_m= n\atop n_k\ge 1, \ 1<m<n}
\[\hf\sum_{k_0=1}^{m-1} \Bigl( \sgn(\cs_{j_{k_0}})-\sgn (\Gamma_{n j_{k_0}})\Bigr)
\Ftr{k_0}(z_\ell) \Ftr{m-k_0}(z_\ell)\]
\prod_{k=1}^m \Fs_{n_k}.
\end{split}
\label{inductF2}
\ee
Here we interchanged the two sums which allows to drop the condition $\ell\in \{j_k\}$, but adds the requirement $m>1$
(following from $\ell\in \{j_k\}$ in the previous representation).
In square brackets one recognizes the r.h.s. of \eqref{inductFn} with $n$ replaced by $m<n$.
Hence, it is subject to the induction hypothesis
which allows to replace this expression by $\Ftr{m}(\{\gamma'_k\},z)$ and shows that \eqref{inductF2} is equal to $-(\Fsp_n-\Fs_n)$
where the second term is due to the condition $m>1$ in the sum over splittings.
Combining the two contributions, one finds $\Fwl_n-\Fsp_n$ which is exactly the first line in the required formula \eqref{F-ansatz}.
The proof of the second line is similar using instead the sign identity
\eqref{mainsign}.
\end{proof}

It is easy to verify the consistency of the representation \eqref{F-ansatz}
with the primitive wall-crossing formula \eqref{primwc}. Since this essentially amounts
to performing the manipulations used in the above proof in reverse, we shall omit
this check, which is at any rate guaranteed by the previous results.
On the other hand, for convenience of the reader, we have collected in appendix \ref{ap-pind}
the explicit expressions for the partial indices $\Ftr{n}$ in the form \eqref{F-ansatz} up to $n=4$.
These expressions have been also checked by a direct recombination of signs in the original definition \eqref{defFpl}.

The representation \eqref{F-ansatz} solves the problems about the sum over attractor flow trees
formulated in section \ref{subsec-afc}. However, it raises a new question: how can one extract
the value of the tree index $\gtr$ at $y=1$? Since by construction $\gtr$ is a Laurent polynomial
in $y$, the limit $y\to 1$ is smooth, and therefore can be obtained by
applying l'H\^opital's rule to Eq. \eqref{gtF},
i.e. acting on the numerator and on the denominator by $( y\partial_y)^{n-1}$ before setting $y=1$.
This amounts to replacing the factor $\frac{y^{\sum_{i<j} \gamma_{\sigma(i)\sigma(j)}}}{(y-y^{-1})^{n-1}}$
inside the sum by $\frac{\[\sum_{i<j} \gamma_{\sigma(i)\sigma(j)}\]^{n-1}}{2^{n-1} (n-1)!}$.
Although this  answers the question in principle, it may be desirable to have an alternative
representation of the refined tree index where the monomials in $y$ are all combined into
 products of $\kappa$-factors. For $n=3$, the result \eqref{Z3charge} provides
 such a representation. Ater recognizing  the expression in the square brackets as the
 partial tree index $\Ftr{3}$ in \eqref{F3res}, it is natural to conjecture the following simple formula for any $n$,
\begin{conj}
The tree index defined in \eqref{defgtree} can be expressed  as
\be
\gtr(\{\gamma_i\},z,y) =
(-1)^{n-1}\, (n-1)!\,
\Sym\Bigl\{\Ftr{n}(\{\gamma_{i}\},z)\prod_{k=2}^{n}\kappa(\beta_{kk})\Bigr\},
\label{gtreen}
\ee
where $\beta_{kk}=\sum_{i=1}^{k-1} \gamma_{ik}$.
\end{conj}
We have checked this conjecture by hand for $n=4$
and on Mathematica for higher $n$.
Since already the $n=4$ case is quite involved, we do not provide here these computations.

\section{Attractor flows and quivers}
\label{sec-quiver}

In this section, we apply the flow tree formula to the context of quiver quantum mechanics,
which describes the interactions of a set of mutually non-local BPS dyons in gauge theories or string
theories with $\cN=2$ supersymmetry in four dimensions \cite{Denef:2002ru}. Using the well-known
relation between the Higgs branch of this system and the moduli space of stable quiver representations,
we obtain a  formula which expresses the Poincar\'e polynomial of this moduli space,
for arbitrary values of the stability conditions, in terms of Poincar\'e polynomials of quiver moduli space
with lower dimension vectors evaluated at their respective attractor points.

\subsection{A brief review of quiver quantum mechanics}

\lfig{An example of quiver with 4 nodes.}
{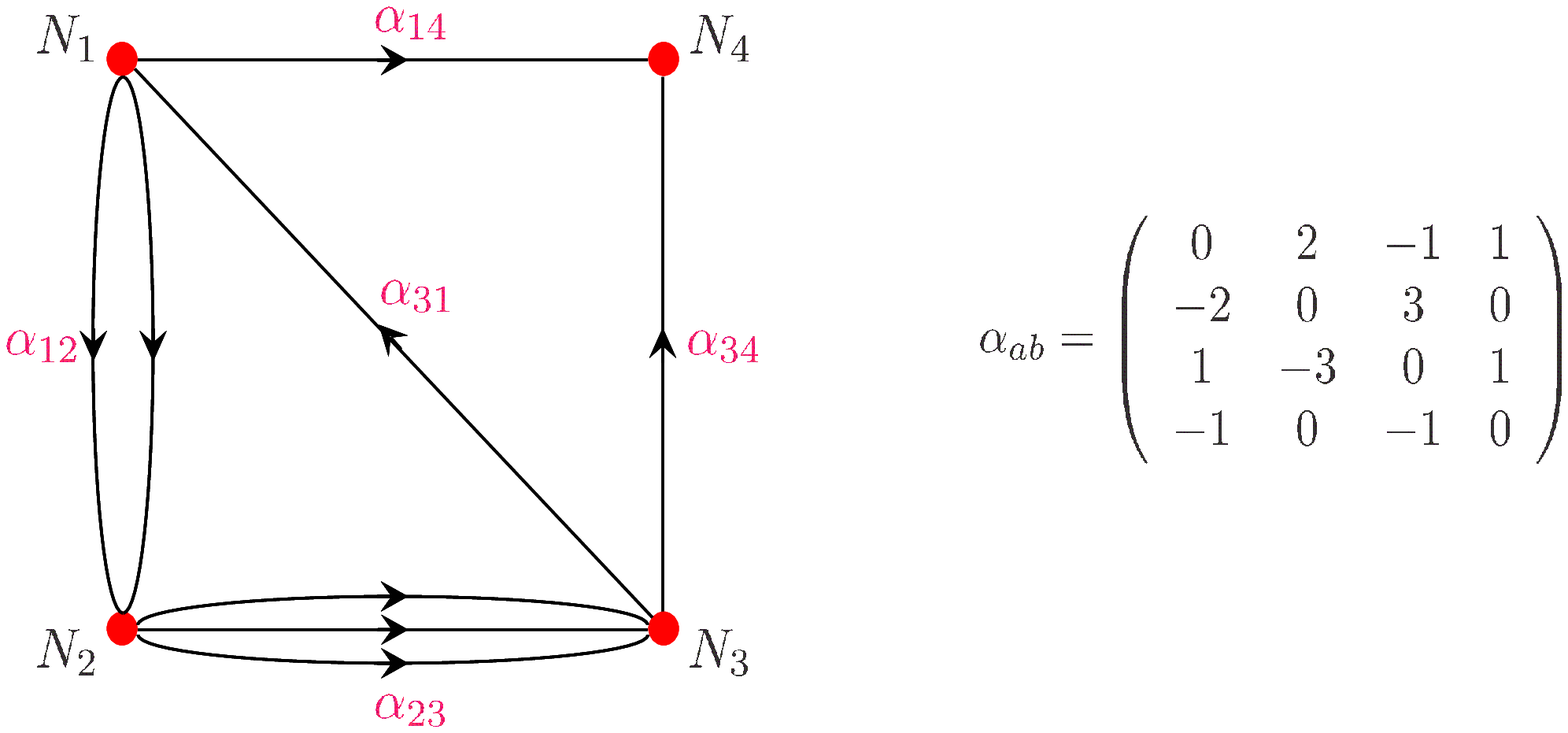}{15cm}{fig-quiver}{-1.cm}

Quiver quantum mechanics is a special class of 0+1 dimensional gauge theories with
four supercharges \cite{Denef:2002ru}. Its field content is encoded in
a $K\times K$ antisymmetric integer matrix $\alpha_{ab}$ and
a vector of positive integers $(N_1,\dots, N_K)$
known as the dimension vector. The model then includes vector multiplets for the
gauge group $G=\prod_{a=1}^K U(N_a)$ and $|\alpha_{ab}|$ chiral multiplets transforming
in the bifundamental representation $(N_a,\bar N_b)$ if $\alpha_{ab}>0$, or its complex
conjugate $(\bar N_a,N_b)$ if $\alpha_{ab}<0$ (note that we do not allow loops from
any node to itself). We shall denote the bosonic component
of these chiral multiplets by $\phi_{ab,A,ss'}$, where $1\leq A \leq|\alpha_{ab}|$,
$1\leq s\leq N_a$, $1\leq s'\leq N_b$. The field content is conveniently represented
by a quiver, i.e. a set of $K$ nodes associated to the $U(N_a)$ gauge groups, and $|\alpha_{ab}|$
arrows going from node $a$ to node $b$ if $\alpha_{ab}>0$, or in the opposite direction if $\alpha_{ab}<0$. The antisymmetric matrix
$\alpha_{ab}$ is the adjacency matrix of the graph formed by the nodes and arrows (see Fig. \ref{fig-quiver}).
When the quiver has oriented loops, the Lagrangian depends on a  superpotential $\cW(\phi)$, which
is a sum of $G$-invariant monomials in the chiral multiplets $\phi_{ab,A,ss'}$.
In addition, the Lagrangian depends
on a real vector $\zeta=(\zeta_1,\dots,\zeta_K)$, whose entries are known as the Fayet-Iliopoulos (FI) parameters and
associated to the $U(1)$ center in each gauge group $U(N_a)$. For the purpose of
counting BPS states, the overall scale of the $\zeta_a$'s is irrelevant, so this vector
can be viewed as a point in real projective space $\IR \IP^K$.

This supersymmetric quantum mechanics describes the interactions of a set of mutually non-local BPS dyons
consisting of $N_1$ dyons of charge $\alpha_1$, $N_2$ dyons of charge $\alpha_2$, etc.,
upon identifying $\alpha_{ab}=\langle \alpha_a,\alpha_b\rangle$, and fixing the FI parameters
as
\be
\label{ciZ}
\zeta_a(z) = \Im\[ Z_{\alpha_a} \bar Z_\gamma(z)\],
\ee
where $\gamma=\sum_{a=1}^K N_a \alpha_a$. Note that for an Abelian quiver ($N_a=1$),
this coincides with the definition of $c_i(z)$ in \eqref{notaion-c}.
The superpotential $\cW(\phi)$ (in case
the quiver has oriented loops) may have a complicated dependence on the moduli,
and we shall assume
that it is generic  in the sense that the Hessian at the critical points has maximal rank.
Moreover, we stress that the quiver description is only valid
in a region where the central charges $Z_{\alpha_a}$ are nearly aligned \cite{Denef:2007vg}.

It will be convenient to introduce the $K$-dimensional charge lattice
$\Lambda$ spanned by the basis vectors $\alpha_a$, and the cone $\Lambda^+$
spanned by linear combinations $\gamma=\sum_{a=1}^K n_a \alpha_a$ with non-negative
integer coefficients
(and not vanishing simultaneously). Note that $\Lambda^+$ admits a natural partial
order, $\gamma\leq \gamma'$ if $n_a\leq n'_a$ for all $1\leq a\leq K$.
For  any charge vector $\gamma=(n_1,\dots, n_K)\in\Lambda^+$, we denote $(\zeta,\gamma)=\sum_{a=1}^K n_a \zeta_a$.
For the special case of the
dimension vector $\gamma=(N_1,\dots, N_K)$, it follows from \eqref{ciZ} that $(\zeta,\gamma)=0$.
Since the only dependence on the
moduli $z$ arises through $\zeta_a(z)$,
we shall often use the symbol $z$ to denote the vector of FI parameters, without necessarily
implying that they originate from a central charge function $Z_{\alpha_a}(z)$ via \eqref{ciZ}.

Semi-classically, the quiver quantum mechanics admits two branches of supersymmetric vacua:
\begin{itemize}
\item the Higgs branch,
where the gauge symmetry is broken to the $U(1)$ center  by the vevs of the chiral multiplet
scalars $\phi_{ab,A,ss'}$;  the supersymmetric vacua are in one-to-one correspondence with
the set of stable orbits of the action of the
complexified gauge group $G_{\mathbb{C}}=\prod_{a=1}^K GL(N_a,\mathbb{C})$ restricted to the
critical locus of the superpotential $\cW(\phi)$, where the stability condition is
determined by the FI parameters. The set $\cM$ of supersymmetric vacua thus coincides with the moduli space
of stable quiver representations widely studied in  mathematics (see e.g. \cite{derksen2005quiver,reineke2008moduli}
for entry points in the vast literature on this subject).

\item the Coulomb branch, where the gauge symmetry is broken to $U(1)^{\sum_{a=1}^K N_a}$
and all chiral multiplets as well as off-diagonal vector multiplets are massive; the space of supersymmetric vacua
is then isomorphic to the phase space $\cM_n(\{\alpha_a^{N_a}\},z)$ which governs multi-centered
BPS solutions in $\cN=2$ supergravity, where $n=\sum_{a=1}^K N_a$ and $\alpha_a^{N_a}$
indicates $N_a$ copies of the vector $\alpha_a$.

\end{itemize}
Quantum mechanically, BPS states on the Higgs branch are harmonic forms on the moduli space
of quiver representations \cite{Denef:2002ru},
while BPS states on the Coulomb branch are harmonic spinors on
$\cM_n(\{\alpha_a^{N_a}\},z)$ \cite{deBoer:2008zn}. The group $SO(3)$ associated to
physical rotations in $\IR^3$ acts
on the  cohomology of the Higgs branch via the Lefschetz action generated by
contraction and wedge product with the K\"ahler form, while it acts
on the cohomology of the Coulomb branch by lifting the Hamiltonian action of \eqref{defJ}
on $\cM_n(\{\alpha_a^{N_a}\},z)$.

For the reasons explained below, the BPS index $\Omega(\gamma,z,y)$ of interest in this set-up
is the refined index of the BPS states on the Higgs branch.
Mathematically, it is defined as the Poincar\'e polynomial
of the moduli space $\cM=\cM(\gamma,\zeta)$ of quiver representations with dimension vector
$\gamma$ and stability conditions $\zeta_a$, rescaled by a factor $(-y)^{-d}$ where
$d$ is the complex dimension of $\cM$:
\be
\label{defpoinca}
\Omega(\gamma,z,y) = \sum_{p=0}^{2d} b_p(\cM)\, (-y)^{p-d} .
\ee
Here, $b_p(\cM)$ are the topological Betti numbers\footnote{It is also possible to define
a two-variable polynomial $\Omega(\gamma,z,y,t)$ which keeps track of the Hodge numbers
$h_{p,q}(\cM)$ and reduces to \eqref{defpoinca} at $t=1$, and to the $\chi_{y^2}$-genus
at $y=t$ \cite[\S 2.3]{Manschot:2012rx}. Here we set $t=1$ for simplicity, but the
flow tree formula has an immediate generalization to $t\neq 1$.}
of $\cM$ and $z$ stands for the set of FI parameters $\zeta_a$. In the
case where the dimension vector $\gamma$ is primitive and the superpotential $\cW$ is generic,
$\cM$ is compact, so $\Omega(\gamma,z,y)$ is a symmetric Laurent polynomial in $y$, which can
be viewed as the character of the Lefschetz action of $SO(3)$ on the moduli space $\cM$.
When $\gamma$ is not primitive, $\cM$ is no longer compact, but one can still define the
Poincar\'e polynomial using intersection cohomology. In that case, we define the rational
invariant $\bOm(\gamma,z,y)$ in the same way as in \eqref{defbOm}.

The simplest example is the Kronecker quiver, with two nodes of rank 1 and
$\alpha_{12}$ arrows from the first node to the second.
The Higgs branch is either empty when $\sign (\zeta_1)= - \sign(\alpha_{12})$,
or given by the complex projective space $\mathbb{P}^{|\alpha_{12}|-1}$.
Its Poincar\'e polynomial is given by the same formula \eqref{gt2}
which was arrived at by quantizing the Coulomb branch $\cM_2(\alpha_1,\alpha_2;z)$
with $\alpha_{12}=\langle \alpha_1,\alpha_2\rangle$.
This coincidence between the cohomology of the Coulomb and Higgs branches is in
fact a general property of quivers without oriented loops \cite{Manschot:2013sya},
and reflects the fact that the support of the BPS wave functions shifts from the
Higgs branch to the Coulomb branch as the string coupling is increased \cite{Denef:2002ru}.
For quivers with loops, the Coulomb branch is in general non-compact
so the corresponding index is ill-defined. In contrast, the Higgs branch
is compact for a generic choice of superpotential $\cW(\phi)$, at least for a
primitive dimension vector.
This is why we focus on the BPS index on this branch.

\subsection{Flow tree formula for quivers}

Since quiver quantum mechanics describes the dynamics of black hole bound states,
it is natural to expect that
the flow tree formula also applies in this  context, and
allows to express the total rational index $\bOm(\gamma,z,y)$ in terms of sums of monomials
in moduli-independent indices $\bOm_*(\gamma_i)$ associated to all decompositions $\gamma=\sum_{i=1}^n \gamma_i$
where the $\gamma_i$'s lie in the positive cone $\Lambda^+$. In order to formulate it
however, we  need to define the notions of `attractor flow' and  `attractor point' in the
context of quivers.

The first notion is obvious from the discussion in \S\ref{subsec-flow}:
the discrete attractor flow \eqref{cistar} only involves the parameters $c_i$
and the DSZ products $\gamma_{ij}$ associated to the constituents $\gamma_i$.
In our setup they can be evaluated in terms of the coefficients of $\gamma_i=\sum_{a=1}^K n_{ia} \alpha_a$
on the basis  $\alpha_a$ associated to the nodes of the quiver as
\be
\label{defgijci}
\gamma_{ij}=\sum_{a,b=1}^K n_{ia} n_{jb}\, \alpha_{ab},
\qquad
c_i=\sum_{a=1}^K n_{ia}\zeta_a .
\ee
Note that these quantities define an auxiliary Abelian quiver with $n$ nodes
associated to the constituents in the decomposition $\gamma=\sum_{i=1}^n \gamma_i$.
From this data, one can then  compute the tree index $\gtr$ by constructing all
stable attractor flow trees with $n$ leaves.
The latter are rooted unordered binary trees $T$, whose vertices $v\in V_T$
are decorated by charge vectors $\gamma_v=\sum_{i=1}^n \mm^i_v \gamma_i$ with $\mm^i_v\in\{0,1\}$,
such that the root carries charge $\gamma$,
the leaves carry charges $\gamma_1,\dots, \gamma_n$, and the charges add up
at each vertex, $\gamma_v=\gamma_{\Lv{v}}+\gamma_{\Rv{v}}$. For a given decoration,
we assign stability parameters $c_{v,i}$ at each vertex, equal to $c_i$ in \eqref{defgijci}
at the root and satisfying the `discrete attractor flow' relation \eqref{cistar} along each
edge,\footnote{Note that \eqref{cistar} assumed that the central charges were aligned on the
walls of marginal stability, as opposed to being anti-aligned, but this assumption is
automatically satisfied in the regime of validity of quiver quantum mechanics.}
\be
\label{cistarQ}
c_{v,i} =
c_{p(v),i} - \frac{\langle \gamma_v,\gamma_i \rangle}{\langle \gamma_v,\gamma_{\Lv{v}}\rangle}
\,\sum_{j=1}^n  \mm^j_{\Lv{v}}  c_{p(v),j}
=  c_{p(v),i} - \frac{\langle \gamma_v,\gamma_i \rangle}{\langle \gamma_v,\gamma_{\Rv{v}}\rangle}
\,\sum_{j=1}^n  \mm^j_{\Rv{v}} c_{p(v),j} .
\ee
The tree index is then obtained via (cf. \eqref{kappaT0}, \eqref{defgtree},
\eqref{kappaT})\footnote{While $\gtr$ only depends on vector $\gamma_i$ through
the DSZ matrix $\gamma_{ij}=\langle \gamma_i,\gamma_j\rangle$ and parameters $c_i$,
for clarity  we denote its arguments by $\{\gamma_i, c_i\}$.}
\be
\label{defgtreeQ}
\gtr(\{\gamma_i,c_i\},y)
=\frac{(-1)^{n-1} }{2^{n-1}} \sum_{T\in \cT_n}
\prod_{v\in V_T} \kappa( \gamma_{\Lv{v}\Rv{v}})
\left[\sign\left( \sum_{i=1}^n  \mm^i_{\Lv{v}} c_{v,i} \right)+
\sign(\gamma_{\Lv{v}\Rv{v}}) \right],
\ee
where we recall that $\kappa(x)$ is defined in \eqref{kappadef}.
In practice, it is easiest to generate the trees recursively, and discard those which
do not satisfy the stability condition at one vertex without exploring further splittings.
Alternatively, one may wish to use the representation \eqref{gtF} of $\gtr$
as a sum of `partial tree indices' $\Ftr{n}(\{\gamma_i,c_i\})$  defined by a sum over planar trees
(see \eqref{defFpl}), or the recursion formulae \eqref{inductFn} and \eqref{F-ansatz} of the
previous section (see also \eqref{gtreen} for a conjectural relation which does not require taking
$y\neq 1$).

As for the notion of attractor point, we observe that for a given dimension vector $\gamma=\sum_{a=1}^K N_a \alpha_a$,
the FI parameters defined by
\be
\label{attFI}
\zeta_{*,a}(\gamma) = \langle \gamma, \alpha_a\rangle = -\sum_{b=1}^K \alpha_{ab} N_b
\ee
are such that for any decomposition $\gamma=\gamma_L+\gamma_R$,
the sign of $\langle \gamma_L,\gamma_R\rangle$ is always opposite to the sign of
$(\zeta_*(\gamma),\gamma_L)$, mimicking the property \eqref{attrZZ} of the supergravity
attractor point  \cite{MPSunpublished}. Moreover, \eqref{attFI} automatically satisfies the condition
$(\zeta_*(\gamma'),\gamma')=0$, and is mapped to zero by the discrete attractor flow,
as expected since a single-centered black hole should be described by a single node quiver.
It is therefore natural to identify the attractor index
$\Omega_*(\gamma)$  with the Poincar\'e polynomial \eqref{defpoinca} evaluated for
the charge vector $\gamma$ and stability parameter $\zeta_*(\gamma)$. Of course, the
attractor point $\zeta_*(\gamma')$ can be defined in the same way for any vector
$\gamma'=\sum_{a=1}^K n_a \alpha_a\in\Lambda^+$, whether or not it coincides with the
dimension vector of the original quiver.

With these identifications, we can now state the flow tree formula for quiver moduli
in a mathematically precise way:

\begin{conj}
For a $K$-node quiver with adjacency matrix $\alpha_{ab}$, dimension vector $\gamma=(N_1,\dots, N_K)$,
stability parameters $\zeta=(\zeta_1,\dots, \zeta_K)$ and generic superpotential $\cW(\phi_{ab,A,ss'})$,
the rescaled Poincar\'e polynomial
\eqref{defpoinca} of the quiver moduli space is given by (cf. \eqref{defbOm})
\be
\label{defOmb}
\Omega(\gamma,z,y) =  \sum_{m|\gamma} \mu(m)\,
 \frac{y-1/y}{m(y^m-1/y^m)} \bOm(\gamma/m,z, y^m),
\ee
where $\mu(m)$ is the Moebius function and (cf. \eqref{Omsumtree})
\bea
\label{OmsumtreeQ}
\bOm(\gamma,z,y) =
\sum_{\gamma=\sum_{i=1}^n \gamma_i}
\frac{\gtr(\{\gamma_i,c_i\},y)}{|{\rm Aut}\{\gamma_i\}|}\,
\prod_{i=1}^n \bOm_*(\gamma_i,y).
\eea
Here, the sum runs over all distinct unordered splittings of $\gamma$ into sums of vectors
$\gamma_i=(n_{i1},\dots, n_{iK})$ with non-negative entries, $|{\rm Aut}\{\gamma_i\}|$ is the
order of the subgroup of the permutation group $S_n$ preserving the ordered set $\{\gamma_i\}$,
$\gtr(\{\gamma_i,c_i\},y)$ is the `tree index' defined using \eqref{defgijci}, 
and $\bOm_*(\gamma_i,y)$ are the rational `attractor indices', i.e. the same quantities 
as in \eqref{defOmb} but evaluated for the dimension vector $\gamma_i$ and stability parameters 
$\zeta_{*,a}(\gamma_i)=-\sum_{b=1}^K \alpha_{ab} n_{ib}$.
\end{conj}

This conjecture is easily proven for Abelian quivers, i.e. when all $N_i\leq 1$.
Indeed, \eqref{OmsumtreeQ} gives the correct value
at the attractor point, since $\gtr$ for $n>1$ vanishes by construction at that point,
and it also  satisfies the primitive wall-crossing formula,  as shown in \S\ref{subsec-afc}.
In order to prove that  \eqref{OmsumtreeQ} also holds for non-Abelian quivers, one would have to
prove that it satisfies the general Kontsevich--Soibelman wall-crossing formula \cite{ks}.
Unfortunately, we do not yet have a general proof of that fact.

For non-Abelian quivers without loops, however, we shall show in section \ref{sec-noloop} that
the tree index $\gtr$ coincides with the Coulomb index $\gref$, so that
the attractor indices $\Omega_*(\gamma_i)$ coincide with the single-centered indices
$\OmS(\gamma_i)$, and the flow tree formula becomes equivalent to the Coulomb branch formula
reviewed in the next section. Since the latter has been shown \cite{Manschot:2013sya} to be equivalent to
Reineke's formula \cite{1043.17010}  for non-Abelian quivers without loops,
it follows that \eqref{OmsumtreeQ} also holds in this case.
In particular, it must be consistent with the general wall-crossing formula.

For the most general case of non-Abelian quivers with loops, we do not have a direct proof
that \eqref{OmsumtreeQ} is consistent with the general wall-crossing formula, but it is physically clear
that the validity of the latter can only depend on general factorization properties of the tree
index, and so should not be sensitive to the existence of loops (indeed, loops are responsible
for existence of scaling solutions, but those are insensitive to wall-crossing). Hence we believe
that the conjecture is true also in this general case,
although this argument falls short of being a mathematical proof.

Alternatively, it may be possible
to derive the flow tree formula from the Joyce-Song formula \cite{Joyce:2008pc,Joyce:2009xv},
which relates the BPS indices
$\Omega(\gamma,z)$ and $\Omega(\gamma,z')$ in arbitrary distinct chambers. Applying
this formula for $z'=z_{\gamma}$, we obtain a sum of products of
$\Omega(\gamma_i,z_{\gamma})$ with $\gamma=\sum \gamma_i$, all evaluated at the
point $z_\gamma$. One can then repeat this process and express each of the
$\Omega(\gamma_i,z_{\gamma})$'s in terms of BPS indices at their
respective attractor points $z_{\gamma_i}$. Since the charge of the constituents $\gamma_i$ is always less than that
of the total charge, this process terminates after a finite number of steps. It is interesting
to note that the Joyce-Song formula also involves a sum over rooted decorated trees, albeit
of a somewhat different type.

\section{Comparing attractor and single-centered indices}
\label{sec-Coulomb}

In this final section, we compare the flow tree formula \eqref{Omsumtree}, expressing the total index
$\Omega(\gamma,z)$ in terms of attractor indices $\Omega_*(\gamma_i)$,
with the Coulomb branch formula which expresses the same index in terms of single-centered indices
$\Omega_S(\gamma_i)$. In \S\ref{sec-gC}, we briefly recall the
statement of the Coulomb branch formula and the definition of  the `Coulomb index'
$\gref$ which plays a central r\^ole in it. In \S\ref{sec-noloop} we show that the latter
coincides with the tree index $\gtr$ for charge configurations
which do not allow scaling solutions, and conclude that $\Omega_*(\gamma)=\Omega_S(\gamma)$
for quivers without oriented loops. In the case where scaling solutions
are allowed, corresponding to quivers with oriented loops, $\gref$ and
$\gtr$ in general differ, and so do $\Omega_*(\gamma)$ and $\Omega_S(\gamma)$.
While we do not yet know how to relate them explicitly in  general, in \S\ref{sec-loop}
we work out their relation in the special cases of 3-centered and 4-centered configurations.

\subsection{Review of the Coulomb branch formula}
\label{sec-gC}

The Coulomb branch formula conjecturally expresses the total rational index $\bOm(\gamma,z)$
defined in \eqref{defbOm} in terms of single-centered indices
$\bOm_S(\gamma_i)$ as follows:
\begin{conj}[\cite{Manschot:2011xc,Manschot:2013sya,Manschot:2014fua}]
\be
\label{CoulombForm1}
\begin{split}
\bOm(\gamma,z,y) =  \sum_{\gamma=\sum_{i=1}^n\gamma_i}
\frac{\gref(\{\gamma_i\},z,y)}{|{\rm Aut}\{\gamma_i\}|}
\prod_{i=1}^n
\left\{
\sum_{m_i\in\IZ \atop m_i \vert \gamma_i}
\frac{y-1/y}{m_i(y^{m_i}-y^{-m_i})}
\Omega_{\rm tot}(\gamma_i/m_i,y^{m_i}) \right\},
\end{split}
\ee
where $\Omega_{\rm tot}$ is determined in terms of $\Omega_S$ via
\be
\label{CoulombForm2}
\Omega_{\rm tot}(\gamma,y)=\Omega_S(\gamma,y)+
\sum_{\gamma=\sum_{i=1}^m m_i \beta_i} H_m(\{\beta_i,m_i\},y)\prod_{i=1}^m \Omega_S(\beta_i,y^{m_i}).
\ee
\end{conj}
In both \eqref{CoulombForm1} and \eqref{CoulombForm2} the sums run over unordered decompositions of $\gamma$
into sums of positive vectors $\gamma_i$ or $\{\beta_i,m_i\}$ with $m_i\geq 1$.
The functions $H_m(\{\beta_i,m_i\},y)$ are determined recursively
by the so called ``minimal modification hypothesis" (see  \cite{Manschot:2013sya,Manschot:2014fua} for details)
and their role is to ensure that the full refined index $\bOm(\gamma,z)$ is a symmetric Laurent polynomial in $y$.
The function $\gref(\{\gamma_i\},z,y)$, known as the Coulomb index,
is the only quantity on the r.h.s. of \eqref{CoulombForm1} which depends on the moduli $z$.
It is defined as the equivariant index of the Dirac operator on $\cM_n(\{\gamma_i\})$,
computed by localization with respect to rotations around a fixed axis. The fixed points
of the action of $J_3$ on $\cM_n(\{\gamma_i\})$ are collinear solutions to the equations \eqref{DenefEq},
classified by the order of the centers along the axis. Equivalently, they are critical points of the
potential\footnote{It is worth noting that at the attractor point \eqref{attFI},
the potential \eqref{ecou0} becomes a sum of pairwise interactions,
$V(\{x_i\})=- \sum_{i<j} \gamma_{ij} V(x_j-x_i)$ with $V(x)=  \sign(x)\,  \log |x| +x$, which
is attractive for $x>0$ and repulsive for $x<0$.}
\be \label{ecou0}
V(\{x_i\})=\sum_{i<j}
\gamma_{ij} \,\sign(x_i-x_j) \log |x_i-x_j| - \sum_{i=1}^n c_i \,x_i\, .
\ee
The sum over
fixed points may be represented as a sum over all permutations $\sigma$ of $\{1,2,\dots n\}$,
\be
\label{gCFy}
\gref(\{\gamma_i\}, z,y)
= \frac{(-1)^{n-1+\sum_{i<j} \gamma_{ij}}}{(y-y^{-1})^{n-1}}
\sum_{\sigma\in S_n}
\FC(\{\gamma_{\sigma(i)}\},z)\,
y^{\sum_{i<j} \gamma_{\sigma(i)\sigma(j)}},
\ee
where the `partial Coulomb index' $\FC(\{\gamma_i\},z)\in\IZ$ counts the critical points
for a fixed ordering $x_1<x_2\dots <x_n$ along the axis, weighted by the sign of the
Hessian of $V$ after removing the trivial translational zero-mode.
Under the reversal symmetry $i \mapsto n+1-i$,
$\FC$ picks up a sign $(-1)^{n-1}$, so that \eqref{gCFy} is invariant
under $y\to 1/y$.

While the computation of the one-dimensional solutions to \eqref{DenefEq}
becomes quickly impractical as $n$ increases,
it was shown in  \cite{Manschot:2013sya} that the partial tree index
$\FC$ can be efficiently
evaluated by first rescaling $\gamma_{ij}$ by $\lambda \gamma_{ij}$ unless $|i-j|=1$,
and then dialing $\lambda$ from $0$ (where only  nearest neighbor interactions are kept) to
$\lambda=1$ (the configuration of interest).\footnote{This prescription assumes that
the initial matrix $\gamma_{ij}$ is generic, in the sense
explained in footnote \ref{foogen} on page \pageref{foogen}. For non-generic cases, one
should first perturb the $\gamma_{ij}$'s
such that they become generic, apply the previous prescription
and then take the limit where  the relevant $\gamma_{ij}$'s become zero. The value of the individual
partial Coulomb indices $\FC$ may depend on the choice of deformation, but
after summing over all permutations, the limit is independent of that choice \cite{Manschot:2013sya}.}
The value at $\lambda=0$ turns out to coincide with the function $\tFwl_n(\{\gamma_i,c_i\})$
already introduced in \eqref{Fwlstar-a}.\footnote{In order
to match this result with  \cite[(2.9)]{Manschot:2013sya}, observe that $\Theta(xy) (-1)^{\Theta(-x)}=\frac12(\sign (x)+\sign (y))$,
where $\Theta(x)$ is the Heaviside step function, and $x,y$
are non-zero real numbers. Similarly, to  make contact between \eqref{FCrecur} below and \cite[(2.32)]{Manschot:2013sya}, 
note that $\sum_{k\leq i,j\leq \ell \atop i\leq j-2} \gamma_{ij}$ and $\sum_{k\leq i<j\leq \ell} \gamma_{ij}$
necessarily have the same sign whenever $\sum_{k\leq i<j\leq \ell} \gamma_{ij} \times
\sum_{i=k}^{\ell-1} \gamma_{i,i+1} <0$.}
In the absence of scaling solutions (or for quivers without loops), 
one can show that no jumps occur as $\lambda$ increases
from 0 to 1, so that $\FC$ is still given by its value at $\lambda=0$,
\be
\label{partialCoulomb}
\FC(\{\gamma_i\},z) = \tFwl_n(\{\gamma_i,c_i\}).
\ee
In the presence of scaling solutions, the partial Coulomb index $\FC$
may jump several times as the deformation parameter $\lambda$ is varied from 0 to 1.
These jumps occur whenever the parameter
$\lambda=\lambda_{k,\ell}$
is such that $\sum_{i,j\in A} \gamma_{ij}(\lambda_{k,\ell})=0$ for some subset $A=\{k,k+1,\dots \ell\}$ of $\{1,\dots, n\}$,
corresponding to all consecutive centers in $A$ colliding at one point.
$\FC$ is then given  by a sum over all possible jumps,
\be
\begin{split}
\FC(\{\gamma_i\},z) = &\,\tFwl_n(\{\gamma_i,c_i\})+   \sum_{1\leq k<\ell\leq n}  \frac12\left[
\sign\left( \sum_{k\leq i<j\leq \ell} \gamma_{ij}\right)
-\sign\left(\sum_{i=k}^{\ell-1} \gamma_{i,i+1}\right) \right]
\\
&\,  \times G_{\ell-k+1}(\gamma'_k,\dots,\gamma'_\ell)\,
\FCn{n+k-\ell}(\{\gamma'_1,\dots\gamma'_{k-1},\gamma'_{k+\cdots+\ell},
\gamma'_{\ell+1},\dots, \gamma'_n\},z ),
\vphantom{\biggr[}
\end{split}
\label{FCrecur}
\ee
where $\gamma'_k$ are charge vectors with deformed inner product $\langle \gamma'_i,\gamma'_j\rangle = \gamma_{ij}(\lambda_{k,l})$,
and $G_n(\{\gamma_i\})$ is
the Coulomb index for colliding solutions. The latter vanishes for $n<3$ and satisfies its own
recursion relation (see \cite[\S 2.3.2]{Manschot:2013sya}), initialized with the result for $n=3$,
\be
G_3(\gamma_1,\gamma_2,\gamma_3) = \frac12\,\bigl(\sign \gamma_{12} + \sign \gamma_{23}\bigr).
\ee

Using this procedure, we can compute the Coulomb index $\gref(\{\gamma_i\},z,y)$ for an arbitrary collection
of non-zero vectors $\gamma_i$ and generic stability parameters $c_i$ (collectively denoted by $z$).
When the charges $\gamma_i$ are such that no scaling solution is allowed, the result
is the equivariant Dirac index of the compact phase space $\cM_n(\{\gamma_i\},z)$ \cite{Manschot:2011xc,Kim:2011sc},
and is therefore a symmetric Laurent polynomial in $y$ with integer coefficients. In the
presence of scaling solutions however, the phase space $\cM_n$ is non-compact
and the above definition of $\gref$ produces instead a rational function of $y$. While it might in principle
be possible to construct a compactification of $\cM_n$ and incorporate additional
fixed points from boundary components in order to produce a Laurent polynomial in $y$,
the Coulomb branch formula uses  the rational function $\gref$ as defined above,
but requires adjusting the functions
$H_m(\{\beta_i,m_i\},y)$ in such a way that the full index $\bOm(\gamma,z)$ obtained
via \eqref{CoulombForm1} is a symmetric Laurent polynomial, provided  the single-centered
indices $\Omega_S(\gamma_i)$ are. The minimal modification hypothesis of \cite{Manschot:2013sya}
gives a unique prescription for computing $H_m$,
based on the assumption that the missing contributions from the boundary of $\cM_n$ carry the minimal
possible angular momentum. Note that this prescription does not take into account the condition
of absence of closed timelike curves, which is presumably irrelevant in the context of quiver
quantum mechanics, but needs to be checked by hand in more general cases
(see e.g. \cite[\S 3.2]{Manschot:2011xc} for an example where this condition makes an important difference).

It follows from the results in \cite{Manschot:2010qz} that the formula \eqref{CoulombForm1} is
consistent with the general wall-crossing formula of \cite{ks,Joyce:2008pc}.
 In  cases where none of the decompositions $\gamma=\sum\gamma_i$ allow for scaling solutions,
relevant for quivers with no loops,
all the factors $H_m$ in \eqref{CoulombForm2} vanish, so the Coulomb branch formula reduces to
\be
\label{CoulombForm3}
\bOm(\gamma,z,y)=
\sum_{\gamma=\sum_{i=1}^n \gamma_i}\,
\frac{g_{C}(\{\gamma_i\},z,y)}{|{\rm Aut}\{\gamma_i\}|}\, \prod_{i=1}^n  \bOm_S(\gamma_i,y),
\ee
which closely resembles \eqref{Omsumgen}.
Indeed, we shall see in \S\ref{sec-noloop} than in this simplified case, the single-centered
indices $\bOm_S(\gamma_i)$
agree with the attractor indices $\bOm_*(\gamma_i)$.

A different simplification occurs when
$\gamma$ is primitive and such that all charge vectors $\gamma_i$ appearing in each decomposition
$\gamma=\sum \gamma_i$ are distinct and primitive. In this case, relevant for Abelian quivers, the Coulomb branch formula
\eqref{CoulombForm1} simplifies to
\be
\label{CoulombForm3a}
\begin{split}
\Omega(\gamma,z,y) = \sum_{\gamma=\sum_{i=1}^n\gamma_i}\!\!\!
\gref(\{\gamma_i\},z,y)\,
\prod_{i=1}^n
\left\{
\Omega_S(\gamma_i,y)
+\!\!\!\!\sum_{\sum_{j=1}^{m_i} \beta_j=\gamma_i}  \!\!\!\!
H_{m_i}(\{\beta_j\},y)\, \prod_{j=1}^{m_i} \OmS(\beta_j,y)  \right\}.
\end{split}
\ee
In this case, the rational functions $H_m(\{\beta_j\},y)$ are fixed by demanding that the coefficient
of the monomial $\prod_{j=1}^m \OmS(\beta_j,y)$ in $\Omega(\gamma,z,y)$ be a Laurent polynomial in $y$.
Requiring that $H_m(\{\beta_j\},y)$ are invariant under $y\to 1/y$ and vanish at $y=\infty$ fixes
them uniquely \cite{Manschot:2011xc}. The reason why \eqref{CoulombForm3a} differs
from \eqref{Omsumgen} is that the attractor indices $\Omega_*(\gamma_i)$ include contributions
both from single-centered black holes and scaling solutions.

\subsection{Quivers without loops \label{sec-noloop}}

In this subsection, we shall show that for quivers without loops, such that no scaling solutions
are allowed, the tree index  $\gtr$ and Coulomb index $\gref$
coincide for any set of charges and moduli. This will turn out to imply that
the attractor and single-centered indices, $\Omega_*(\gamma_i)$ and $\OmS(\gamma_i)$, also coincide for all $\gamma_i$.

In order to show the equality $\gtr=\gref$, the main observation is that the
corresponding partial indices $\Ftr{n}$ and $F_{C,n}$ are locally constant functions of
the parameters $c_i$ whose only discontinuities lie on the walls of marginal stability where $\sum_{i=1}^k c_i=0$,
where they both jump according to  the primitive wall-crossing formula.
Therefore, it suffices to show that they coincide at one value of the $c_i's$.

A convenient choice is to take the attractor point \eqref{attFI} which in our case is $c_i^*(\gamma) =\beta_{ni}$.
It is an immediate consequence of the iterative equation \eqref{inductFn}
that at this point the partial tree index $\Ftr{n}$ vanishes.
On the other hand, the partial Coulomb index \eqref{partialCoulomb} reduces to the function
$\tFs_{n}(\{\gamma_i\})$
defined in \eqref{Fwlstar}.
We shall now show that this function vanishes whenever  $\gamma_{ij}$ is the adjacency matrix
of a generic $n$-node quiver without loops. The result for non-generic matrices $\gamma_{ij}$
then follows by continuity.

First, note that there exists $n!$ different choices of signs for $\gamma_{ij}$ (out of $2^{n(n-1)/2}$)
such that the quiver has no oriented loops, and all those choices are related by permutations
(indeed, any generic quiver defines a total order on $[1,n]$). For each of these choices of signs,
there exists a unique source $s$ and sink $t$. Consider the restricted quiver obtained by keeping only
nearest neighbor interactions. The restricted quiver has a set of sources $\{s_i\}$ and sinks $\{t_i\}$
which lie either at the endpoints $v_1$, $v_{r+1}$, or at the points $v_r$ where $\gamma_{r,r+1}$ changes sign.
Obviously, $s\in\{s_i\}$ and $t\in\{t_i\}$. Now, assume that $\tFs_{n}(\{\gamma_i\})$ was non-zero.
This means that for each $r$ in $[1,n-1]$, the sign of $\gamma_{r,r+1}$ is the same as the sign
of $\Gamma_{nr}$. Thus, whenever $\gamma_{r,r+1}$ changes sign, so does $\Gamma_{nr}$.
The key observation is the following: whenever $\Gamma_{nr}$ and $\Gamma_{n,r+1}$ have opposite sign,
then the sign of $\Gamma_{nr}$ is opposite to the sign of $\Gamma_{n,r+1}-\Gamma_{nr}$,
which is equal to $\sum_i \gamma_{i,r+1}$. Thus, if $v_{r+1}$ is a sink of the restricted quiver
distinct from the endpoints, i.e. $\gamma_{r,r+1}>0$, $\gamma_{r+1,r+2}<0$, then $\Gamma_{nr}>0$,
hence $\sum_i \gamma_{i,r+1}<0$, which shows that $v_{r+1}$ cannot be a sink of the full quiver.
Similarly, if  $v_{r+1}$ is a source of the restricted quiver
distinct from the endpoints, then it cannot be a source of the full quiver.
Thus, the source and sink of the full quiver must be the endpoints.
But again, $v_1$ cannot be the source nor the sink, since $\gamma_{12}$ and $\Gamma_{n1}$ have the same sign.
Similarly, $v_{n+1}$ cannot be the source
nor the sink. Thus, we have reached a contradiction with the hypothesis that $\tFs_{n}(\{\gamma_i\})$ was non-zero.
Therefore, the partial Coulomb index at the attractor point vanishes for any generic quiver without oriented loops.
As a result, we conclude that $\Ftr{n}=\FC$, and therefore $\gtr=\gref$. Since both $\gtr$ and $\gref$
are continuous functions of $\gamma_{ij}$ (for generic, fixed values of $c_i$),
this equality continues to hold even if $\gamma_{ij}$ is not generic.

Having established that $\gtr=\gref$, we can now apply the Coulomb branch formula
\eqref{CoulombForm3} for $z$ equal to $z_{\gamma}$. Since $\gtr$ (and therefore
$\gref$) vanishes at that point whenever $n\geq 2$, it immediately follows that
$\Omega_*(\gamma)=\OmS(\gamma)$ for any dimension vector. The flow tree formula
\eqref{Omsumtree} is therefore equivalent to the  Coulomb branch formula \eqref{CoulombForm3},
which is known to agree with Reineke's formula for quivers without loops provided
$\OmS(\gamma)=1$ whenever $\gamma$ is one of the basis vectors $\alpha_1,\dots, \alpha_K$
and zero otherwise \cite{Manschot:2013sya}. We conclude
that the flow tree formula holds for quivers without loops.

\subsection{Abelian quivers with loops \label{sec-loop}}

Let us now allow for loops in the quiver diagram, and hence for the presence of scaling solutions, restricting to the Abelian case.
This corresponds to the inclusion of the second term in the partial Coulomb index \eqref{FCrecur} and,
given the results of the previous subsection, of the second term in the second representation
of the partial tree index given in \eqref{F-ansatz}. Besides, we have to take into account the contributions
to the BPS index generated by the functions $H_m(\{\beta_j\},y)$ and captured by the formula \eqref{CoulombForm3a}.
Unfortunately, we have not yet been able to find general relation between the attractor and single-centered indices
following from equating the two expansions. Below we provide explicit results  for
 two particular cases: $n=3$ and $n=4$.

\subsubsection{Three centers}

For $n=3$, the partial tree index can be found in \eqref{F3res}, whereas
the partial Coulomb index is given in \eqref{FC3}.
In the latter formula, the term in the first line is the result for a quiver with nearest-neighbor interactions,
while the second line arises from contributions of scaling solutions when the parameter $\lambda$
is changed from $\lambda=0$ to $\lambda=1$. It is easy to check that the second line vanishes unless
$\gamma_{12}, \gamma_{23}, \gamma_{31}$ all have the same sign (so that the quiver has an oriented loop)
and satisfy $|\gamma_{12}|, |\gamma_{23}| < |\gamma_{31}| < |\gamma_{12}| +  |\gamma_{23}|$,
in which case it equals $-1$. The latter condition implies that $|\gamma_{12}|,
|\gamma_{23}|, |\gamma_{31}|$ satisfy the triangular inequalities (so that scaling solutions are allowed),
the restriction that $ |\gamma_{31}|$ be the largest of all three being due to the special
choice of ordering.

The matching of the first terms in \eqref{F3res} and \eqref{FC3} is just the statement \eqref{partialCoulomb}.
The second terms are however different. Using the sign identity \eqref{signprop-ap}
repeatedly, one can rewrite the difference between the two partial indices as
\be
\FCn{3}-\Ftr{3}=\frac14\,\bigl( \sign(\gamma_{12} +\gamma_{23}+\gamma_{13})-\sign(\gamma_{1+2,3}) \bigr)
\bigl( \sign(\gamma_{12} +\gamma_{23}+\gamma_{13})-\sign(\gamma_{1,2+3}) \bigr)\, ,
\label{diffF3}
\ee
which is moduli-independent and vanishes unless $\gamma_{12}, \gamma_{23}, \gamma_{31}$ all
have the same sign and satisfy $|\gamma_{12}|, |\gamma_{23}| < |\gamma_{31}| < |\gamma_{12}| +  |\gamma_{23}|$.
The moduli-independence is a  consequence of the fact that
$\FCn{3}$ and $\Ftr{3}$ have the same discontinuities across the walls of marginal stability $c_1=0$ and $c_3=0$.
Because $\Ftr{3}$ vanishes at the attractor point,
the conditions for non-vanishing of the difference are, of course, the same which ensure the non-vanishing of $\FCn{3}$ itself.

The difference between the Coulomb and tree indices, $\gref$ and $\gtr$, can be obtained by summing
\eqref{diffF3} over permutations. If, for example, $\gamma_{12}>\gamma_{23}>\gamma_{31}>0$,
which is a configuration allowing scaling solutions, then the difference reads
\be
\gref - \gtr
=  (-1)^{\gamma_{12}+\gamma_{23}+\gamma_{31}}\,\frac{y^{\gamma_{23}+\gamma_{31}-\gamma_{12}}
+y^{-\gamma_{23}-\gamma_{31}+\gamma_{12}}}{(y-1/y)^2}\, .
\label{ex-diffg3}
\ee
Note that in this case $\gref$ is a rational function with a double pole
at $y=1$, while $\gtr$ is always a symmetric Laurent polynomial in $y$.

We can now  relate the single-centered invariant $\Omega_S(\gamma)$ to the
attractor index $\Omega_*(\gamma)$. For this purpose, we compare the Coulomb branch formula
\eqref{CoulombForm3a} with \eqref{Omsumtree} for $\gamma=\gamma_1+\gamma_2+\gamma_3$.
Using the fact that $\Omega_*(\gamma_i)=\Omega_S(\gamma_i)$ for basis vectors,
$\Omega_*(\gamma_{i+j})=\Omega_S(\gamma_{i+j})$ for sums of two basis vectors,
and $\gref(\gamma_L,\gamma_R)=\gtr(\gamma_L,\gamma_R)$ for any
pairs of vectors
$\gamma_L,\gamma_R$, we conclude that
\be
\begin{split}
\Omega_*(\gamma) = \OmS(\gamma)  +
\Bigl[\gref(\gamma_1,\gamma_2,\gamma_3;z)
- \gtr (\gamma_1,\gamma_2,\gamma_3;z)
 + H_3(\gamma_1,\gamma_2,\gamma_3)\Bigr]  \prod_{i=1}^3\Omega_S(\gamma_i).
\end{split}
\label{diffOm3}
\ee
Applying the minimal modification hypothesis to determine $H_3(\gamma_1,\gamma_2,\gamma_3)$
\cite{Manschot:2013sya}, for the case considered in \eqref{ex-diffg3}
we arrive at
\be
\begin{split}
\Omega_*(\gamma) = \OmS(\gamma) +
 \kappa\(\half\,(\gamma_{23}+\gamma_{31}-\gamma_{12}+\eps)\)
\kappa\(\half\,(\gamma_{23}+\gamma_{31}-\gamma_{12}-\eps)\)
\prod_{i=1}^3\Omega_S(\gamma_i),
\end{split}
\ee
where $\eps$ is the parity (0 or 1) of ${\gamma_{12}+\gamma_{23}+\gamma_{31}}$.
By construction, the difference is a symmetric Laurent polynomial.

\subsubsection{Four centers}

For $n=4$, the two partial indices are given in \eqref{F4res} and \eqref{FC4res}.
While the first terms in each expression coincide, the remaining terms are quite different,
and  the moduli-dependence does not cancel in their difference $\Ftr{4}-\FCn{4}$.
This is in fact consistent with the structure of the two expansions of the BPS index.
Indeed, comparing \eqref{CoulombForm3a} and \eqref{Omsumtree} and taking into account \eqref{diffOm3},
one finds
\be
\begin{split}
\Omega_*(\gamma) =&\, \Omega_S(\gamma)+ \Bigl[\gref(\gamma_1,\gamma_2,\gamma_3,\gamma_4;z) -
\gtr(\gamma_1,\gamma_2,\gamma_3,\gamma_4;z)+ H_4(\gamma_1,\gamma_2,\gamma_3,\gamma_4)
\\
+&\,
\Bigl( \gtr(\gamma_1,\gamma_{2+3+4};z) \bigl( \gtr(\gamma_2,\gamma_3,\gamma_4;z) -
\gref(\gamma_2,\gamma_3,\gamma_4;z) \bigr) + {\rm perm} \Bigr)\Bigr]\prod_{i=1}^4\Omega_S(\gamma_i).
\end{split}
\label{diffOm4}
\ee
Thus, the moduli-dependence of the difference of the two indices at $n=4$
must be non-trivial, so as to cancel  the moduli-dependence in the second line.
To check that this is indeed the case, note that the relation \eqref{diffOm4} can be rewritten as
\be
\Omega_*(\gamma) = \Omega_S(\gamma)
+\[H_4(\gamma_1,\gamma_2,\gamma_3,\gamma_4)
-\frac{3\,(-1)^{\sum_{i<j} \gamma_{ij}}}{(y-y^{-1})^{3}}
\Sym\Bigl\{
D_4(\{\gamma_i\})\,
y^{\sum_{i<j} \gamma_{ij}}\Bigr\}
\]\prod_{i=1}^4\Omega_S(\gamma_i),
\label{diffOm4a}
\ee
where
\be
\begin{split}
D_4=&\,
8\Bigl[\FCn{4}(\gamma_1,\gamma_2,\gamma_3,\gamma_4)
-\Ftr{4}(\gamma_1,\gamma_2,\gamma_3,\gamma_4)
\\
&\,
+\Ftr{2}(\gamma_1,\gamma_{2+3+4})\bigl(\Ftr{3}(\gamma_2,\gamma_3,\gamma_4)-\FCn{3}(\gamma_2,\gamma_3,\gamma_4)\bigr)
\\
&\,
+\Ftr{2}(\gamma_{1+2+3},\gamma_4)\bigl(\Ftr{3}(\gamma_1,\gamma_2,\gamma_3)-\FCn{3}(\gamma_1,\gamma_2,\gamma_3)\bigr)\Bigr].
\end{split}
\ee
Evaluating this combination using in particular \eqref{diffF3}, one arrives at
\bea
D_4&=&\sign\(\sum_{i<j}\gamma_{ij}\)\Bigl[\bigl(\sign (\gamma_{12}) + \sign (\gamma_{23})\bigr)\sign(\gamma_{12}+\gamma_{23}+\gamma_{13})
\Bigr.
\nn\\
&&\Bigl.\qquad
+\(\sign\gamma_{23}+\sign\gamma_{34}\)\sign(\gamma_{23}+\gamma_{24}+\gamma_{34})
+\sign(\gamma_{12})\sign(\gamma_{34})-1\Bigr]
\nn\\
&&
-\sign\(\gamma_{12}+\gamma_{23}+\gamma_{34}\)\Bigl[\sign(\gamma_{12})\sign(\gamma_{23})
+\sign(\gamma_{23})\sign(\gamma_{34})
+\sign(\gamma_{12})\sign(\gamma_{34})+1\Bigr]
\nn\\
&&
-\bigl(\sign(\beta_{41})+\sign(\gamma_{12})\bigr)\bigl(\sign(\beta_{4,1+2})+\sign(\gamma_{23})\bigr)
\bigl(\sign(\beta_{44})-\sign(\gamma_{34})\bigr)
\nn\\
&&
-\sign(\beta_{41})\bigl(\sign(\gamma_{23}+\gamma_{34}+\gamma_{24})-\sign(\gamma_{23})\bigr)
\bigl(\sign(\gamma_{23}+\gamma_{34}+\gamma_{24})-\sign(\gamma_{34})\bigr)
\nn\\
&&
-\sign(\beta_{44})\bigl(\sign(\gamma_{12}+\gamma_{23}+\gamma_{13})-\sign(\gamma_{12})\bigr)
\bigl(\sign(\gamma_{12}+\gamma_{23}+\gamma_{13})-\sign(\gamma_{23})\bigr).
\label{diffOmm4}
\eea
As expected, the result is moduli-independent, as required for the consistency of \eqref{diffOm4a}.
Moreover,  $D_4$ vanishes if $\gamma_{ij}$  is the adjacency matrix
of a generic quiver without loops.
Unfortunately, the result \eqref{diffOmm4} does not immediately suggest a generalization to $n\geq 5$.

\acknowledgments

The authors are grateful to Jan Manschot for collaboration at an initial stage of this project, and
to Sibasish Banerjee, Frederik Denef, Greg Moore and Ashoke Sen for discussions. The hospitality and financial support of the
Theoretical Physics Department of CERN, where this work was initiated,  is also gratefully
acknowledged.

\appendix

\section{Sign identities}
\label{ap-signs}

In this appendix we collect several sign identities which are used in the main text.
The basic identity, which is used in most manipulations, is
\be
(\sgn(x_1)-\sgn(x_2))\,\sgn(x_1-x_2)=1-\sgn(x_1)\,\sgn(x_2).
\label{signprop-ap}
\ee
Its validity follows from that its two sides are locally constant functions
with the same discontinuities and having the same value for, say, positive $x_i$.

Next, let us prove the following identity\footnote{If $y_i=\sign(z_i)$, the r.h.s. can be rewritten as
$ \prod_{i=1}^{n-1} \(\sgn(z_i)+\sgn(z_{i+1})\)$.}
\be
\sum_{j=1}^n \prod_{i=1\atop i\ne j}^n \(\sgn(x_i-x_j)+y_i\)
= \hf\,\prod_{i=1}^{n} \(y_i+1\)-\hf\,\prod_{i=1}^{n} \(y_i-1\).
\label{signident4}
\ee
We proceed by induction. For $n=2$ it trivially holds.
Assuming that it holds for $n-1$, we find the smallest $x_i$ and order the variables so that this is $x_n$.
This allows to write
\bea
\sum_{j=1}^n \prod_{i=1\atop i\ne j}^n  \(\sgn(x_i-x_j)+y_i\)
&=&\(y_n-1\)\sum_{j=1}^{n-1} \prod_{i=1\atop i\ne j}^{n-1} \(\sgn(x_i-x_j)+y_i\)
+\prod_{i=1}^{n-1} \(\sgn(x_i-x_n)+y_i\)
\nn\\
&=& \hf\(y_n-1\)\[\prod_{i=1}^{n-1} \(y_i+1\)-\prod_{i=1}^{n-1} \(y_i-1\)\]+\prod_{i=1}^{n-1} \(y_i+1\)
\\
&=&
\hf\,\prod_{i=1}^{n} \(y_i+1\)-\hf\,\prod_{i=1}^{n} \(y_i-1\),
\nn
\eea
where at the second step we have used the induction hypothesis.

The identity \eqref{signident4} can be used to derive another one.
Let us choose $i,j$ in \eqref{signident4} to run from 0 to $n$ and take $x_0=y_0=0$.
Then the identity gives
\be
-\sum_{j=1}^n \sgn(x_j)\prod_{i=1\atop i\ne j}^n \(\sgn(x_i-x_j)+y_i\)+\prod_{i=1}^n \(\sgn(x_i)+y_i\)=
\hf\,\prod_{i=1}^{n} \(y_i+1\)+\hf\,\prod_{i=1}^{n} \(y_i-1\).
\ee
Taking the sum of this new relation with \eqref{signident4} and rearranging the terms, one then obtains
\be
\sum_{i=1}^n \(\sgn(x_i)-1\)\prod_{i=1\atop i\ne j}^n \(\sgn(x_i-x_j)+y_i\)=
\prod_{i=1}^{n} \(\sgn(x_i)+y_i\)-\prod_{i=1}^{n} \(1+y_i\).
\label{mainsign}
\ee

Finally, a useful particular case of the above identities is obtained by setting $y_i=0$.
In this way, one finds
\bea
\sum_{j=1}^n \prod_{i=1\atop i\ne j}^n \sgn(x_i-x_j)&=& \eps(n),
\label{signident1}
\\
\sum_{j=1}^n \(\sgn(x_j)-1\) \prod_{i=1\atop i\ne j}^n \sgn(x_i-x_j)&=& \prod_{i=1}^n \sgn(x_i)-1,
\label{signident3}
\eea
where $\eps(n)$ is the parity of $n$, i.e. it equals 1 if $n$ is odd and 0 if $n$ is even.

\section{Flow vectors and cancellation of fake discontinuities}
\label{ap-vectors}

Given the results of section \ref{subsec-flow}, it is natural to attach to each vertex $v\in V_T$ on the flow tree
a $\ell$-dimensional vector
\be
\Avl_v=(a_{1},\dots, a_{\ell})\in\IR^\ell,
\ee
where $\ell-2$ is the depth of the vertex and the coefficients $a_i$ are constructed
from charges $\alpha_i$, attached to the branch of the tree joining the vertex with the root (see Fig. \ref{fig-AFtreen}),
as in \eqref{expr-a}.
The arguments of the sign functions appearing in the stability conditions
will then reduce to inner products of $\Avl_v$ with the `central charge vector' $\Zvl=(c_{\alpha_1},\dots, c_{\alpha_\ell})$
where $c_{\alpha} =  \Im\bigl[ Z_{\alpha}\bZ_\gamma\bigr]$ (cf. \eqref{notaion-c}).
However, this construction must be slightly modified in order
to take into account that the coefficients $a_{i}$ are defined only up to an overall shift.
Moreover,  it is useful to allow for a non-trivial
metric $\gl^{ij}$ on $\IR^\ell$, depending on the physical context,\footnote{For instance,
for the D4-D2-D0 brane system on a Calabi-Yau threefold, the relevant metric is $g^{ij}=\kappa_{abc}p_i^at^b t^c \, \delta^{ij}$
where $\kappa_{abc}$ are the triple intersection numbers,
$p^a_i$ is the D4-component of the electromagnetic charge $\alpha_i$, and $t^a$ are the K\"ahler moduli of the Calabi-Yau.}
so that the inner product takes the form
$(\Av,\Av')\equiv\sum_{i,j}^\ell \gl^{ij}\Av_i\Av'_j$.
The results presented here are completely independent of this metric, so for the purposes of
the present analysis we could take $g^{ij}=\delta^{ij}$.

Thus, we define the {\it flow vector} assigned to a vertex $v\in V_T$ as the projection
of $\Avl_v$ on the hyperplane orthogonal to the unity vector $\unit=(1,\dots, 1)$,
\be
\Cvl_v\equiv \Avl_{v\perp \unit}=\Avl_v-\frac{(\Avl,\unit)}{(\unit,\unit)}\, \unit,
\ee
or in terms of components
\be
\Cvl_{vi}=\gl^{-1}\sum_{j=1}^\ell(a_i-a_j) \gl^j ,
\label{vecCv}
\ee
where $\gl^i=\sum_j \gl^{ij}$ and $\gl=\sum_i \gl^i=(\unit,\unit)$.
We also modify the definition of the central charge vector $\Zvl$ taking its components to be
\be
\Zvl_{i}=\sum_{j=1}^\ell (\gl^{-1})_{ij}\,c_{\alpha_j}.
\ee
It is then easy to see that by virtue of $(\unit,\Zvl)=0$ we have
\be
\sign(\Cvl_v,\Zvl)=\sign(\Avl_v,\Zvl)=-\sign\,\Im\bigl(Z_{\gamma_{\Lv{v}}}\bar Z_{\gamma_{\Rv{v}}}(z_{p(v)})\bigr),
\ee
i.e. the flow vectors correctly encode the moduli dependent signs entering in \eqref{kappaT}.

Note that if one of the charges, say $\alpha_{i_0}$, can be decomposed into a sum $\alpha_{i_0}=\sum_{s=1}^{r+1}\alpha'_s$
(for instance, $\alpha'_s$ can be the charges assigned to the leaves of the tree),
then the above construction has a natural embedding into $\IR^{\ell+r}$.
To display it, let us introduce two operations: embedding and contraction.
For arbitrary vectors $\Vvi{\ell}$ and $\Uvi{\ell+r}$,
they are defined by $\Vvi{\ell+r}\equiv\iota_{i_0,r}\Vvi{\ell}$ and $\Uvi{\ell}\equiv\varsigma_{i_0,r}\Uvi{\ell+r}$ where
\be
\begin{array}{llll}
\Vvi{\ell+r}_{i}= \Vvi{\ell}_{i}, \qquad &  i<i_0, \qquad
& \Uvi{\ell}_{i}=\Uvi{\ell+r}_{i}, \qquad &  i<i_0,
\\
\Vvi{\ell+r}_{i}=\Vvi{\ell}_{i_0}, \qquad &  i_0\le i\le i_0+r,\qquad
& \Uvi{\ell}_{i_0}=\sum_{i=i_0}^{i+r}\Uvi{\ell+r}_{i}, \qquad &
\\
\Vvi{\ell+r}_{i}=\Vvi{\ell}_{i-r}, \qquad & i>i_0+r,\qquad
& \Uvi{\ell}_{i}=\Uvi{\ell+r}_{i+r}, \qquad &  i>i_0,
\end{array}
\label{embedC}
\ee
and have also obvious extension to matrices.
Then, if we require that $\gi{\ell}=\varsigma_{i_0,r}\gi{\ell+r}$
and take the $\ell+r$-dimensional central charge vector $\Zvi{\ell+r}$ to be defined as above
with the set of charges $\{\alpha_i\}$ replaced by
$\{\alpha_1,\dots,\alpha_{i_0-1},\alpha'_1,\dots,\alpha'_{r+1},\alpha_{i_0+1}\dots,\alpha_\ell\}$,
it is easy to check that the vector $\Cvi{\ell+r}_v\equiv \iota_{i_0,r}\Cvl_v$
satisfies $(\Cvi{\ell+r}_v,\Zvi{\ell+r})=(\Cvl_v,\Zvl)$.
Thus, both constructions are equally suitable for describing the stability conditions of the flow tree.

This freedom allows to avoid the inconvenience of working with vectors of different size.
To this end, it is sufficient to expand all charges $\alpha_i$ in the basis of charges assigned to the leaves of the tree
and embed all vectors $\Cvl_v$ into $\IR^n$ using the above prescription.
In practice, however, this is not necessary. In this appendix we will work simultaneously
with at most two vectors when we define their mutual orthogonal projections.
In such case the relevant branch is shown in Fig. \ref{fig-AFtreenm} and in the notations of the picture
it is enough to embed the vectors corresponding to vertices $v$ and $v'$ in the minimal common space $\IR^{\ell+r}$.
Moreover, if $r=0$ (i.e. $v'$ belongs to the branch connecting $v$ with the root), only one of the vectors requires the embedding.
Below we always choose the flow vectors in the minimal possible representation and drop the dimension label.

\lfig{The relevant branch of an attractor flow tree connecting the root and two vertices.}
{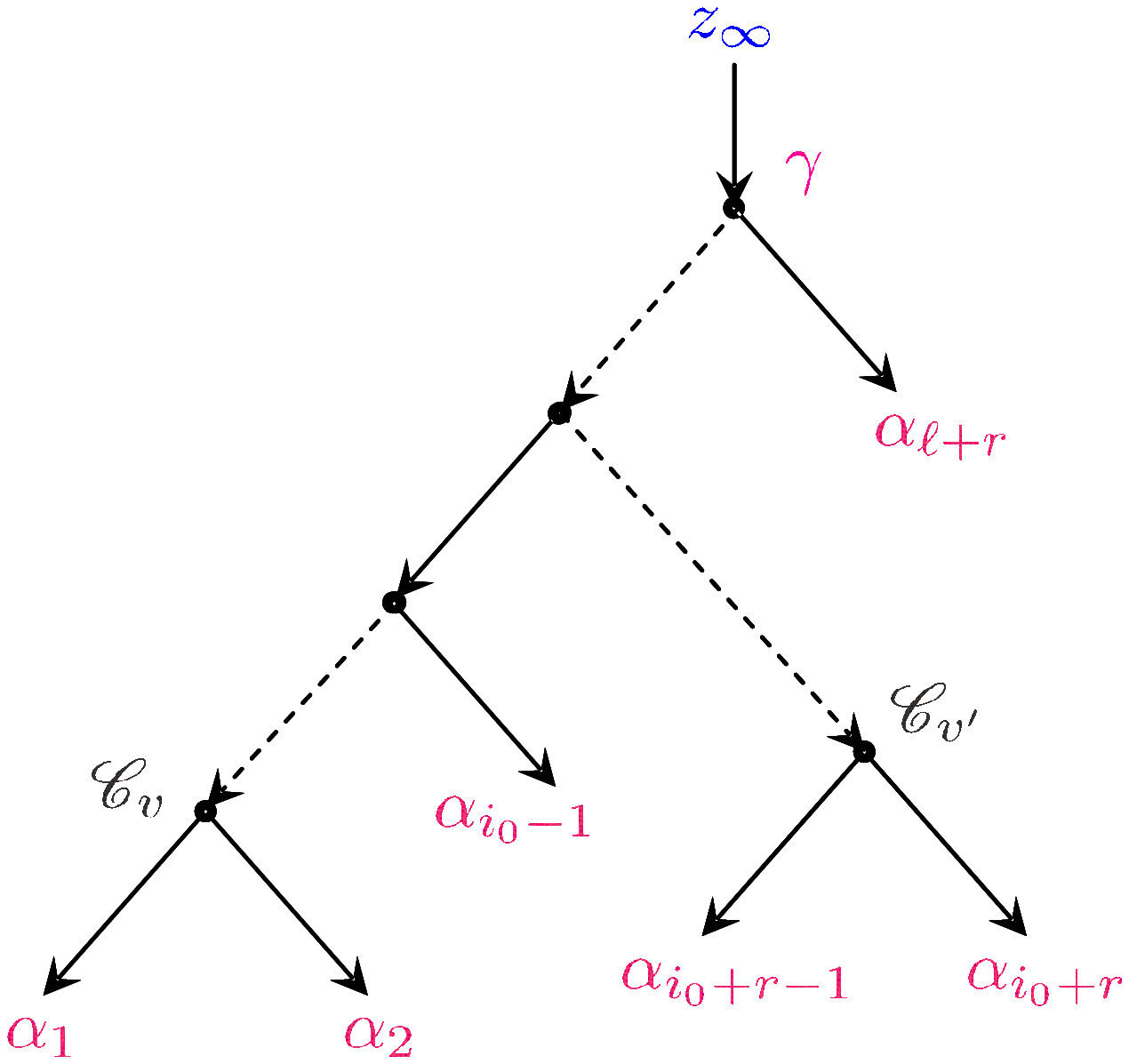}{8cm}{fig-AFtreenm}{-1.2cm}

The flow vectors turn out to be very handy for  showing the cancellation of  discontinuities across
`fake walls'
in the sum over flow trees.
Moreover, very similar vectors play a crucial role in the study of convergence and modular properties of
indefinite theta series defining the BPS partition function for D4-D2-D0 black holes
\cite{Alexandrov:2016enp,Alexandrov:2017qhn,ap-to-appear}.
Before turning to the first of these
topics however, we state some important properties satisfied by the flow vectors,
referring to subsection \ref{ap-proof} for the proofs.

\begin{proposition}\label{prop1}
If $v$ is not the root vertex, the vector $\Cv_v$ has the form
\be
\Cv_{v}=\[\prod_{v'\in \wp(v)}\gamma_{\Lv{v'}\Rv{v'}}\]^{-1} \tCv_v,
\label{formCv}
\ee
where $\wp(v)$ is the set of vertices on the path from $v$ towards the root (including the root), and $\tCv_v$
is a vector which is cyclically symmetric in $\gamma_{\Lv{v}}$, $\gamma_{\Rv{v}}$ and $\gamma_{\Rv{p(v)}}$
(assuming that $v$ belongs to the left branch of $p(v)$, i.e. that $\gamma_{\Lv{p(v)}}=\gamma_{\Lv{v}}+\gamma_{\Rv{v}}$).
\end{proposition}

Next, let us introduce the projection of $\Cv_{v'}$ on the subspace orthogonal to $\Cv_{v}$
\be
\Cv_{v'\perp v}=\Cv_{v'}-\frac{(\Cv_{v'},\Cv_{v})}{(\Cv_{v},\Cv_{v})}\, \Cv_{v}.
\ee
Such projections naturally appear when one evaluates the discontinuities of BPS indices across walls
determined by equations $(\Cv_v,\Zv)=0$.
It turns out that the projections also possess certain symmetry properties summarized in
the following proposition:

\begin{proposition}\label{theorem1}
\label{th-projections}
Depending on the relative position of vertices $v$ and $v'$ on the tree, one has four different situations:

\begin{enumerate}
\item
\label{case1}
If $v$ is the root vertex of an attractor flow tree $T$ and $v'$ belongs to its left branch, then $\Cv_{v'\perp v}$
coincides with $\Cv_{v'}^{T_1}$ which is constructed for the tree $T_1$ obtained from $T$ by removing the root and its right branch.
Equivalently,
\be
\sign(\Cv_{v'}^{T_1},\Zv)=-\sign\,\Im\bigl(Z_{\gamma_{\Lv{v}}}\bar Z_{\gamma_{\Rv{v}}}(z_{p(p(v))})\bigr),
\ee
i.e. the attractor flow is undone by one step and the central charges are evaluated at the moduli
corresponding to the parent to parent vertex.

\item
\label{case2}
If $v$ is not the root vertex and is either ancestor of $v'$, or a child of one of its ancestors,
then $\Cv_{v'\perp v}$ is the same for the three trees shown in Fig. \ref{fig-AFtreebig}.

\lfig{Three attractor flow trees relevant for the case \protect\ref{case2} of the Proposition.
Here the circles denote the branches of the flow tree attached to the corresponding vertices, whereas
the dots indicate that there can be any number of vertices with the corresponding branches.}
{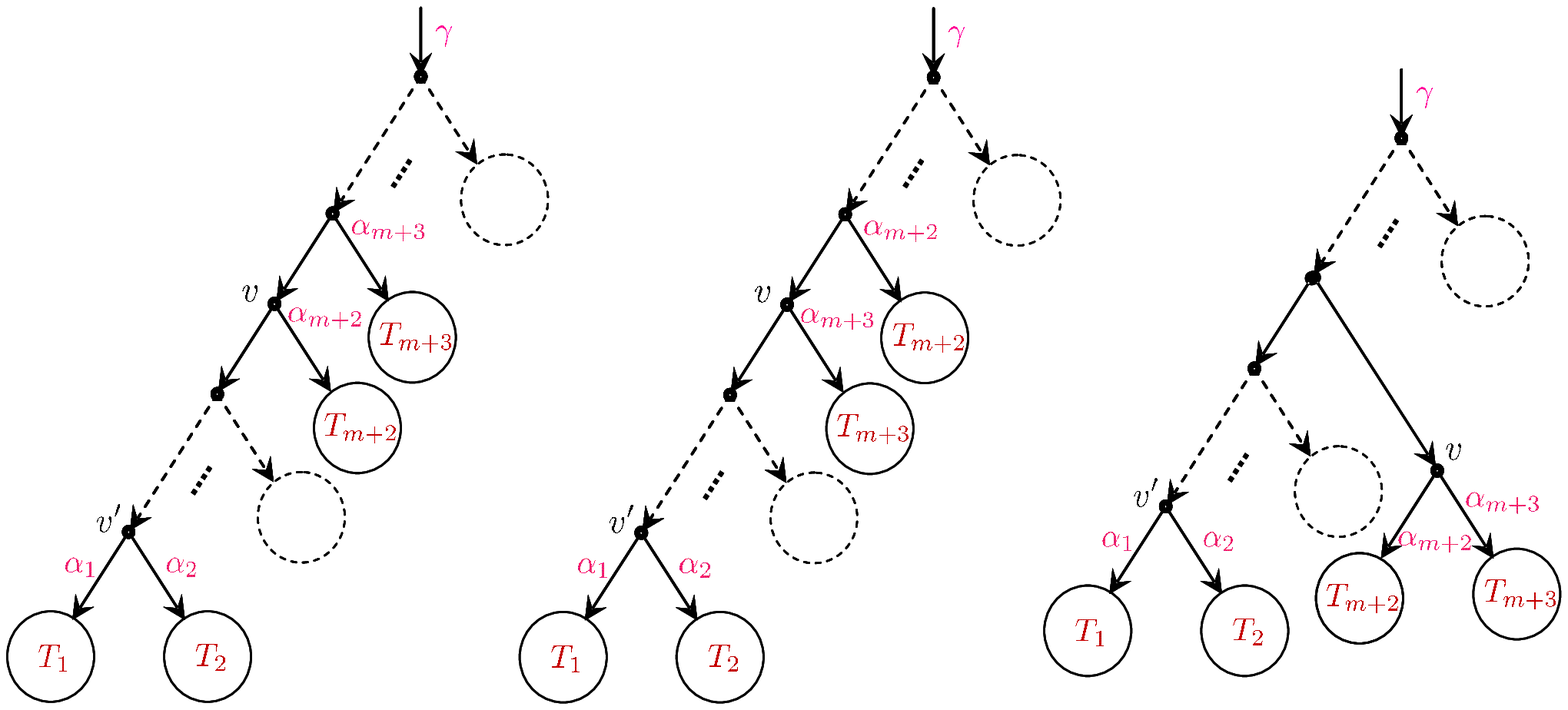}{18cm}{fig-AFtreebig}{-1.2cm}

\item
\label{case3}
If $v$ is a descendant of $v'$, then
\be
\Cv_{v'\perp v}=\gamma_{\Lv{v'}\Rv{v'}} \,\tCv_{v'\perp v},
\label{formCvvp}
\ee
where $\tCv_{v'\perp v}$ has the same symmetry as $\tCv_v$, i.e.
it is cyclically symmetric in $\gamma_{\Lv{v}}$, $\gamma_{\Rv{v}}$ and $\gamma_{\Rv{p(v)}}$.

\item
\label{case4}
If the pair $v,v'$ does not correspond to any of the previous cases,
i.e. a path from $v$ to $v'$ goes at least 2 steps up and some steps down,
then $\Cv_{v'\perp v}$ is cyclically symmetric in $\gamma_{\Lv{v}}$, $\gamma_{\Rv{v}}$ and $\gamma_{\Rv{p(v)}}$.
\end{enumerate}
\end{proposition}

The first part of Proposition \ref{theorem1}  immediately allows to see the consistency
of the attractor flows with the primitive wall-crossing formula \eqref{primwc}.
Indeed, the discontinuity of the weight $\Delta(T)$ \eqref{kappaTa}
due to the sign factor arising at the root vertex $v_0$ evaluates to
\be
\label{disc-kappaT}
\disc_{v_0}\Delta(T)=
\Delta_{v_0}(T_L;T)\,\Delta_{v_0}(T_R;T),
\ee
where $T_L$ and $T_R$ are the two branches of the tree $T$ growing from the root vertex and we defined
\be
\Delta_{v}(T';T)=\prod_{v'\in T'_{0}}\,\frac12\, \Bigl[\sign ( \Cv_{v'\perp v}^{T},\Zv)- \sign (\gamma_{\Lv{v'}\Rv{v'}}) \Bigr]
\label{kappaperp}
\ee
for a subtree $T'\subset T$ not containing $v$. Here we indicated explicitly
with an upper index that the vector $\Cv_{v'\perp v}$ is defined for the tree $T$.
But due to Proposition \ref{theorem1} (case \ref{case1}), for $v=v_0$ it coincides with the vector
$\Cv_{v'}$ defined for either $T_L$ or $T_R$ so that
$\Delta_{v}(T_{L,R};T)=\Delta(T_{L,R})$. Thus, one obtains
\be
\disc_{v_0}\Delta(T)
=
\Delta(T_L)\,\Delta(T_R) .
\ee
After multiplying by $\kappa(T)=-\kappa( \gamma_{\Lv{v_0}\Rv{v_0}})\kappa(T_L)\kappa(T_R)$, by the product
of the attractor indices $\prod_{i=1}^n\bOm_*(\gamma_i)$, and summing over all trees,
this is indeed in agreement with the primitive wall crossing formula.

\lfig{Three trees whose contributions are discontinuous on the same fake wall,
but which conspire to give a smooth contribution to the BPS index.}
{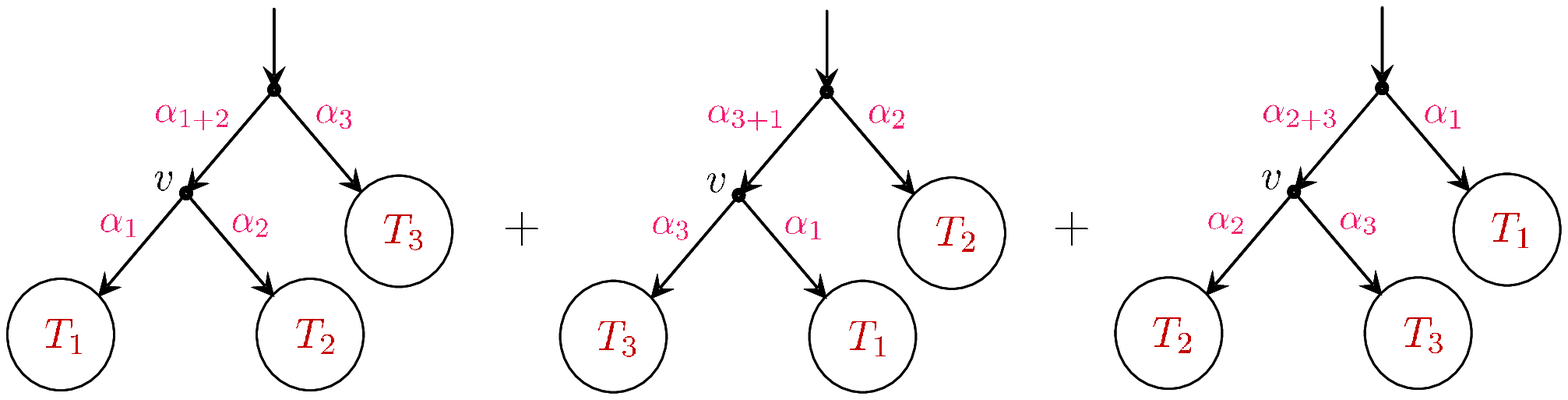}{18cm}{fig-AFtree33}{-1.5cm}

While only the first part of Proposition \ref{theorem1} entered this proof,
the remaining three parts ensure the cancelation of the fake discontinuities appearing
due to the sign factors in \eqref{kappaT} attached to non-root vertices. To see this,
note that due to Proposition \ref{prop1},
the same fake wall arises for three different trees whose parts containing vertex $v$
are shown in Fig. \ref{fig-AFtree33}.
They are obtained by cyclic permutations of the subtrees $T_1$, $T_2$ and $T_3$.
Their discontinuities at the fake wall are given by
\be
\disc_{v}\bigl[\kappa(T)\Delta(T)\bigr]=
\kappa(T)\,
\Delta_{v}(T_v;T)\prod_{i=1}^3\Delta_{v}(T_i;T)
\prod_{v'\in \wp(v)}\sign(\gamma_{\Lv{v'}\Rv{v'}}),
\ee
where the factors $\Delta_v$ have been defined in \eqref{kappaperp},
$T_v=T\setminus(T_1\cup T_2\cup T_3\cup\{v\})$ is the part of $T$
not shown in Fig. \ref{fig-AFtree33}, and the last factor comes from \eqref{formCv}.
Case \ref{case2} of the Proposition \ref{theorem1} ensures that the factors $\Delta_{v}(T_i;T)$
are the same for all three choices of $T$, whereas cases \ref{case3} and \ref{case4} tell us the same about
\be
\Delta_{v}(T_v;T)\, \prod_{v'\in \wp(v)}\sign(\gamma_{\Lv{v'}\Rv{v'}}).
\ee
Finally, one finds that
\be
\kappa(T)=\kappa(\gamma_{\Lv{v}\Rv{v}} )\kappa( \gamma_{\Lv{p(v)}\Rv{p(v)}})
\[\frac{\kappa(T_v)}{\kappa( \gamma_{\Lv{p(v)}\Rv{p(v)}})}\prod_{i=1}^3\kappa(T_i)\],
\label{decomkappa}
\ee
where the factor in the square brackets is the same for all trees shown in Fig. \ref{fig-AFtree33}.
Thus, only the first two $\kappa$-factors in \eqref{decomkappa} differ the three discontinuities.
It remains to note that in the notations of Fig. \ref{fig-AFtree33}, for the first tree,
$\gamma_{\Lv{v}\Rv{v}}=\alpha_{12}$ and $\gamma_{\Lv{p(v)}\Rv{p(v)}}=\alpha_{1+2,3}$,
whereas for other trees they are given by cyclic permutations.
Then the cancellation of fake discontinuities follows from the identity \eqref{kappa123}.

\subsection{Proof of the symmetry properties of the flow vectors}
\label{ap-proof}

Now we fill the gap and prove the symmetry properties of the flow vectors and their orthogonal projections stated above.
We concentrate only on the relevant branch of the tree and label the charges as in Fig. \ref{fig-AFtreenm}.
The flow vectors are constructed from the coefficients $a_i$ \eqref{expr-a} as explained in the beginning of this section.
These coefficients satisfy an important symmetry property under permutations of $(\alpha_1,\alpha_2,\alpha_3)$,
which will be the starting point of our analysis.

\begin{lemma}\label{lemma1}
Under a permutation $\sigma$ of $\alpha_1,\alpha_2,\alpha_3$, the coefficients $a_i$ transform as
\be
\beta_{33} a_i\ \mapsto\ \eps_\sigma \beta_{33}(a_{\sigma(i)}-a_{\sigma(1)}),
\label{cycl-a}
\ee
where $\eps_\sigma$ is the parity of the permutation.
\end{lemma}

\begin{proof}
Let us use the freedom of shifting all $a_i$ by the same constant and define
\be
a'_i = \alpha_{23} + \beta_{33} a_i ,
\ee
where, as usual, $\alpha_{ij}=\langle\alpha_i,\alpha_j\rangle$.
Then
\be
a'_1= \alpha_{23},
\qquad
a'_2= \alpha_{31},
\qquad
a'_3= \alpha_{12}
\ee
are permuted under permutations of $\alpha_1,\alpha_2,\alpha_3$ with the sign given by $\eps_\sigma$, while $a'_{i\geq 4}$
are invariant up to the same sign, as can be shown inductively from the recursion
\be
\beta_{ii} a'_i = - \sum_{j=4}^{i-1} \beta_{ij} a'_j - \gamma_{23} \beta_{i1}
- \gamma_{31} \beta_{i2} - \gamma_{12} \beta_{i3}
\qquad (i\geq 4)
\ee
and the fact that $\beta_{ij}$ is invariant under  cyclic permutations of $\alpha_1,\alpha_2,\alpha_3$ whenever $i,j\geq 4$.
Returning to the original coefficients $a_i$, we arrive at the transformation \eqref{cycl-a}.
\end{proof}

Using this Lemma, it is straightforward to prove Proposition \ref{prop1}.

\begin{proof}[\bf Proof of  Proposition \ref{prop1}:]
In this case it is sufficient to consider only the branch of the tree connecting vertex $v$ to the root.
Our aim is to prove that $\beta_{33}\Cv_{vi}$ is mapped by a cyclic permutation $\sigma$ to $\beta_{33}\Cv_{v\sigma(i)}$.\footnote{We multiply
the argument of the sign by $\beta_{33}$ only, whereas in \eqref{formCv}
the factor $\prod_{j=3}^n\beta_{jj}$ is extracted. But $\beta_{jj}$ for $j>3$ depend on the first three charges only in
the combination $\alpha_1+\alpha_2+\alpha_3$ and therefore are cyclically symmetric in these charges. Thus, they are irrelevant
for the present discussion.}
This can be done using the explicit expression \eqref{vecCv} for the components of the vector $\Cv_v$
and Lemma \ref{lemma1}. Indeed, the transformation law \eqref{cycl-a} implies
(for a cyclic permutation $\eps_\sigma=1$)
\be
\beta_{33}\Cv_{vi}\mapsto \beta_{33}g^{-1}\sum_{j=1}^\ell(a_{\sigma(i)}-a_{\sigma(j)}) g^{\sigma(j)}
=\beta_{33}g^{-1}\sum_{j=1}^\ell(a_{\sigma(i)}-a_j)g^j =\beta_{33}\Cv_{v\sigma(i)},
\ee
which is the required statement.
\end{proof}

In the following we will also need the coefficients of $\Cvi{\ell}_{v_{m+1}}\equiv\iota_{1,m}\Cvi{\ell-m}_{v_{m+1}}$,
which is, in the notations of Fig. \ref{fig-AFtreen}, the flow vector assigned to
vertex $v_{m+1}$ and embedded into $\IR^\ell$, the habitat of the flow vector for $v_1$.
From \eqref{embedC} it follows that
\be
\Cvi{\ell}_{v_{m+1},i}=g^{-1}\sum_{j=1}^\ell\(\am{m}_i-\am{m}_j\)g^j,
\label{vecm-def}
\ee
where the coefficients $\am{m}_i$ are obtained by the following substitution
\be
\begin{split}
\am{m}_i=&\, 0, \qquad \qquad\qquad\qquad\quad\ i\le m+1,
\\
\am{m}_i=&\, a_{i-m}|_{\alpha_1\to \alpha_1+\cdots+\alpha_{m+1} \atop \alpha_j\to \alpha_{j+m},\ \ j\ge 2}, \qquad i>m+1 \, .
\end{split}
\label{def-apm}
\ee
Note that under this substitution, $\beta_{ij}\to \beta_{i+m,j+m}$ for $j>m$.
For $m=1$, only such $\beta_{ij}$ appear in the expression for $a_i$ \eqref{expr-a}. Then one has the following


\begin{lemma}
Under a permutation $\sigma$ of $\alpha_1,\alpha_2,\alpha_3$, the coefficients $\am{1}_i$ transform as
\be
\am{1}_i\ \mapsto\ A_\sigma\( \am{1}_{\sigma(i)}-\am{1}_{\sigma(1)}\)+A^{(1)}_\sigma \(a_{\sigma(i)}-a_{\sigma(1)}\),
\label{cycl-ap}
\ee
where
\be
\begin{split}
A_\sigma=&\, \eps_\sigma\(a_{\sigma(1)}-a_{\sigma(2)}\)=\frac{\beta_{3\sigma(3)}}{\beta_{33}},
\\
A^{(1)}_\sigma=&\, \eps_\sigma\(\am{1}_{\sigma(2)}-\am{1}_{\sigma(1)}\).
\end{split}
\ee
\end{lemma}

\begin{proof}
The check of the transformation for $i=1,2,3$ is straightforward. For $i\ge 4$ we use
the recursive formula
\be
\am{1}_{i} = -\sum_{j=3}^{i-1} \frac{\beta_{ij}}{\beta_{ii}}\,\am{1}_{j}.
\label{rec-ap}
\ee
Proceeding by induction and noticing that the sum in \eqref{rec-ap} can be extended to start from $j=1$, one obtains
\be
\begin{split}
\am{1}_i\mapsto &\,  -\sum_{j=1}^{i-1}\frac{\beta_{i\sigma(j)}}{\beta_{ii}}
\(A_\sigma\( \am{1}_{\sigma(j)}-\am{1}_{\sigma(1)}\)+A^{(1)}_\sigma\(a_{\sigma(j)}-a_{\sigma(1)}\)\)
\\
=&\,A_\sigma\( \am{1}_{\sigma(i)}-\am{1}_{\sigma(1)}\)+A^{(1)}_\sigma\(a_{\sigma(i)}-a_{\sigma(1)}\),
\end{split}
\ee
where we changed the summation variable $j\to\sigma^{-1}(j)$ and used the property $\sum_{j=1}^i \beta_{ij}=0$
as well as the two recursion relations \eqref{rec-a}, \eqref{rec-ap}.
\end{proof}

These shifted coefficients are useful, in particular, to make explicit the properties of $a_i$ under the exchange of $\alpha_{m+2}$
and $\alpha_{m+3}$ for $m\ge 1$.
Indeed, whereas for $i<m+2$, $a_i$ are independent of these charges, for $i>m+2$ we can write
\be
a_i=-a_{m+2} \am{m}_i +\bm{m}_i,
\qquad
i>m+2,
\label{aiam}
\ee
where $\bm{m}_i$ are defined by the recursion relation (cf. \eqref{rec-a})
\be
\bm{m}_{i} = -\sum_{j=2}^{m+1} \frac{\beta_{ij}}{\beta_{ii}}\,a_{j}-\sum_{j=m+3}^{i-1}
\frac{\beta_{ij}}{\beta_{ii}}\,\bm{m}_{j}.
\label{rec-api}
\ee
It is immediate to see that
\be
\bm{m}_i=-\bm{m}_{m+3} \am{m+1}_i+\cm{m}_i,
\qquad
i>m+3,
\label{apiam}
\ee
where
\be
\bm{m}_{m+3}=-\sum_{j=2}^{m+1} \frac{\beta_{m+3,j}}{\beta_{m+3,m+3}}\,a_{j},
\label{expr-am3}
\ee
whereas $\cm{m}_i$, defined by
\be
\cm{m}_{i} = -\sum_{j=2}^{m+1} \frac{\beta_{ij}}{\beta_{ii}}\,a_{j}-\sum_{j=m+4}^{i-1}
\frac{\beta_{ij}}{\beta_{ii}}\,\cm{m}_{j},
\label{rec-appi}
\ee
are invariant under the exchange of $\alpha_{m+2}$ and $\alpha_{m+3}$.

Let us also introduce
\be
\dm{m}_{ij}=a_i \am{m}_j-a_j \am{m}_i .
\label{def-dij}
\ee
Using \eqref{aiam}, one finds
\be
\dm{m}_{ij}=\left\{
\begin{array}{lcl}
0,
& \qquad & i,j<m+2,
\\
a_i\am{m}_j,
& \qquad & i<m+2,\ j\ge m+2,
\\
\bm{m}_j,
& \qquad & i=m+2,\ j>m+2,
\\
\bm{m}_i\am{m}_j-\bm{m}_j \am{m}_i ,
& \qquad & i,j>m+2.
\end{array}
\right.
\label{expr-dij}
\ee
Furthermore, using \eqref{apiam}, one obtains
\be
\dm{m}_{ij}=\left\{
\begin{array}{lcl}
\bm{m}_{m+3},
& \qquad & i=m+2,\ j=m+3,
\\
-\bm{m}_{m+3} \am{m+1}_j+\cm{m}_j,
& \qquad & i=m+2,\ j> m+3,
\\
\bm{m}_{m+3}\(\am{m}_j+\am{m}_{m+3}\am{m+1}_j\)-\cm{m}_j\,\am{m}_{m+3} ,
& \qquad & i=m+3,\ j>m+3,
\\
\bm{m}_{m+3}\(\am{m+1}_j\am{m}_i- \am{m+1}_i \am{m}_j \) +\cm{m}_i\am{m}_j-\cm{m}_j\am{m}_i,
& \qquad & i,j>m+3.
\end{array}
\right.
\label{expr-dijp}
\ee


\begin{lemma}
Under the exchange of $\alpha_{m+2}$ and $\alpha_{m+3}$, the coefficients $\dm{m}_{ij}$ transform as
\be
\beta_{m+3,m+3}\dm{m}_{ij}\ \mapsto\ -\beta_{m+3,m+3}\dm{m}_{\sigma(i)\sigma(j)}.
\label{tr-dij}
\ee
\end{lemma}

\begin{proof}
It is straightforward to check this transformation using \eqref{expr-dij}, \eqref{expr-dijp}, \eqref{expr-am3},
$\am{m}_{m+3}=\frac{\beta_{m+3,m+2}}{\beta_{m+3,m+3}}$, and that under this exchange $\am{m}_i$ and $\am{m+1}_i$
transform as in \eqref{cycl-a} and \eqref{cycl-ap}, respectively, where one should replace $a_i$ by $\am{m}_i$, $\am{1}_i$ by $\am{m+1}_i$,
and $\beta_{33}$ by $\beta_{m+3,m+3}$.
\end{proof}

Now we are ready to prove Proposition \ref{theorem1}.

\begin{proof}[\bf Proof of case \ref{case1}:]
If $v$ is the root vertex, the components of $\Cv_{v'}$ and $\Cv_v$ embedded into $\IR^\ell$ are given by
\be
\Cv_{vi}=-g^{-1}g^\ell, \quad i<\ell,
\qquad
\Cv_{v\ell}=g^{-1}g',
\ee
\be
\Cv_{v'i}=g^{-1}\sum_{j=1}^\ell(a_i-a_j)g^j,
\label{compCvp}
\ee
where $g'=\sum_{j=1}^{\ell-1} g^j $ and we label the charges as in Fig. \ref{fig-AFtreen} with $v_1=v'$ and $v_{\ell-1}=v$.
One finds for these vectors
\be
\Cv_v^2=g^{-1}g' g^\ell ,
\qquad
\Cv_{v'}^2=\frac{1}{2g}\sum_{i,j=1}^\ell (a_i-a_j)^2\,g^i g^j,
\ee
\be
\begin{split}
(\Cv_v,\Cv_{v'})=&\, g^{-1}\[\sum_{i=1}^{\ell-1}\sum_{j=1}^{\ell}(a_j-a_i)g^i g^j-g'\sum_{j=1}^{\ell-1}(a_j-a_\ell) g^j\]
\\
=&\, g^{-1} g^\ell\sum_{j=1}^{\ell}(a_\ell-a_j)g^j.
\end{split}
\ee
This implies
\be
\begin{split}
\Cv_{v'\perp v,i}=&\, (gg')^{-1}\sum_{j=1}^\ell(g' a_i+g^\ell a_\ell-g a_j)g^j
=g'^{-1}\sum_{j=1}^{\ell-1}(a_i-a_j) g^j, \quad i<\ell,
\\
\Cv_{v'\perp v,\ell}=&\, 0.
\end{split}
\ee
This is precisely the vector corresponding to the tree obtained by removing the root.
This proves the first statement of the Proposition.
\end{proof}

\begin{proof}[\bf Proof of case \ref{case2}:]
Labeling the charges as in Fig. \ref{fig-AFtreebig}, the components of the vectors
$\Cv^{(1)}_{v'}$ and $\Cv^{(1)}_{v}$, corresponding to the first tree and embedded into $\IR^\ell$,
can be found in \eqref{compCvp} and \eqref{vecm-def}, respectively.
For these vectors, one obtains
\be
(\Cv^{(1)}_{v})^2=\frac{1}{2g}\sum_{i,j=1}^\ell \(\am{m}_i-\am{m}_j\)^2g^i g^j,
\qquad
(\Cv^{(1)}_{v'})^2=\frac{1}{2g}\sum_{i,j=1}^\ell (a_i-a_j)^2\,g^ig^j,
\ee
\be
\begin{split}
(\Cv^{(1)}_v,\Cv^{(1)}_{v'})=&\, \frac{1}{2g}\sum_{i,j=1}^\ell (a_i-a_j)\(\am{m}_i-\am{m}_j\)g^ig^j.
\end{split}
\ee
This implies
\be
\begin{split}
\Cv^{(1)}_{v'\perp v,i}=&\,
\frac{\sum_{j,l=1}^\ell (a_i-a_j)\(\am{m}_i-\am{m}_l\)\(\am{m}_j-\am{m}_l\)g^jg^l}
{\hf\sum_{j,l=1}^\ell \(\am{m}_j-\am{m}_l\)^2g^jg^l}
\\
=&\,
g^{-1}(\Cv^{(1)}_{v})^{-2}\sum_{j,l=1}^\ell\(\dm{m}_{ij}+\hf\,\dm{m}_{jl}\)\(\am{m}_j-\am{m}_l\)g^jg^l ,
\end{split}
\label{proj-m}
\ee
where $\dm{m}_{ij}$ is defined in \eqref{def-dij}.

First, to show the equality of $\Cv_{v'\perp v}$ for the first two trees, we need to show that the vector \eqref{proj-m}
is invariant under the exchange of $\alpha_{m+2}$ and $\alpha_{m+3}$. Using that under this exchange $\am{m}_i$ transform
as in \eqref{cycl-a} where one should replace $a_i$ by $\am{m}_i$ and $\beta_{33}$ by $\beta_{m+3,m+3}$, whereas
the transformation of $\dm{m}_{ij}$ is given by Lemma 3 (see \eqref{tr-dij}), it is immediate to see that
\be
\Cv^{(1)}_{v'\perp v,i} \mapsto \Cv^{(2)}_{v'\perp v,i}=\Cv^{(1)}_{v'\perp v,\sigma(i)}.
\ee

For the third tree, we need to consider the embedding $\iota_{m+2,1}\Cvi{\ell-1}_{v'}$.
The components of such vector are given by
\be
\Cv^{(3)}_{v'i}=g^{-1}\sum_{j=1}^\ell(\ta_i-\ta_j) g^j,
\label{vecm3}
\ee
where
\be
\begin{split}
\ta_i =&\, a_i,
\hspace{4cm}
i< m+2,
\\
\ta_{m+2} =&\, a_{m+2}|_{\alpha_{m+2}\to \alpha_{m+2}+\alpha_{m+3}},
\\
\ta_i =&\, a_{i-1}|_{\alpha_{m+2}\to \alpha_{m+2}+\alpha_{m+3}\atop \alpha_j\to \alpha_{j+1},\ \ j>m+2},
\qquad
i> m+2,
\end{split}
\label{def-ta}
\ee
On the other hand, the vector $\Cv^{(3)}_{v}$, due to Proposition \ref{prop1}, coincides up to a factor,
which anyway cancels in the orthogonal projection, with $\Cv^{(1)}_{v}$.
Thus, now
\be
\begin{split}
\Cv^{(3)}_{v'\perp v,i}
=&\,
g^{-1}(\Cv^{(1)}_{v})^{-2}\sum_{j,l=1}^\ell\(\tdm{m}_{ij}+\hf\,\tdm{m}_{jl}\)\(\am{m}_j-\am{m}_l\)g^jg^l ,
\end{split}
\label{proj-m3}
\ee
where
\be
\tdm{m}_{ij}=\ta_i \am{m}_j-\ta_j \am{m}_i .
\label{def-tdij}
\ee

Let us find the coefficients $\ta_i$ explicitly.
First, we have
\be
\ta_{m+2}=\ta_{m+3}=-\sum_{j=2}^{m+1} \frac{\beta_{m+3,j}}{\beta_{m+3,m+2}+\beta_{m+3,m+3}}\,a_{j}.
\ee
Note, in particular, the following property
\be
\ta_{m+2}\(1+\am{m}_{m+3}\)=\bm{m}_{m+3},
\ee
where the quantity on the r.h.s. was defined in \eqref{expr-am3}.
Next, we have (cf. \eqref{aiam})
\be
\ta_i=-\ta_{m+2} \tam{m}_i +\tbm{m}_i,
\qquad
i>m+3,
\label{taiam}
\ee
where the tilde means the same substitution as in \eqref{def-ta}.
In particular, one has
\bea
\tbm{m}_{i}&=& \bm{m}_{i-1}|_{\alpha_{m+2}\to \alpha_{m+2}+\alpha_{m+3}\atop \alpha_j\to \alpha_{j+1},\ \ j>m+2}
= \[-\sum_{j=2}^{m+1} \frac{\beta_{i-1,j}}{\beta_{i-1,i-1}}\,a_{j}-\sum_{j=m+3}^{i-2} \frac{\beta_{i-1,j}}{\beta_{i-1,i-1}}\,\bm{m}_{j}
\]_{\alpha_{m+2}\to \alpha_{m+2}+\alpha_{m+3}\atop \alpha_j\to \alpha_{j+1},\ \ j>m+2}
\nn\\
&=& \[-\sum_{j=2}^{m+1} \frac{\beta_{ij}}{\beta_{ii}}\,a_{j}-\sum_{j=m+4}^{i-1} \frac{\beta_{ij}}{\beta_{ii}}\,\tbm{m}_{j}
\].
\label{rec-tapi}
\eea
The resulting recursive relation coincides with \eqref{rec-appi}, which allows to conclude that
\be
\tbm{m}_{i}=\cm{m}_{i}.
\label{eq-tbc}
\ee
Furthermore, using the linearity of $\am{m}_i$ in $\beta_{j,m+2}$ (which follows from the linearity of $a_i$ in $\beta_{j2}$),
it is easy to see that
\be
\tam{m}_i=\am{m}_{i-1}|_{\alpha_{m+2}\to\alpha_{m+2}+\alpha_{m+3}\atop \alpha_j\to \alpha_{j+1},\ \ j>m+2}
=\am{m+1}_i+\am{m+1}_i|_{\alpha_{m+3}\leftrightarrow \alpha_{m+2}}.
\ee
The last term can be evaluated using Lemma 2, which gives
\be
\tam{m}_i=\(1+\frac{\beta_{m+3,m+2}}{\beta_{m+3,m+3}}\)\am{m+1}_i+\am{m}_i=\(1+\am{m}_{m+3}\)\am{m+1}_i+\am{m}_i.
\label{tami}
\ee
As a result, one obtains
\be
\ta_i=\cm{m}_i-\bm{m}_{m+3}\am{m+1}_i-\ta_{m+2}\am{m}_i.
\label{tai-res}
\ee
Substituting this into the definition \eqref{def-tdij}, it is straightforward to verify that the result coincides with
\eqref{expr-dij} and \eqref{expr-dijp}, i.e.
\be
\tdm{m}_{ij}=\dm{m}_{ij}.
\label{eqdij}
\ee
Then comparing \eqref{proj-m3} and \eqref{proj-m}, one concludes that
the vectors describing orthogonal projections for the two trees indeed coincide,
\be
\Cv^{(3)}_{v'\perp v,i}=\Cv^{(1)}_{v'\perp v,i}.
\ee
This completes the proof of the second statement of the Proposition.
\end{proof}

\begin{proof}[\bf Proof of case \ref{case3}:]
In this case, if $v$ is {\it not} a child of $v'$, the statement is trivial. Indeed, in this case
$\Cv_{v'}$ depends on $\gamma_{\Lv{v}}$, $\gamma_{\Rv{v}}$ and $\gamma_{\Rv{p(v)}}$ only through the sum
of these charges so that it is automatically cyclically symmetric. Furthermore, by Proposition \ref{prop1},
$\Cv_v$ is also cyclically symmetric up to an overall factor,
but this factor cancels in the orthogonal projection $\Cv_{v'\perp v}$.
This allows to conclude that $\Cv_{v'\perp v}$ is cyclically symmetric,
and since $\gamma_{\Lv{v'}\Rv{v'}}$ also depends only on the sum of the permuted charges,
the statement follows.

Thus, it remains to analyze the case when $v$ is a child of $v'$.
Let us label the charges as in Fig. \ref{fig-AFtreen} with $v_1=v$ and $v_2=v'$
so that the components of the vectors $\Cv_{v}$ and $\Cv_{v'}$ are given by
\be
\Cv_{vi}=g^{-1}\sum_{j=1}^\ell(a_i-a_j) g^j ,
\qquad
\Cv_{v'i}=g^{-1}\sum_{j=1}^\ell\(\am{1}_i-\am{1}_j\) g^j,
\ee
where the coefficients $\am{1}_i$ are defined in \eqref{def-apm}.
For these vectors, one obtains
\be
\Cv_{v}^2=\frac{1}{2g}\sum_{i,j=1}^\ell (a_i-a_j)^2\,g^ig^j,
\qquad
\Cv_{v'}^2=\frac{1}{2g}\sum_{i,j=1}^\ell \(\am{1}_i-\am{1}_j\)^2\,g^ig^j,
\ee
\be
\begin{split}
(\Cv_v,\Cv_{v'})=&\, \frac{1}{2g}\sum_{i,j=1}^\ell \(\am{1}_i-\am{1}_j\)(a_i-a_j)\,g^ig^j.
\end{split}
\ee
This implies
\be
\Cv_{v'\perp v,i}=\frac{\sum_{j,l=1}^\ell \(\am{1}_i-\am{1}_j\)(a_i-a_l)(a_j-a_l)\,g^jg^l}
{\hf\sum_{j,l=1}^\ell (a_j-a_l)^2\,g^jg^l}\,.
\ee
Applying the transformations \eqref{cycl-a} and \eqref{cycl-ap} for a cyclic permutation ($\eps_\sigma=1$), one finds
\bea
\beta_{33}^{-1}\Cv_{v'\perp v,i}&\mapsto&  \frac{1}{\hf\sum_{j,l=1}^\ell (a_j-a_l)^2\,g^jg^l}
\sum_{j,l=1}^\ell \[\beta_{33}^{-1}\( \am{1}_{\sigma(i)}-\am{1}_{\sigma(j)}\)
+A^{(1)}_\sigma\beta_{3\sigma(3)}^{-1}\(a_{\sigma(i)}-a_{\sigma(j)}\)\]
\nn\\
&& \qquad \times
\(a_{\sigma(i)}-a_{\sigma(l)}\)\(a_{\sigma(j)}-a_{\sigma(l)}\)\,g^{\sigma(j)}g^{\sigma(l)}
\nn\\
&=& \beta_{33}^{-1}\Cv_{v'\perp v,{\sigma(i)}}
+A^{(1)}_\sigma\beta_{3\sigma(3)}^{-1}\,\frac{\sum_{j,l=1}^\ell \(a_{\sigma(i)}-a_j\)\(a_{\sigma(i)}-a_l\)(a_j-a_l)\,g^jg^l}
{\hf\sum_{j,l=1}^\ell (a_j-a_l)^2\,g^jg^l}
\nn\\
&=& \beta_{33}^{-1}\Cv_{v'\perp v,{\sigma(i)}},
\eea
where the last term vanishes due to symmetrization in $j,l$.
This result proves the third statement of the Proposition.
\end{proof}

\begin{proof}[\bf Proof of case \ref{case4}:]
The proof in this case is analogous to the previous one for the case when $v$ was not a child of $v'$.
Indeed, in the present case again $\Cv_{v'}$ depends on $\gamma_{\Lv{v}}$,
$\gamma_{\Rv{v}}$ and $\gamma_{\Rv{p(v)}}$ only through their sum,
whereas  $\Cv_v$ is cyclically symmetric by Proposition
\ref{prop1} up to a factor which is canceled in the projection.
Thus, $\Cv_{v'\perp v}$ is indeed cyclically symmetric.
\end{proof}

\section{Explicit expressions for partial indices}
\label{ap-indFF}

In this appendix we provide explicit expressions for the partial tree and Coulomb indices up to $n=4$.
The results are conveniently formulated in terms of
\be
\cs_k=\sum_{i=1}^k c_i,
\qquad
\Gamma_{kl}=\sum_{i=1}^k\sum_{j=1}^l \gamma_{ij},
\qquad
\Gamma_{kl,m}=\sum_{i=m+1}^k\sum_{j=m+1}^l \gamma_{ij}.
\label{notSG}
\ee

\subsection{Tree index}
\label{ap-pind}

Recombining the sign functions in the definition \eqref{defFpl} and writing the result
in the same form as in (the second line of) \eqref{F-ansatz}, one obtains
\be
\Ftr{1}=1,
\label{resF1}
\ee
\bea
\Ftr{2}&=& \hf\,\Bigl[\sgn(c_1)+\sgn(\gamma_{12})\Bigr],
\label{resF2}
\eea
\be
\begin{split}
\Ftr{3}=&\, \frac{1}{4}\,\Bigl[\bigl(\sgn(\cs_1) +\sgn(\gamma_{12})\bigr)\bigl(\sgn(\cs_2)+ \sgn (\gamma_{23}) \bigr)
\\
&\,
- \bigl(\sgn(\Gamma_{31}) +\sgn(\gamma_{12})\bigr)\bigl(\sgn(\Gamma_{32})+ \sgn (\gamma_{23}) \bigr)\Bigr],
\end{split}
\label{F3res}
\ee
\be
\begin{split}
\Ftr{4} =&\, \frac{1}{8}\,\Bigl[\bigl(\sgn(\cs_1) +\sgn(\gamma_{12})\bigr)\bigl(\sgn(\cs_2)
+ \sgn \gamma_{23} \bigr)\bigl(\sgn(\cs_3)+ \sgn \gamma_{34} \bigr)
\\
&\,
- \bigl(\sgn(\Gamma_{41}) +\sgn(\gamma_{12})\bigr)\bigl(\sgn(\Gamma_{42})+ \sgn \gamma_{23} \bigr)
\bigl(\sgn(\Gamma_{43})+ \sgn \gamma_{34} \bigr)
\\
&\,
-\bigl(\sgn(\cs_1)-\sgn(\Gamma_{41})\bigr)
\bigl(\sgn(\Gamma_{42,1})+\sgn(\gamma_{23})\bigr)\bigl(\sgn(\Gamma_{43,1})+\sgn(\gamma_{34})\bigr)
\\
&\,
-\bigl(\sgn(\cs_3)-\sgn(\Gamma_{43})\bigr)\bigl(\sgn(\Gamma_{31})+\sgn(\gamma_{12})\bigr)
\bigl(\sgn(\Gamma_{32})+\sgn(\gamma_{23}))\bigr)\Bigr].
\end{split}
\label{F4res}
\ee
It is easy to see that the structure of these expressions is exactly the same as predicted by \eqref{F-ansatz}.

\subsection{Coulomb index}
\label{ap-ind4}

The partial Coulomb index $\FCn{2}$ coincides with \eqref{resF2}, whereas
the result for $n=3$ follows from \eqref{FCrecur} and is given by \cite[(2.57)]{Manschot:2013sya}
\be
\label{FC3}
\begin{split}
\FCn{3} = &\,
\frac14 \,\Bigl[ \bigl(\sgn(\cs_1) +\sgn(\gamma_{12})\bigr)\bigl(\sgn(\cs_2)+ \sgn (\gamma_{23}) \bigr)
\\
 &\, + \bigl( \sign(\gamma_{12}) +\sign (\gamma_{23})\bigr)
\bigl(\sign(\gamma_{12} +\gamma_{23}+\gamma_{13})-\sign(\gamma_{12}+\gamma_{23}) \bigr)\Bigr] .
\end{split}
\ee

The expression for $\FCn{4}$ is already much more complicated. It is found to be \cite[(2.61)]{Manschot:2013sya}
\bea
\FCn{4}&=&
\frac18\, \prod_{k=1}^3 \bigl(\sign(\cs_k)+ \sign (\gamma_{k,k+1})\bigr)
\label{FC4}\\
&&
+ \frac12\, G_4(\gamma_1^{(1)},\dots, \gamma_4^{(1)})
\( \sign\(\sum_{i<j} \gamma_{ij}\) - \sign(\gamma_{12}+\gamma_{23}+\gamma_{34})\)
\nn\\
&&
+\frac14\,
G_3(\gamma_2,\gamma_3,\gamma_4)
\bigl(\sign(\gamma_{23}+\gamma_{34}+\gamma_{24})-\sign(\gamma_{23}+\gamma_{34})\bigr)
\bigl(\sign \(\gamma_{12}+\lambda_2 \gamma_{1,3+4} \) + \sign (\cs_1) \bigr)
\nn\\
&&
+\frac14\,
G_3(\gamma_1,\gamma_2,\gamma_3)
\bigl(\sign(\gamma_{12}+\gamma_{23}+\gamma_{13})-\sign(\gamma_{12}+\gamma_{23})\bigr)
\bigl(\sign \(\gamma_{34}+\lambda_3 \gamma_{1+2,4})\) + \sign(\cs_3)\bigr)
\nn
\eea
and is determined by the function \cite[(2.59)]{Manschot:2013sya}
\be
\begin{split}
G_4 = &\,
 \frac14\,\Bigl[\bigl(\sign (\gamma_{12}) + \sign (\gamma_{23})\bigr)
\bigl(\sign(\gamma_{34})-\sign(\gamma_{1+2+3,4})\bigr)
\\
&\,+\bigl(\sign (\gamma_{23}) + \sign (\gamma_{34})\bigr)
\bigl( \sign(\gamma_{23}+\gamma_{34}+\gamma_{24}) -\sign(\gamma_{23})\bigr)\Bigr]
\end{split}
\ee
evaluated at $\gamma_{ij}^{(1)}= \gamma_{ij}$ if $|i-j|<2$, and
$\gamma_{ij}^{(1)}= \lambda_1\gamma_{ij}$ if $|i-j|\geq 2$, where the parameters $\lambda_{1,2,3}$
are given by\footnote{The necessary condition that $\lambda_i\in [0,1]$ is ensured by the fact
that a rational number $-a/b$ lies in $]0,1[$ iff $a(a+b)<0$ (in which case $ab<0$ and $b(a+b)>0$) and
by non-vanishing of the sign factors in \eqref{FC4} multiplying the sign dependent on given $\lambda_i$.}
\be
\lambda_1=-\frac{\gamma_{12}+\gamma_{23}+\gamma_{34}}{\gamma_{14}+\gamma_{24}+\gamma_{13}},
\qquad
\lambda_2=-\frac{\gamma_{23}+\gamma_{34}}{\gamma_{24}},
\qquad
\lambda_3=-\frac{\gamma_{12}+\gamma_{23}}{\gamma_{13}}.
\ee
Substituting $G_4(\{\gamma_i^{(1)}\})$ into \eqref{FC4} and applying repeatedly the sign identity \eqref{signprop-ap},
one can show that
\bea
\FCn{4}&=&\frac18\,\Biggl[ \prod_{k=1}^3 \bigl(\sign(\cs_k)+ \sign (\gamma_{k,k+1})\bigr)
\nn\\
&&
+\sign (\cs_1)\Bigl(\bigl(\sign(\gamma_{23})+\sign(\gamma_{34})\bigr)\,\sign(\gamma_{23}+\gamma_{24}+\gamma_{34})
-\sign(\gamma_{23})\sign(\gamma_{34})-1\Bigr)
\nn\\
&&
+\sign (\cs_3)\Bigl(\bigl(\sign (\gamma_{12}) + \sign (\gamma_{23})\bigr)\,\sign(\gamma_{12}+\gamma_{23}+\gamma_{13})
-\sign(\gamma_{12})\sign(\gamma_{23})-1\Bigr)
\nn\\
&&
+\sign\(\sum_{i<j}\gamma_{ij}\)\Bigl(\bigl(\sign (\gamma_{12}) + \sign (\gamma_{23})\bigr)\,\sign(\gamma_{12}+\gamma_{23}+\gamma_{13})
\Bigr.
\label{FC4res}\\
&&\Bigl.\qquad
+\bigl(\sign(\gamma_{23})+\sign(\gamma_{34})\bigr)\,\sign(\gamma_{23}+\gamma_{24}+\gamma_{34})
+\sign(\gamma_{12})\sign(\gamma_{34})-1\Bigr)
\nn\\
&&
-\sign\(\gamma_{12}+\gamma_{23}+\gamma_{34}\)\Bigl(\sign(\gamma_{12})\sign(\gamma_{23})
+\sign(\gamma_{23})\sign(\gamma_{34})
+\sign(\gamma_{12})\sign(\gamma_{34})+1\Bigr)\Biggr].
\nn
\eea


\begin{thebibliography}{10}

\bibitem{Pioline:2011gf}
B.~Pioline, ``{Four ways across the wall},'' {\em J.Phys.Conf.Ser.} {\bf 346}
  (2012) 012017,
\href{http://www.arXiv.org/abs/1103.0261}{{\tt 1103.0261}}.

\bibitem{ks}
M.~Kontsevich and Y.~Soibelman, ``{Stability structures, motivic
  Donaldson-Thomas invariants and cluster transformations},''
  \href{http://www.arXiv.org/abs/0811.2435}{{\tt 0811.2435}}.

\bibitem{Joyce:2008pc}
D.~Joyce and Y.~Song, ``{A theory of generalized Donaldson-Thomas
  invariants},'' American Mathematical Soc., 2012, 
\href{http://www.arXiv.org/abs/0810.5645}{{\tt 0810.5645}}.

\bibitem{Joyce:2009xv}
D.~Joyce, ``{Generalized Donaldson-Thomas invariants},''
\href{http://www.arXiv.org/abs/0910.0105}{{\tt 0910.0105}}.

\bibitem{Denef:2007vg}
F.~Denef and G.~W. Moore, ``{Split states, entropy enigmas, holes and halos},''
  {\em JHEP} {\bf 1111} (2011) 129,
\href{http://www.arXiv.org/abs/hep-th/0702146}{{\tt hep-th/0702146}}.

\bibitem{Gaiotto:2008cd}
D.~Gaiotto, G.~W. Moore, and A.~Neitzke, ``{Four-dimensional wall-crossing via
  three-dimensional field theory},'' {\em Commun.Math.Phys.} {\bf 299} (2010)
  163--224, \href{http://www.arXiv.org/abs/0807.4723}{{\tt 0807.4723}}.

\bibitem{Andriyash:2010qv}
E.~Andriyash, F.~Denef, D.~L. Jafferis, and G.~W. Moore, ``{Wall-crossing from
  supersymmetric galaxies},'' {\em JHEP} {\bf 1201} (2012) 115,
\href{http://www.arXiv.org/abs/1008.0030}{{\tt 1008.0030}}.

\bibitem{Manschot:2010qz}
J.~Manschot, B.~Pioline, and A.~Sen, ``{Wall Crossing from Boltzmann Black Hole
  Halos},'' {\em JHEP} {\bf 1107} (2011) 059,
\href{http://www.arXiv.org/abs/1011.1258}{{\tt 1011.1258}}.

\bibitem{Diaconescu:2007bf}
E.~Diaconescu and G.~W. Moore, ``{Crossing the wall: Branes versus bundles},''
  {\em Adv. Theor. Math. Phys.} {\bf 14} (2010), no.~6, 1621--1650,
\href{http://www.arXiv.org/abs/0706.3193}{{\tt 0706.3193}}.

\bibitem{Dimofte:2009bv}
T.~Dimofte and S.~Gukov, ``{Refined, Motivic, and Quantum},'' {\em Lett. Math.
  Phys.} {\bf 91} (2010) 1,
\href{http://www.arXiv.org/abs/0904.1420}{{\tt 0904.1420}}.

\bibitem{Gaiotto:2010be}
D.~Gaiotto, G.~W. Moore, and A.~Neitzke, ``{Framed BPS States},'' {\em Adv.
  Theor. Math. Phys.} {\bf 17} (2013), no.~2, 241--397,
\href{http://www.arXiv.org/abs/1006.0146}{{\tt 1006.0146}}.

\bibitem{Denef:2000nb}
F.~Denef, ``{Supergravity flows and D-brane stability},'' {\em JHEP} {\bf 0008}
  (2000) 050, \href{http://www.arXiv.org/abs/hep-th/0005049}{{\tt
  hep-th/0005049}}.

\bibitem{Bates:2003vx}
B.~Bates and F.~Denef, ``{Exact solutions for supersymmetric stationary black
  hole composites},'' {\em JHEP} {\bf 11} (2011) 127,
\href{http://www.arXiv.org/abs/hep-th/0304094}{{\tt hep-th/0304094}}.

\bibitem{Ferrara:1995ih}
S.~Ferrara, R.~Kallosh, and A.~Strominger, ``{$N=2$} extremal black holes,''
  {\em Phys. Rev.} {\bf D52} (1995) 5412--5416,
\href{http://www.arXiv.org/abs/hep-th/9508072}{{\tt hep-th/9508072}}.

\bibitem{Denef:2001xn}
F.~Denef, B.~R. Greene, and M.~Raugas, ``{Split attractor flows and the
  spectrum of BPS D-branes on the quintic},'' {\em JHEP} {\bf 05} (2001) 012,
\href{http://www.arXiv.org/abs/hep-th/0101135}{{\tt hep-th/0101135}}.

\bibitem{Andriyash:2010yf}
E.~Andriyash, F.~Denef, D.~L. Jafferis, and G.~W. Moore, ``{Bound state
  transformation walls},'' {\em JHEP} {\bf 1203} (2012) 007,
\href{http://www.arXiv.org/abs/1008.3555}{{\tt 1008.3555}}.

\bibitem{Manschot:2010xp}
J.~Manschot, ``{Wall-crossing of D4-branes using flow trees},'' {\em
  Adv.Theor.Math.Phys.} {\bf 15} (2011) 1--42,
\href{http://www.arXiv.org/abs/1003.1570}{{\tt 1003.1570}}.

\bibitem{Manschot:2011xc}
J.~Manschot, B.~Pioline, and A.~Sen, ``{A Fixed point formula for the index of
  multi-centered N=2 black holes},'' {\em JHEP} {\bf 1105} (2011) 057,
\href{http://www.arXiv.org/abs/1103.1887}{{\tt 1103.1887}}.

\bibitem{Collinucci:2008ht}
A.~Collinucci and T.~Wyder, ``{The Elliptic genus from split flows and
  Donaldson-Thomas invariants},'' {\em JHEP} {\bf 05} (2010) 081,
\href{http://www.arXiv.org/abs/0810.4301}{{\tt 0810.4301}}.

\bibitem{Manschot:2013sya}
J.~Manschot, B.~Pioline, and A.~Sen, ``{On the Coulomb and Higgs branch
  formulae for multi-centered black holes and quiver invariants},'' {\em JHEP}
  {\bf 05} (2013) 166,
\href{http://www.arXiv.org/abs/1302.5498}{{\tt 1302.5498}}.

\bibitem{Chowdhury:2012jq}
A.~Chowdhury, S.~Lal, A.~Saha, and A.~Sen, ``{Black Hole Bound State
  Metamorphosis},'' {\em JHEP} {\bf 1305} (2013) 020,
\href{http://www.arXiv.org/abs/1210.4385}{{\tt 1210.4385}}.

\bibitem{Denef:2002ru}
F.~Denef, ``{Quantum quivers and Hall/hole halos},'' {\em JHEP} {\bf 10} (2002)
  023,
\href{http://www.arXiv.org/abs/hep-th/0206072}{{\tt hep-th/0206072}}.

\bibitem{Hori:2014tda}
K.~Hori, H.~Kim, and P.~Yi, ``{Witten Index and Wall Crossing},'' {\em JHEP}
  {\bf 01} (2015) 124,
\href{http://www.arXiv.org/abs/1407.2567}{{\tt 1407.2567}}.

\bibitem{Cordova:2014oxa}
C.~Cordova and S.-H. Shao, ``{An Index Formula for Supersymmetric Quantum
  Mechanics},''
\href{http://www.arXiv.org/abs/1406.7853}{{\tt 1406.7853}}.

\bibitem{Hwang:2014uwa}
C.~Hwang, J.~Kim, S.~Kim, and J.~Park, ``{General instanton counting and 5d
  SCFT},'' {\em JHEP} {\bf 07} (2015) 063,
  \href{http://www.arXiv.org/abs/1406.6793}{{\tt 1406.6793}}.
[Addendum: JHEP04,094(2016)].

\bibitem{MPSunpublished}
J.~Manschot, B.~Pioline, and A.~Sen, unpublished (2013).

\bibitem{Manschot:2012rx}
J.~Manschot, B.~Pioline, and A.~Sen, ``{From Black Holes to Quivers},'' {\em
  JHEP} {\bf 1211} (2012) 023,
\href{http://www.arXiv.org/abs/1207.2230}{{\tt 1207.2230}}.

\bibitem{Manschot:2014fua}
J.~Manschot, B.~Pioline, and A.~Sen, ``{The Coulomb Branch Formula for Quiver
  Moduli Spaces},'' {\em Confluentes Mathematici} {\bf 2} (2017) 49--69,
\href{http://www.arXiv.org/abs/1404.7154}{{\tt 1404.7154}}.

\bibitem{Bena:2012hf}
I.~Bena, M.~Berkooz, J.~de~Boer, S.~El-Showk, and D.~Van~den Bleeken,
  ``{Scaling BPS Solutions and pure-Higgs States},'' {\em JHEP} {\bf 1211}
  (2012) 171,
\href{http://www.arXiv.org/abs/1205.5023}{{\tt 1205.5023}}.

\bibitem{Sen:2008yk}
A.~Sen, ``{Entropy Function and AdS(2) / CFT(1) Correspondence},'' {\em JHEP}
  {\bf 0811} (2008) 075,
\href{http://www.arXiv.org/abs/0805.0095}{{\tt 0805.0095}}.

\bibitem{Sen:2009vz}
A.~Sen, ``{Arithmetic of Quantum Entropy Function},'' {\em JHEP} {\bf 08}
  (2009) 068,
\href{http://www.arXiv.org/abs/0903.1477}{{\tt 0903.1477}}.

\bibitem{Dabholkar:2010rm}
A.~Dabholkar, J.~Gomes, S.~Murthy, and A.~Sen, ``{Supersymmetric Index from
  Black Hole Entropy},'' {\em JHEP} {\bf 1104} (2011) 034,
\href{http://www.arXiv.org/abs/1009.3226}{{\tt 1009.3226}}.

\bibitem{Lee:2012sc}
S.-J. Lee, Z.-L. Wang, and P.~Yi, ``{Quiver Invariants from Intrinsic Higgs
  States},'' {\em JHEP} {\bf 1207} (2012) 169,
\href{http://www.arXiv.org/abs/1205.6511}{{\tt 1205.6511}}.

\bibitem{Lee:2012naa}
S.-J. Lee, Z.-L. Wang, and P.~Yi, ``{BPS States, Refined Indices, and Quiver
  Invariants},'' {\em JHEP} {\bf 1210} (2012) 094,
\href{http://www.arXiv.org/abs/1207.0821}{{\tt 1207.0821}}.

\bibitem{Maldacena:1997de}
J.~M. Maldacena, A.~Strominger, and E.~Witten, ``{B}lack hole entropy in
  {M}-theory,'' {\em JHEP} {\bf 12} (1997) 002,
\href{http://www.arXiv.org/abs/hep-th/9711053}{{\tt hep-th/9711053}}.

\bibitem{Manschot:2009ia}
J.~Manschot, ``{Stability and duality in N=2 supergravity},'' {\em
  Commun.Math.Phys.} {\bf 299} (2010) 651--676,
  \href{http://www.arXiv.org/abs/0906.1767}{{\tt 0906.1767}}.

\bibitem{Andriyash:2008it}
E.~Andriyash and G.~W. Moore, ``{Ample D4-D2-D0 Decay},''
\href{http://www.arXiv.org/abs/0806.4960}{{\tt 0806.4960}}.

\bibitem{deBoer:2008fk}
J.~de~Boer, F.~Denef, S.~El-Showk, I.~Messamah, and D.~Van~den Bleeken,
  ``{Black hole bound states in $AdS_3 \times S^2$},'' {\em JHEP} {\bf 0811}
  (2008) 050,
\href{http://www.arXiv.org/abs/0802.2257}{{\tt 0802.2257}}.

\bibitem{Alexandrov:2012au}
S.~Alexandrov, J.~Manschot, and B.~Pioline, ``{D3-instantons, Mock Theta Series
  and Twistors},'' {\em JHEP} {\bf 1304} (2013) 002,
\href{http://www.arXiv.org/abs/1207.1109}{{\tt 1207.1109}}.

\bibitem{Alexandrov:2016tnf}
S.~Alexandrov, S.~Banerjee, J.~Manschot, and B.~Pioline, ``{Multiple
  D3-instantons and mock modular forms I},'' {\em Commun. Math. Phys.} {\bf
  353} (2017), no.~1, 379--411,
\href{http://www.arXiv.org/abs/1605.05945}{{\tt 1605.05945}}.

\bibitem{ap-to-appear}
S.~Alexandrov and B.~Pioline, ``{Black holes, instantons and mock modular
  forms}'', \href{http://www.arXiv.org/abs/1808.08479}{{\tt 1808.08479}}, 
  to appear in {\em Commun. Math. Phys.} (2019).

\bibitem{deBoer:2008zn}
J.~de~Boer, S.~El-Showk, I.~Messamah, and D.~Van~den Bleeken, ``{Quantizing N=2
  Multicenter Solutions},'' {\em JHEP} {\bf 05} (2009) 002,
\href{http://www.arXiv.org/abs/0807.4556}{{\tt 0807.4556}}.

\bibitem{Kim:2011sc}
H.~Kim, J.~Park, Z.~Wang, and P.~Yi, ``{Ab Initio Wall-Crossing},'' {\em JHEP}
  {\bf 1109} (2011) 079,
\href{http://www.arXiv.org/abs/1107.0723}{{\tt 1107.0723}}.

\bibitem{derksen2005quiver}
H.~Derksen and J.~Weyman, ``Quiver representations,'' {\em Notices of the AMS}
  {\bf 52} (2005), no.~2, 200--206.

\bibitem{reineke2008moduli}
M.~Reineke, ``Moduli of representations of quivers,'' 
Proceedings of ICRA XII, Toru\'n, Poland, August 15-24, 2007, 
\href{http://www.arXiv.org/abs/0802.2147}{{\tt 0802.2147}}.

\bibitem{1043.17010}
M.~Reineke, ``{The Harder-Narasimhan system in quantum groups and cohomology of
  quiver moduli.},'' {\em Invent. Math.} {\bf 152} (2003), no.~2, 349--368.

\bibitem{Alexandrov:2016enp}
S.~Alexandrov, S.~Banerjee, J.~Manschot, and B.~Pioline, ``{Indefinite theta
  series and generalized error functions},''
{\em  Selecta Mathematica}  {\bf 24} (2018) 3927--3972,
\href{http://www.arXiv.org/abs/1606.05495}{{\tt 1606.05495}}.

\bibitem{Alexandrov:2017qhn}
S.~Alexandrov, S.~Banerjee, J.~Manschot, and B.~Pioline, ``{Multiple
  D3-instantons and mock modular forms II},'' {\em Commun. Math. Phys.} {\bf
  359} (2018), no.~1, 297--346,
\href{http://www.arXiv.org/abs/1702.05497}{{\tt 1702.05497}}.

\end{thebibliography}

\providecommand{\href}[2]{#2}\begingroup\raggedright\endgroup

\end{document}